\documentclass{article}
\usepackage{graphicx}
\usepackage{float}
\usepackage{topcapt}

\usepackage{lipsum}
\usepackage{epigraph}

\usepackage{amsmath}
\usepackage{amssymb}
\usepackage{amsthm}
\usepackage{url}
\usepackage{eurosym}

\usepackage[T2A]{fontenc}
\usepackage[utf8]{inputenc}
\usepackage[russian,english,british]{babel}

\usepackage{longtable}
\usepackage{fancyvrb}
\usepackage{tocbibind}
\usepackage{setspace}
\usepackage{epigraph}
\usepackage{etoolbox}

\newtheorem{theorem}{Theorem}[section]
\newtheorem{lemma}{Lemma}[section]

\begin{document}

\title{Trading Strategies with Position Limits}
\author{Valerii Salov}
\date{}
\maketitle

\begin{abstract}
Whether you trade futures for yourself or a hedge fund, your strategy is counted. Long and short position limits make the number of unique strategies finite. Formulas of the numbers of strategies, transactions, do nothing actions are derived. A discrete distribution of actions, corresponding probability mass, cumulative distribution and characteristic functions, moments, extreme values are presented. Strategies time slice distributions are determined. Vector properties of trading strategies are studied. Algebraic not associative, commutative, initial magmas with invertible elements control trading positions and strategies. Maximum profit strategies, MPS, and optimal trading elements can define trading patterns. Dynkin introduced the term interpreted in English as "Markov time" in 1963. Neftci applied it for the formalization of Technical Analysis in 1991.
\end{abstract}

\section{Introduction}

\setlength{\epigraphwidth}{0.95\textwidth}
\begin{epigraphs}

\qitem{Technical analysis is anathema to  the academic world.}
{\textsc{Burton Malkiel, \cite[p. 127]{malkiel2007}}}

\qitem{... technical analysis is a broad class of prediction rules with unknown statistical properties, developed by practitioners without reference to any formalism.}
{\textsc{Salih Neftci, \cite[p. 549]{neftci1991}}}

\qitem{... considerable part of time the particle spends in one semi-plane. These paradoxical regularities of particle transition from positive side of line to negative and vice versa are covered by the theorem called the "arcsine law".}
{\textsc{Andrey Kolmogorov, Igor Zhurbenko, Alexander Prokhorov, \cite[pp. 91 - 92 in Symmetrical Random Walk]{kolmogorov1982}}}

\end{epigraphs}

Modern theories of prices exploit mathematics of stochastic processes \cite{neftci1996} such as a \textit{Brownian motion}. \textit{"Brownian paths are wilder than we can imaging} \cite[p. 19]{rogers2000}. They are continuous functions of time \cite[p. 10, Theorem 6.1]{rogers2000}, non-differentiable almost everywhere \cite[p. 19 and Theorem 10.1]{rogers2000}. Theoretical continuity implies that the frequency of observations can be increased arbitrary. In \cite[p. 74 and p. 81]{goodhart1997}, we read \textit{"In some markets, second-by-second data is now available, allowing virtually continuous observations of price ..."} and \textit{"A fundamental property of high frequency data is that observations can occur at varying time intervals"}. The sentences remind the author about school experiments determining the \textit{gravitational acceleration} $g \approx 9.8 \frac{\mathrm{meter}}{\mathrm{second}^2}$. In a dark room, a steel ball is dropped from three meters along the inverted vertical ruler. A stroboscope flashes four times per second. The ball is visible at centimeter labels 31, 123, 276, on the floor. \textit{While invisible, the ball is still moving.}

Engle, estimating the growth of the frequency of observations, says \textit{"The limit in nearly all cases, is achieved when all transactions are recorded"} \cite[p. 2]{engle2000}. Time \& Sales data distributed by the Chicago Mercantile Exchange, CME, Group \url{http://www.cmegroup.com/} for the futures contracts electronically traded on the Globex platform is an example. There are no other ticks between two neighbors. Each associates time, price, number of traded contracts, and market conditions symbol. The irregular time intervals between neighboring transactions, waiting times or durations, and corresponding price increments, named by the author a- and b-increments, are studied \cite{salov2013}, \cite{salov2017}. Algebraic sums of these elementary, indecomposable further increments, \textit{financial atoms}, represent all popular minute, hourly, daily increments and chart price bars.

The prices are discrete and occur at irregular times, where financial instruments with high leverage and large trading positions magnify discreteness. Modern finance attempts \textit{to touch the left ear by the right hand}. It brings to discrete markets continuous models to create later sophisticated discretization schemata, for example, for Monte Caro simulations \cite{glasserman2003}. Kolmogorov foresaw: \textit{"It is quite probable that with the development of modern computing techniques, it will become understood that in many cases it is reasonable to study real phenomena without making use of intermediate step of their stylization in the form of infinite and continuous mathematics, passing directly to discrete models"} \cite{kolmogorov1983}.

Successes in the numerical integration of Brownian motions, diffusions \cite{milstein1995}, \cite{kloeden1999}, \cite{protter2004}, jump-diffusions \cite{hanson2007} applying higher order integration formulas are accompanied by new evidences of non-Gaussian properties of futures price increments and micro-structure of jumps resembling non-equilibrium properties of explosions \cite[p. 41 "Non-Gaussian atoms", pp. 41 - 43 A Comment on Disequilibrium]{salov2013}, \cite[pp. 30 - 36 Randomness of Price Increments, pp. 37 - 40 Non-Gaussian Relative b-Increments, pp. 50 - 52 Jumps. Chain Reactions]{salov2017}.

Ignorance of intra-day ticks yields \textit{controversial} daily price applications. Indeed, absolute daily price increments $P_1-P_0$, $P_2 - P_1$ or logarithms of price ratios $\ln(\frac{P_1}{P_0})$, $\ln(\frac{P_2}{P_1})$ for three days 0, 1, 2 are sums of significantly different numbers of summands $N_1$ and $N_2$ in day trading sessions 1 and 2, intra-day price increments or logarithms of price ratios. Even, if the latter are independent and identically distributed, i.i.d., random variables, the variance of the sum $\mu_2^S$, being under such conditions the sum of the variances, $N_1\mu_2$ and $N_2\mu_2$, \cite{gnedenko1949}, makes the sums non-i.i.d. random variables because $N_1 \ne N_2$. This compromises assumption on i.i.d. for daily increments or returns. What a mess! \textit{The paper considers prices, transaction costs, trading strategies, positions taken with traded assets as discrete entities and illustrates results using futures contracts.}

\section{A Portrait of Futures}

Leo Melamed \cite{melamed2014}: \textit{"According to the Bank of International Settlements (BIS), 81.3\% of all futures traded in 2013 were financial futures and options. The notional value of those traded equaled an astounding \$1,886,283.4 billion"}.

From several futures contracts on the same commodity, index, or security, the contract with the closest settlement date is called the \textit{nearby futures contract}. Accordingly, there are \textit{next}, \textit{distant}, and \textit{the most distant} futures contracts. While a daily price chain for an individual contract is \textit{weird}, to escape arbitrage, prices of the mentioned contracts move in sync, Figure \ref{FigESM17_ESU17_Daily_Prices}. The daily prices of the  neighbors, M June and U September, differ $\frac{P_{i}^{\textrm{ESM17}}}{P_{i}^{\textrm{ESU17}}} \ne 1$, Figure \ref{FigESM17_ESU17_Daily_Price_Ratio}.
\begin{figure}[!h]
  \centering
  \includegraphics[width=120mm]{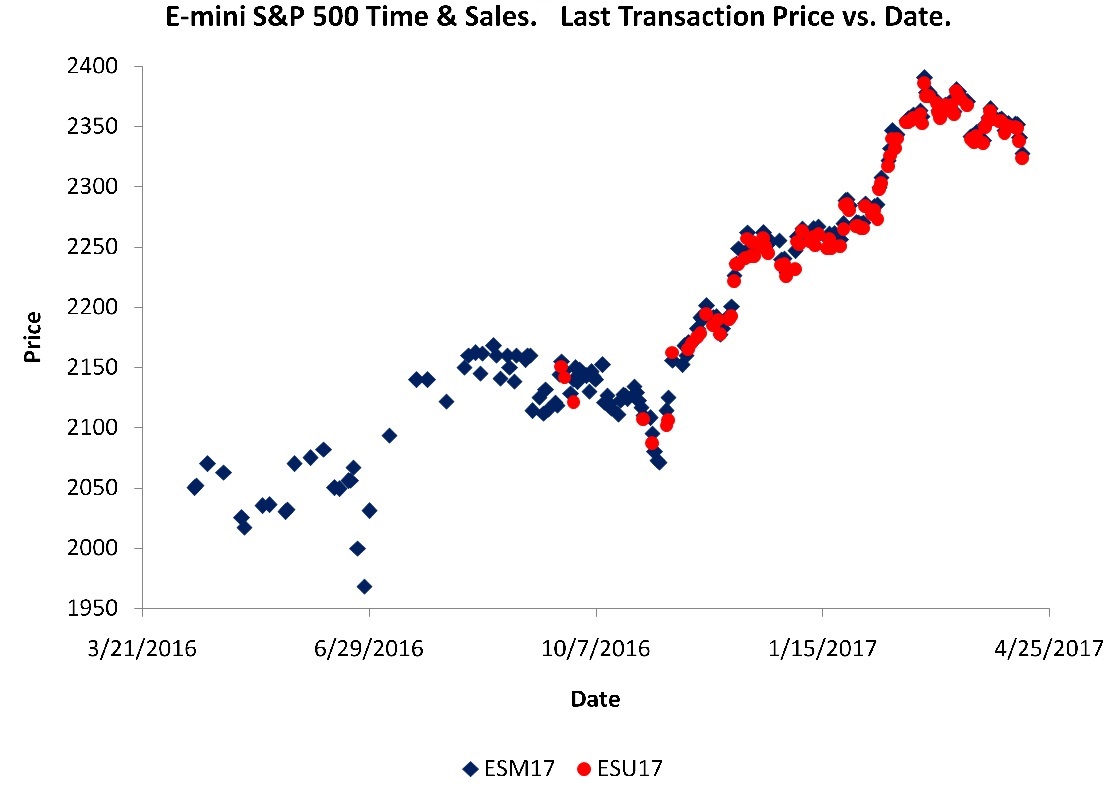}
  \caption[FigESM17_ESU17_Daily_Prices]
   {Time \& Sales Globex, \url{http://www.cmegroup.com/}, last transaction prices of ESM17 (\textit{nearby} in April 2017), and ESU17 (\textit{next} in April 2017) from session ranges finishing at 15:15:00, Central Standard Time, CST. 185 ESM17 and 103 ESU17 sessions, Wednesday April 13, 2016 - Thursday April 13, 2017. Plotted using Microsoft Excel.}
  \label{FigESM17_ESU17_Daily_Prices}
\end{figure}
\begin{figure}[!h]
  \centering
  \includegraphics[width=95mm]{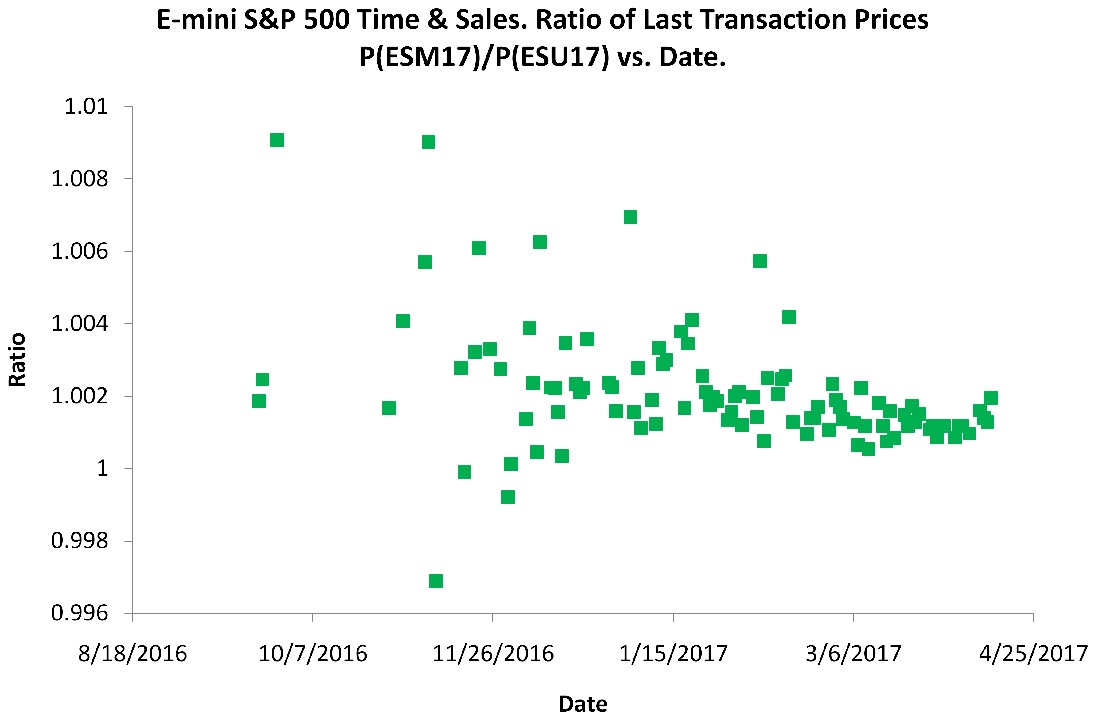}
  \caption[FigESM17_ESU17_Daily_Price_Ratio]
   {Time \& Sales Globex, \url{http://www.cmegroup.com/}, ratios of last transaction prices of ESM17, and ESU17 from session ranges finishing at 15:15:00, CST. Number of sessions $N=103$. Plotted using Microsoft Excel.}
  \label{FigESM17_ESU17_Daily_Price_Ratio}
\end{figure}
\begin{figure}[!h]
  \centering
  \includegraphics[width=95mm]{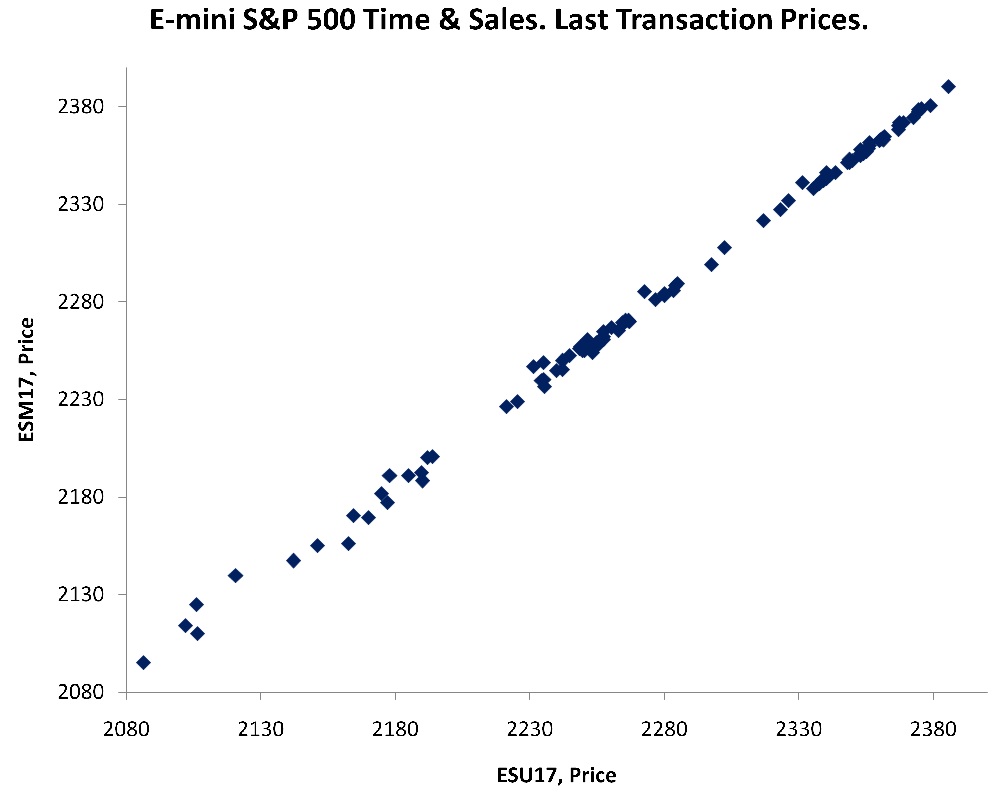}
  \caption[FigESM17_ESU17_Daily_Price_Regression]
   {Time \& Sales Globex, \url{http://www.cmegroup.com/}, regression of last transaction prices of ESM17, and ESU17 from session ranges finishing at 15:15:00, CST: $P_{\textrm{ESM17}}=(1.0021 \pm 0.0002)P_{\textrm{ESU17}}$; coefficient of correlation = 0.999998; intercept is set to zero; confidence two sided probability is 95\%; $N = 103$. Computed using Microsoft Excel, Data Analysis, Regression.}
  \label{FigESM17_ESU17_Daily_Price_Regression}
\end{figure}

$P_i^{\textrm{ESU17}}$ is in [minimum 2086.50, Tuesday November 1, 2016; maximum 2385.75, Wednesday March 1, 2017]. The price range counted from 2086.50 is less than 15\%. Therefore, after a zoom in, the look and feel of $P_i^{\textrm{ESM17}} - P_i^{\textrm{ESU17}}$ is similar to the relative price increment $\frac{P_i^{\textrm{ESM17}} - P_i^{\textrm{ESU17}}}{P_i^{\textrm{ESU17}}} = \frac{P_i^{\textrm{ESM17}}}{P_i^{\textrm{ESU17}}} - 1$, a shift down of Figure \ref{FigESM17_ESU17_Daily_Price_Ratio}. The latter, for small $|P_i^{\textrm{ESM17}} - P_i^{\textrm{ESU17}}|$, is close to $\ln(\frac{P_i^{\textrm{ESM17}}}{P_i^{\textrm{ESU17}}})$. Under observed conditions, \textit{one} Figure \ref{FigESM17_ESU17_Daily_Price_Ratio} gives an idea about \textit{four} quantities. Since $\frac{P_i^{\textrm{ESM17}}}{P_i^{\textrm{ESU17}}} \approx 1$, Figure \ref{FigESM17_ESU17_Daily_Price_Regression} is not a surprise. The empirical regression with 103 points and coefficient correlation almost equal to one expresses what we understand under "moving in sync".

Figure \ref{FigESM17_ESU17_Daily_Price_Regression} is possible because $P_i^{\textrm{ESM17}}$ and $P_i^{\textrm{ESU17}}$ are linked by dates $i$. On the left side of Figure \ref{FigESM17_ESU17_Daily_Prices}, $P_i^{\textrm{ESU17}}$ circles are missing: the contract was not yet traded. Later, dots are missed due to low liquidity or lost data. For 185 sessions and prices of ESM17, only 103 "corresponding" points of ESU17 are collected. This is enough to conclude about the almost linear dependence.

By eye, tick prices on Monday April 10, 2017 are in sync, Figure \ref{FigESM17_ESU17_ESZ17_ESH18_Price_20170410}. Transaction volumes are too, Figure \ref{FigESM17_ESU17_ESZ17_ESH18_Volume_20170410}. Ticks \{date-time price size\} arrive at random times \cite{salov2013}, \cite{salov2017}. The depicted discrete properties of the S\&P E-Mini futures with a single tick for ESH18 on April 10, 2017 is a guide behind mathematical formulas. \textit{The paper is about the futures trading strategies and maximum profit strategies, which can define patterns.}
\begin{figure}[!h]
  \centering
  \includegraphics[width=135mm]{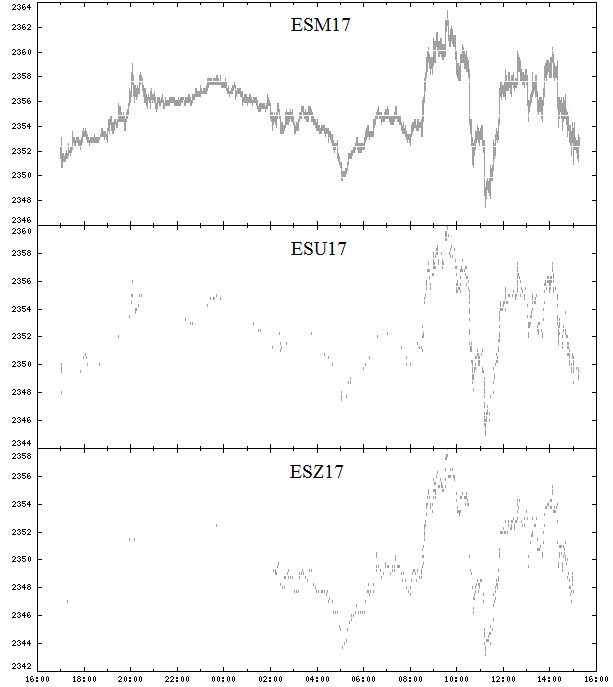}
  \caption[FigESM17_ESU17_ESZ17_ESH18_Price_20170410]
   {E-mini S\&P 500 Futures Time \& Sales Globex, \url{http://www.cmegroup.com/}, transaction prices of ESM17, ESU17, ESZ17, and ESH18 for the time range [Sunday April 9, 2017, 17:00:00 - Monday April 10, 2017, 15:15:00], CST. The most distant contract ESH18 had only one tick \{date time price size\} = \{2017/04/10 11:18:21 2342 1\}. Plotted using custom C++  and Python programs and gnuplot \url{http://www.gnuplot.info/}.}
  \label{FigESM17_ESU17_ESZ17_ESH18_Price_20170410}
\end{figure}
\begin{figure}[!h]
  \centering
  \includegraphics[width=135mm]{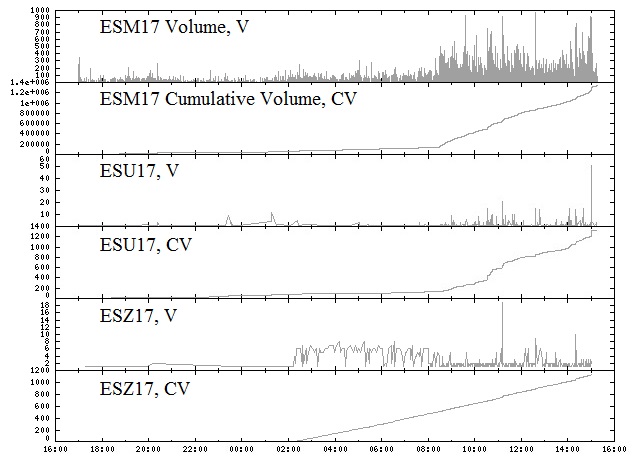}
  \caption[FigESM17_ESU17_ESZ17_ESH18_Volume_20170410]
   {E-mini S\&P 500 Futures Time \& Sales Globex, \url{http://www.cmegroup.com/}, transaction volumes and cumulative volumes of ESM17, ESU17, ESZ17, and ESH18 for the time range [Sunday April 9, 2017, 17:00:00 - Monday April 10, 2017, 15:15:00], CST. ESH18 had one tick \{date time price size\} = \{2017/04/10 11:18:21 2342 1\}. Plotted using custom C++  and Python programs and gnuplot \url{http://www.gnuplot.info/}.}
  \label{FigESM17_ESU17_ESZ17_ESH18_Volume_20170410}
\end{figure}

\section{Combinatorial Properties of Trading Strategies}

Consider the chains of positive prices $P_1, \dots, P_i, \dots, P_n$,  non-negative transaction costs $C_1, \dots, C_i, \dots, C_n$, subtracted from profit and losses, $PL$, making the latter \textit{not better}, and integer positive \textit{buy}, negative \textit{sell}, and zero \textit{do nothing actions} $U_1, \dots, U_i, \dots, U_n$, representing the numbers of traded units. Positive and negative actions are \textit{transactions}. The chain of actions is a \textit{trading strategy}.

Time \& Sales \textit{ticks} arrive as triplets $\{t_i, P_i, V_i\}$, where $V$ is the volume or size of the \textit{trade}. The $i$ associates with the order of arrival. $V=0$ represents \textit{indicative prices} or special market conditions rather than trades. The numbers of bought and sold units in a trade are equal and $V$ shows one side. The trade is a combination of transactions made from the same or different accounts after matching prices of the \textit{limit orders} accepted in different times by a \textit{trading book}. For trading one instrument, yielding Time \& Sales ticks, we assume that transactions are made from one account. The pair of opposite  transactions with zero sum is a \textit{round-trip trade}. The trades associate with adjusting, overlapping, or disjoint time intervals. A \textit{net action} with transactions and zero sum can be broken down on round-trip trades in several ways. Accounting may regroup transactions maximizing realized profits and leaving losses to open marked-to-market positions. With a good luck, the final list will contain only profits. If the good luck turns away, then the natural time order of transactions can yield smaller losses or a mixture. \textit{All variations must have the same total $PL$.}

Actions affect \textit{trading positions} - the numbers of securities that are owned or borrowed and then sold. The \textit{long}, \textit{short}, or \textit{out of market} positions are positive, negative, or zero integers forming the chain $W_1,\dots,W_i,\dots,W_n$, where $W_i=W_0+\sum_{j=1}^{i}U_i$. $W_0$ helps to express $U_i=W_i-W_{i-1}$ for $i=1,\dots, n$. $W_i$ will show the position always after the action $U_i$. The \textit{marked to market} $PL$ of a trading strategy is equal to \cite[Equation 54, page 82]{salov2013}:

\begin{equation}
\label{EqPL}
PL=k\left(P_n\sum_{i=1}^{n}U_i-\sum_{i=1}^{n}P_iU_i\right)-\sum_{i=1}^{n}C_i|U_i|-C_n|\sum_{i=1}^{n}U_i|,
\end{equation}
where $k$ is the dollar value of a \textit{full price point}. For instance, for the S\&P 500 E-Mini futures contract the price value of the full point is 50 US dollars. The absolute minimum non-zero price fluctuation is $\delta = 0.25$ or 12.5 US dollars. The chain of three prices could be $P=(2369.50, 2369.75, 2370.00)$. Due to the current S\&P 500 E-Mini contract specifications, the prices between these levels are impossible. The $k$ allows to apply conventional market price quotes and get dollar equivalents. In contrast, transactions costs $C_i$ are expressed in dollars per unit per transaction. The strategy $U=(1, 0, -1)$ with the cost $C=(5,5,5)$ yields $PL=50 \times (2370.00 \times (1 + 0 + (-1)) - (2369.50 \times 1+ 2369.75 \times 0 + 2370.00 \times (-1)) - (5 \times |1| + 5 \times |0| + 5 \times |-1|) - 5 \times |1 + 0 + (-1)| = 15$ US dollars. In futures, costs, fixed per contract per transaction, are common. Brokers may reduce them to promote large transactions. In equities, commissions can be a fixed fraction of price. Then, to continue with the formula, $C$ can be evaluated on a preliminary step using the commissions discount or $P$ and fraction.

The chains of $n$ elements are \textit{column-vectors} in $n$-dimensional spaces. \textit{Bold symbols denote vectors.} $\sum_{i=1}^{n}P_i U_i$ is the scalar product of $\boldsymbol{P}$ and $\boldsymbol{U}$, $W_n=\sum_{i=1}^{n}U_i$ for $W_0=0$: $PL=k(P_nW_n - \boldsymbol{P}^T\boldsymbol{U})-\boldsymbol{C}^T\textrm{abs}(\boldsymbol{U})-C_n|W_n|$, where $\mathrm{abs}()$ returns a vector with the absolute values of coordinates. $T$ means \textit{transpose}. For programming these equations, the Standard C++ classes \verb|std::vector|, \verb|std::array| \cite[pp. 902 - 906, 974 - 977]{stroustrup2013}, \cite[pp. 897 - 897, 871 - 874]{iso2017}, and algorithms \verb|std::inner_product|, \verb|std::transform|, \verb|std::accumulate| \cite[pp. 1178 - 1179, 935 - 936, 1177 - 1178]{stroustrup2013}, \cite[pp. 1131, 1023, 1130]{iso2017} are useful.

\begin{theorem}
\label{TMStrategiesPositionsBijection}
There is one and only one strategy $\boldsymbol{U}$ for position $\boldsymbol{W}$ and $W_0$.
\end{theorem}

\begin{proof}
For $\boldsymbol{W}$ and $W_0$, there is unique $\boldsymbol{U}$ since $U_i=W_i-W_{i-1}$. For  $\boldsymbol{U}$, there is unique $\boldsymbol{W}$ with $W_0$ since $W_i=W_0+\sum_{j=1}^{i}U_j$.
\end{proof}

Given $\boldsymbol{W}$ and $W_0$, $\boldsymbol{U}$ can be computed using the Standard C++ algorithm \verb|std::adjacent_difference|, while applying  \verb|std::partial_sum| \cite[pp. 1179 - 1180]{stroustrup2013}, \cite[pp. 1133, 1137 - 1138]{iso2017} recovers $\boldsymbol{W}$ from $\boldsymbol{U}$ and $W_0$. The number of vectors with $n$ integer coordinates is infinite but \textit{positions and actions are limited}. \textit{A position limit is a natural number $W \in \mathbb{N}$.}

\paragraph{Position limit $W$.} With the futures account $A=\$10,000$ and initial margin $IM_{ESZ17}=\$1,237.50$ covering intraday trading of the single futures December 2017 E-mini Standard and Poor's 500 Stock Price Index contract ESZ17, one can buy or sell $\lfloor \frac{10,000}{1,237.50} \rfloor = \lfloor 8.08\dots \rfloor = 8$ contracts. For a retail trader, constant fees and commissions per contract per one side, buy or sell entering the position $W_1=8$ or $-8$, can be $C=\$4.68$. Due to the costs, right after the transaction, the equity drops in our hypothetical example to $A - |W_1|*C = \$10,000 - 8 * \$4.68 = \$9,962.56$. The maintenance margin $MM_{ESZ17}=\$1,125$, typically smaller than $IM$, requires $|W_1| \times MM_{ESZ17} = 8 \times \$1,125 = \$9,000$. If prices move \text{favorably}, then the equity increases and eventually position is allowed to add contracts. If losses reduce the equity below the total maintenance margin requirement, then funds must be added to the account, or the position or its part is mandatory liquidated. At closing the position, the remaining cost is $8 \times \$4.68 = \$37.44$. The difference between the equity after the anticipated closing cost and maintenance margin equity is $\$9,962.56 - 8 \times \$4.68 - \$9,000 = \$925.12$. This is $\frac{\$925.12}{8} = \$115.64$ per contract. Conversion to price points yields $\frac{\$115.64}{k = \$50} < 2.5$. Depending on market conditions, the ESZ17 price can move 2.5 points up or down in a matter of a few seconds or minutes \cite{salov2013}, \cite{salov2017}. Establishing a position size up to the available account equity is too risky and can quickly ruin an account.
 \textit{The capital limit and margins dictate the position limit $W$. However, due to these factors only, $W$ does not have to be constant.}

One can voluntary trade small limited $|W_i| \le W$ or fixed $|W_i| = W$ positions. Depending on \textit{trading rules}, \textit{statistics of PL}, \textit{market conditions}, trading fixed positions can be too inefficient or risky \cite{kelly1956}, \cite{vince1992}, \cite{vince1995}, \cite{jones1999}, \cite[pp. 55- 79]{salov2007}. Still, such strategies can be useful for the evaluation of trading rules generating individual trading signals and their separation from the money management  answering which portion of the capital should be devoted to a next trade.

Without consideration of capital limits, position limits are known for the CME futures from the contract specifications, \url{http://www.cmegroup.com/}. For the E-mini futures the \textit{all month limit} is 60000 contracts. For Corn, the \textit{initial spot-month limit} is 600 or 3000000 bushels and the \textit{single month limit} is 33000 or 165000000 bushels. Bushel, a volume measure, is not from the International System of Units. It is eight US dry gallons. For corn with 15.5\% of moisture this is $\approx 25.4$ kg and the limit is equivalent of 4191000 metric tons. The National Agricultural Statistics Service of the United States Department of Agriculture reported for 2016: \textit{"U.S. corn growers produced 15.1 billion bushels, up 11 percent from 2015"}, \url{https://www.nass.usda.gov/Newsroom/2017/01_12_2017.php}. This is 383.5 million metric tons: the single month limit is $\approx1$\% of the U.S. 2016 corn crop. The "romance", when legendary Jesse Livermore could "corner the U.S. wheat market" \cite{lefevre1923}, \cite{livermore1940}, is in the past. 5000 bushels of one contract, $\approx 127$ tons, fit two \textit{hopper wagons}. Details of corn futures ticks are in \cite{salov2017}.

\begin{theorem}
\label{TMNumberOfStrategies}
There are $S=(2W+1)^{n-1}$ unique positions $\boldsymbol{W}_j$ and strategies $\boldsymbol{U}_j$ for $|W_{i,j}| \le W$, $i=1,\dots,n$ ticks, $W_{0,j}=W_{n,j}=0$.
\end{theorem}

\begin{proof}
The $n$th coordinate in $\boldsymbol{W_j}$ is zero. The remaining $n-1$ coordinates are $2W+1$ independent $-W, \dots , -1, 0, 1, \dots, W$. Thus, the number of unique combinations and $\boldsymbol{W_j}$ is $S=(2W+1)^{n-1}$. By Theorem \ref{TMStrategiesPositionsBijection}, the number of corresponding unique strategies $\boldsymbol{U_j}$ is the same.
\end{proof}

The sets of positions $\boldsymbol{W}_j$ with $W_{0,j}=W_{n,j}=0$, $|W_{i,j}| \le W$, $i \in [1,n]$, $j \in [1,S=(2W+1)^{n-1}]$ and corresponding strategies $\boldsymbol{U}_j$ are $\mathfrak{W}$ and $\mathfrak{U}$.

Example: $W=1$, $n=3$ yield nine pairs $\boldsymbol{W_j}^T \leftrightarrow \boldsymbol{U_j}^T$: $(0, 0, 0)\leftrightarrow(0, 0, 0)$ do nothing; $(1,0,0)\leftrightarrow(1,-1,0)$, $(-1,0,0)\leftrightarrow(-1,1,0)$; $(1,1,0)\leftrightarrow(1,0,-1)$, $(-1,-1,0)\leftrightarrow(-1,0,1)$; $(0,1,0)\leftrightarrow(0,1,-1)$, $(0,-1,0)\leftrightarrow(0,-1,1)$; $(1,-1,0)\leftrightarrow(1,-2,1)$, $(-1,1,0)\leftrightarrow(-1,2,-1)$.

\begin{theorem}
\label{TMSampleMeanPL}
The sample mean $PL$ of the strategies $\mathfrak{U}$ does not depend on price $\boldsymbol{P}$ and equal to $a_1^{PL}=\frac{-\sum_{j=1}^{(2W+1)^{n-1}} \boldsymbol{C}^T abs(\boldsymbol{U_j})}{(2W+1)^{n-1}}$.
\end{theorem}

\begin{proof}
The number of strategies given by Theorem \ref{TMNumberOfStrategies} is odd $(2W+1)^{n-1}$. The single \textit{do nothing strategy}, d.n.s, has $PL=0$. The remaining even number forms two sets with $\frac{(2W+1)^{n-1} - 1}{2}$ strategies each: not do nothing strategies $\boldsymbol{U}_j$ and their \textit{inverses} $-\boldsymbol{U}_j$. From Equation \ref{EqPL}, $PL(\boldsymbol{U}_j)+PL(-\boldsymbol{U}_j)=-2\boldsymbol{C}^T\textrm{abs}(\boldsymbol{U}_j)+C_n|W_n|)=-2\boldsymbol{C}^T\textrm{abs}(\boldsymbol{U}_j)$. Averaging gives $a_1^{PL}$.
\end{proof}

 From the market prospective, if $\boldsymbol{U}$ has been applied, no matter by whom, then $-\boldsymbol{U}$ has been applied too and $PL(\boldsymbol{U}) \ne -PL(-\boldsymbol{U})$. The sum outcome negative for traders is a payment to the industry. This expresses the known property of a \textit{negative non-zero sum game}.

\begin{theorem}
\label{TMNumberOfPositionsActions}
There are $n(2W+1)^{n-1}$ positions and actions in $\mathfrak{W}$ and $\mathfrak{U}$.
\end{theorem}

\begin{proof}
By Theorem \ref{TMNumberOfStrategies}, the numbers of positions and strategies are $(2W+1)^{n-1}$. Each has $n$ coordinates.
\end{proof}

The unique positions and strategies are indexed by $j \in [1,(2W+1)^{n-1}]$, yielding $n(2W+1)^{n-1}$ $W_{i,j}$ and $U_{i,j}$, $i \in [1,n]$. $\forall j$, $W_{0,j}=W_{n,j}=0$, $U_{1,j}=W_{1,j}-W_{0,j}=W_{1,j}$, $U_{n,j}=W_{n,j}-W_{n-1,j}=-W_{n-1,j}$. $\forall i \in [1,n-1]$, numbers of $W_{i,j}=W_l\in[-W,W]$, $l\in[1,2W+1]$ are equal. Uniformness follows from independence of positions, Theorem \ref{TMNumberOfStrategies}. Then, the numbers of $U_{1,j}=W_{1,j}$ and $U_{n,j}=-W_{n-1,j}$ of each kind $W_l$ in all strategies are $\frac{(2W+1)^{n-1}}{2W+1}=(2W+1)^{n-2}$.

Positions are \textit{state functions}. Actions are \textit{transition functions}. The do nothing action $U_{i,j}=0$ does not change the state $W_{i-1,j} \rightarrow W_{i,j}$. It is neither a loss nor a profit for a trader and does not pay to the industry.

\begin{theorem}
\label{TMNumberOfDNActions}
There are $n(2W+1)^{n-2}$ do nothing actions in $\mathfrak{U}$.
\end{theorem}

\begin{proof}
$\#(U_{1,j}=0)+\#(U_{n,j}=0)=2(2W+1)^{n-2}$ and for $i\in[2,n-1]$, each $W_l$ is represented by $(2W+1)^{n-2}$ strategies. $U_{i,j}=0$, if it does not change $W_{i-1,j}$. Therefore, there is only one do nothing action for each subset yielding $(n-2)(2W+1)^{n-2}$. Adding for $i=1$ and $i=n$ totals $n(2W+1)^{n-2}$.
\end{proof}

\begin{theorem}
\label{TMNumberOfTransactions}
There are $2nW(2W+1)^{n-2}$ transactions in $\mathfrak{U}$.
\end{theorem}

\begin{proof}
$n(2W+1)^{n-1} - n(2W+1)^{n-2} = 2nW(2W+1)^{n-2}$.
\end{proof}

\paragraph{The Market Universe.} Figure \ref{FigESM17_ESU17_ESZ17_ESH18_Price_20170410} depicts $n=134909$ ticks for ESM17. Low resolution and discreteness hide some. The number of strategies with $|W_{i,j}|\le 1$ is $3^{134908}$. The Sun mass is $1.99 \times 10^{30}$ kg \url{http://solar-center.stanford.edu/vitalstats.html}. In photosphere, 73.46\% by mass is Hydrogen. The mass of its atom is $1.67 \times 10^{-27}$ kg \url{https://en.wikipedia.org/wiki/Unified_atomic_mass_unit}. If the fraction is valid for the entire star, then the number of Hydrogen atoms is $\frac{0.7346 \times 1.99 \times 10^{30}}{1.67 \times 10^{-27}} \approx 8.8 \times 10^{56}$. The latter is \textit{nothing} comparing with the number of potential strategies for ESM17 traded on Monday April 10, 2017, making the $a_1^{PL}$ formula in Theorem \ref{TMSampleMeanPL} not robust. The formula in Theorem \ref{TMNumberOfTransactions} is robust for the number of transactions, not dollars. The distribution of $-2W \le U_{i,j} \le 2W$ is not uniform. Example, for $W=1$, $n=3$: $\#(U_i=-2)=1$, $\#(U_i=-1)=8$, $\#(U_i=0)=9$, $\#(U_i=1)=8$, $\#(U_i=2)=1$. Theorems \ref{TMNumberOfPositionsActions} and \ref{TMNumberOfDNActions} give the total number of actions 27 and $\#(U_i=0)=9$. \textit{The distribution of actions is needed to compute dollars.}

\section{Distribution of Actions}

There are $4W+1$ action types $m \in [-2W, 2W]$, if $W_{0,j}=W_{n,j}=0$, $|W_{i,j}|\le W$. The frequency of do nothing actions $p(m=0,W,n)=\frac{n(2W+1)^{n-2}}{n(2W+1)^{n-1}}=\frac{1}{2W+1}$ is independent on $n$, Theorems \ref{TMNumberOfPositionsActions}, \ref{TMNumberOfDNActions}. "To guess" formulas for $m \ne 0$, the author wrote a C++ program, reserving memory for $(2W+1)^{n-1}$ positions vectors of size $n$ each using \verb|std::vector<std::vector<int>>|, \verb|std::vector<T>::reserve|. Each integer $[0, (2W+1)^{n-1} - 1]$ is divided $n-1$ times by modulo $2W+1$ returning remainder $[0, 2W]$, "pushed back", \verb|std::vector<int>::push_back|, to a corresponding vector. Subtraction $W$ fits values to $[-W, W]$. $n$th zero is "pushed back" to each vector. \verb|std::adjacent_difference| computes actions in own "vector of vectors of integers", simplifying counting actions for $m$. For $W=3$, $m=-3$, it reports $(n, count)$: $(1, 0)$, $(2, 2)$, $(3, 18)$, $(4, 154)$, $(5, 1274)$, $(6, 10290)$, $(7, 81634)$. For $m=0$, counts are in agreement with Theorem \ref{TMNumberOfDNActions}. "Guessing" is dividing the count by $2W+1$ to find the quotient and power of the factor:  $2=14\times 7^{-1}$, $18=18\times 7^{0}$, $154=22\times 7^{1}$, $1274=26\times 7^{2}$, $10290=30\times 7^{3}$, $81634=34\times 7^{4}$. The quotients linearly depend on $n$. The "guessed formula" is $(4n+6)(2W+1)^{n-3}$. Formulas do not work for $n=1$ with single d.n.s. The "guessed formulas" are in Table \ref{TblGuessedCounts}.
\begin{center}
\begin{longtable}{|r|r|r|r|}
\caption[Guessed Counts]{Guessed Counts Formulas, $n\in[2, 9]$.} \label{TblGuessedCounts} \\
 \hline
 \multicolumn{1}{|c|}{$W$} &
 \multicolumn{1}{|c|}{$m$} &
 \multicolumn{1}{c|}{Count} &
 \multicolumn{1}{c|}{Sum} \\
 \hline 
 \endfirsthead
 \multicolumn{4}{c}%
 {\tablename\ \thetable{} -- continued from previous page} \\
 \hline
 \multicolumn{1}{|c|}{$W$} &
 \multicolumn{1}{|c|}{$m$} &
 \multicolumn{1}{c|}{Count} &
 \multicolumn{1}{c|}{Sum} \\
 \hline 
 \endhead
 \hline \multicolumn{4}{|r|}{{Continued on next page}} \\ \hline
 \endfoot
 \hline
 \endlastfoot
 &-2&$(n-2) \times (2W+1)^{n-3}$&\\
 &-1&$(2n+2) \times (2W+1)^{n-3}$&\\
 1&0&$n \times (2W+1)^{n-2}=3n \times (2W+1)^{n-3}$&$9n\times (2W+1)^{n-3}=n \times 3^{n-1}$\\
 &1&$(2n+2) \times (2W+1)^{n-3}$&\\
 &2&$(n-2) \times (2W+1)^{n-3}$&\\
 \hline
 &-4&$(n-2) \times (2W+1)^{n-3}$&\\
 &-3&$(2n-4) \times (2W+1)^{n-3}$&\\
 &-2&$(3n+4) \times (2W+1)^{n-3}$&\\
 &-1&$(4n+2) \times (2W+1)^{n-3}$&\\
 2&0&$n \times (2W+1)^{n-2}=5n \times (2W+1)^{n-3}$&$25n\times (2W+1)^{n-3}=n \times 5^{n-1}$\\
 &1&$(4n+2) \times (2W+1)^{n-3}$&\\
 &2&$(3n+4) \times (2W+1)^{n-3}$&\\
 &3&$(2n-4) \times (2W+1)^{n-3}$&\\
 &4&$(n-2) \times (2W+1)^{n-3}$&\\
 \hline
 &-6&$(n-2) \times (2W+1)^{n-3}$&\\
 &-5&$(2n-4) \times (2W+1)^{n-3}$&\\
 &-4&$(3n-6) \times (2W+1)^{n-3}$&\\
 &-3&$(4n+6) \times (2W+1)^{n-3}$&\\
 &-2&$(5n+4) \times (2W+1)^{n-3}$&\\
 &-1&$(6n+2) \times (2W+1)^{n-3}$&\\
 3&0&$n \times (2W+1)^{n-2}=7n \times (2W+1)^{n-3}$&$49n\times (2W+1)^{n-3}=n \times 7^{n-1}$\\
 &1&$(6n+2) \times (2W+1)^{n-3}$&\\
 &2&$(5n+4) \times (2W+1)^{n-3}$&\\
 &3&$(4n+6) \times (2W+1)^{n-3}$&\\
 &4&$(3n-6) \times (2W+1)^{n-3}$&\\
 &5&$(2n-4) \times (2W+1)^{n-3}$&\\
 &6&$(n-2) \times (2W+1)^{n-3}$&\\
\end{longtable}
\end{center}
The formulas are products of a line $b(m,W) \times n + a(m,W)$ and $(2W+1)^{n-3}$. Distributions are symmetrical. For $m \ne 0$, the frequencies depend on $n$ but $\lim_{n \rightarrow \infty}p(m,W,n)=\lim_{n \rightarrow \infty}\frac{(b(m,W) \times n + a(m,W)) \times (2W+1)^{n-3}}{n\times (2W+1)^{n-1}} = \frac{b(m,W) \times n + a(m,W)}{n \times (2W+1)^2} = \frac{b(m,W)}{(2W+1)^2}$ does not. $a(m=0,W)=0$. The $b(m,W)=2W+1-|m|$, where $m \in [-2W, 2W]$, satisfies all formulas in Table \ref{TblGuessedCounts}.

The $a(m, W)$ is obtained from $n=2$, where the second, \textit{last}, position is zero and the first action $|U_{1,j}| \le W$. The strategies and inverses count two for $|m| \le W$, making $a(m,W)=2|m|$, and zero for $W < |m| \le 2W$, yielding $a(m,W)=2|m|-2(2W+1)$. The united formulas are
\begin{equation}
\label{EqActionsDistribution}
\begin{split}
& A = \{m : |m| \le W \}, \; B = \{m : W < |m| \le 2W \};\\
& Count_A(m, W, n) = \left[(2W+1)n-(n-2)|m|\right](2W+1)^{n-3};\\
& Count_B(m, W, n)=Count_A(m, W, n)-2(2W+1)^{n-2}=\\
&=[(2W+1)n-(n-2)|m| - 2(2W+1)](2W+1)^{n-3};\\
& p_A(m,W,n)=\frac{Count_A(m, W, n)}{n(2W+1)^{n-1}}=\frac{1}{2W+1}-\frac{(n-2)|m|}{n(2W+1)^2};\\
& p_B(m,W,n)=\frac{Count_B(m, W, n)}{n(2W+1)^2}=p_A(m,W,n)-\frac{2}{n(2W+1)}.\\
\end{split}
\end{equation}
The two counts "contain" \textit{all} formulas from Table \ref{TblGuessedCounts}, reproduce $8 \times (5 + 9 + 13) = 216$ C++ experimental values, and satisfy Theorem \ref{TMNumberOfDNActions}, since $\sum_{i=0}^{i=n}i=\sum_{i=1}^{i=n}i=\frac{n(n+1)}{2}$, $\sum_{i=n+1}^{i=2n}i =\sum_{i=1}^{i=2n}i - \sum_{i=1}^{i=n}i=\frac{n(3n+1)}{2}$,

$$\sum_{m=-W}^{m=W}Count_A=n(2W+1)^{n-1}-(n-2)W(W+1)(2W+1)^{n-3};$$
$$\sum_{m=-2W}^{m=-W-1}Count_B+\sum_{m=W+1}^{m=2W}Count_B=2\sum_{m=W+1}^{m=2W}Count_B=$$
$$=(W^2n+Wn-2W^2-2W)(2W+1)^{n-3};$$
$$\sum_{m=-W}^{m=W}Count_A+2\sum_{m=W+1}^{m=2W}Count_B=n(2W+1)^{n-1},$$

\begin{proof}
Vladimir Arnold recollects \cite[p. 29]{arnold2004} the words of his teacher Andrey Kolmogorov (VS's translation): \textit{"Do not look for a mathematical sense in my hydrodynamics achievements. ... I did not derive them from initial axioms or definitions (as physicists say, from the "first principals"): my results are not proved but valid and this is much more important!"} The C++ experiments convinced the author of the correctness of formulas \ref{EqActionsDistribution} and left admiration of the Kolmogorov's words. However, Anderzej Pelc's \textit{"Why Do We Believe Theorems?"} \cite{pelc2011} "pressed" not to publish the formulas without a proof.

The author could not move from $n$ to $n+1$ using \textit{mathematical induction}. \textit{Generating functions} \cite{stanley1990}, \cite{lando2007} require coefficients - a \textit{vicious circle}. But ...

By construction, positions $[-W, W]$ are uniformly distributed in the matrix $n$ ticks $\times$ $[S=(2W+1)^{n-1}]$ strategies within the first $1, \dots, n-1$ rows
\[
Positions=\mathcal{W}=
   \begin{bmatrix}
   W_{1,1} & W_{1,2} & \dots & W_{1,S} \\
   \dots & \dots & \dots & \dots \\
   W_{n-1,1} & W_{n-1,2} & \dots & W_{n-1,S} \\
   0 & 0 & \dots & 0 \\
   \end{bmatrix}
\]
Each row, except $n$th, has $\frac{(2W+1)^{n-1}}{2W+1}=(2W+1)^{n-2}$ positions of each type. The actions are adjacent differences in columns $U_{i,j}=W_{i,j}-W_{i-1,j}$
\[
Actions=\mathcal{U}=
   \begin{bmatrix}
   U_{1,1}=W_{1,1}-0 & \dots & U_{1,S}=W_{1,S}-0 \\
   U_{2,1}=W_{2,1}-W_{1,1} & \dots & U_{2,S}=W_{2,N}-W_{1,S} \\
   \dots & \dots & \dots \\
   U_{n-1,1}=W_{n-1,1}-W_{n-2,1} & \dots & U_{n-1,S}=W_{n-1,S}-W_{n-2,S} \\
   U_{n,1}=0-W_{n-1,1} & \dots & U_{n,S}=0-W_{n-1,S} \\
   \end{bmatrix}
\]
The matrix $(2W+1) \times (2W+1)$ of all \textit{individual} transitions
$$
  \begin{array}{rrrrrrrr}
  | & -W & -W+1 & \dots & 0 & \dots & W-1 & W \\
  ---- & ---- & ---- & -- & ---- & -- & ---- & ---- \\
  -W \rightarrow & 0 & 1 & \dots & W & \dots & 2W-1 & 2W \\
  -W+1 \rightarrow & -1 & 0 & \dots & W-1 & \dots & \dots & 2W - 1 \\
  \dots \rightarrow & \dots & \dots & \dots & \dots & \dots & \dots & \dots \\
  0 \rightarrow & -W & -W+1 & \dots & 0 & \dots & W-1 & W \\
  \dots \rightarrow & \dots & \dots & \dots & \dots & \dots & \dots & \dots \\
  W-1 \rightarrow & -2W+1 & -2W+2 & \dots & -W+1 & \dots & 0 & 1 \\
  W \rightarrow & -2W & -2W+1 & \dots & -W & \dots & -1 & 0 \\
  \end{array}
$$
is applied to the ticks $[1, n-2]$. Due to uniformness of positions in ticks $[1, n - 1]$, each of the $(2W+1)^{n-2}$ types, in moves from ticks $[1, n-2]$, is changed to $(2W+1)^{n-2}$ types: $\frac{1}{2W+1}$ of actions transfer a position type to another, both from $[-W,W]$. The number of transitions from one type to another is $\frac{(2W+1)^{n-2}}{2W+1}=(2W+1)^{n-3}$. Therefore, the number of actions of one type $m\in[-2W,2W]$ is the length of the diagonal $2W+1 - |m|$. For the ticks $[1, n-2]$ this yields $(n-2)(2W+1-|m|)$ \textit{individual} actions, which must be multiplied by $(2W+1)^{n-3}$. The total is $Count_B=[(2W+1)n-(n-2)|m|-2(2W+1)](2W+1)^{n-3}$. Remaining transitions $0 \rightarrow 1$, $(n-1) \rightarrow n$ add $2(2W+1)(2W+1)^{n-3}$ actions only for $m\in[-W, W]$. Adding it to $Count_B$ yields $Count_A$ and completes the proof of the next Theorem for the \textit{new discrete distribution}.
\end{proof}
\begin{theorem}
\label{TMActionsDistribution}
Formulas \ref{EqActionsDistribution} give the distribution of actions $m$ in $\mathfrak{U}$.
\end{theorem}
For $n=4$, $W=1$, the illustration of transposed matrices is
$$
\mathcal{W}^T=
   \left [ \begin{array}{rrrr}
   -1 & -1 & -1 & 0 \\
   0 & -1 & -1 & 0 \\
   1 & -1 & -1 & 0 \\
   -1 & 0 & -1 & 0 \\
   0 & 0 & -1 & 0 \\
   1 & 0 & -1 & 0 \\
   -1 & 1 & -1 & 0 \\
   0 & 1 & -1 & 0 \\
   1 & 1 & -1 & 0 \\

   -1 & -1 & 0 & 0 \\
   0 & -1 & 0 & 0 \\
   1 & -1 & 0 & 0 \\
   -1 & 0 & 0 & 0 \\
   0 & 0 & 0 & 0 \\
   1 & 0 & 0 & 0 \\
   -1 & 1 & 0 & 0 \\
   0 & 1 & 0 & 0 \\
   1 & 1 & 0 & 0 \\

   -1 & -1 & 1 & 0 \\
   0 & -1 & 1 & 0 \\
   1 & -1 & 1 & 0 \\
   -1 & 0 & 1 & 0 \\
   0 & 0 & 1 & 0 \\
   1 & 0 & 1 & 0 \\
   -1 & 1 & 1 & 0 \\
   0 & 1 & 1 & 0 \\
   1 & 1 & 1 & 0 \\
   \end{array}
   \right ],
\;
\mathcal{U}^T=
   \left [ \begin{array}{rrrr}
   -1 & 0 & 0 & 1 \\
   0 & -1 & 0 & 1 \\
   1 & -2 & 0 & 1 \\
   -1 & 1 & -1 & 1 \\
   0 & 0 & -1 & 1 \\
   1 & -1 & -1 & 1 \\
   -1 & 2 & -2 & 1 \\
   0 & 1 & -2 & 1 \\
   1 & 0 & -2 & 1 \\

   -1 & 0 & 1 & 0 \\
   0 & -1 & 1 & 0 \\
   1 & -2 & 1 & 0 \\
   -1 & 1 & 0 & 0 \\
   0 & 0 & 0 & 0 \\
   1 & -1 & 0 & 0 \\
   -1 & 2 & -1 & 0 \\
   0 & 1 & -1 & 0 \\
   1 & 0 & -1 & 0 \\

   -1 & 0 & 2 & -1 \\
   0 & -1 & 2 & -1 \\
   1 & -2 & 2 & -1 \\
   -1 & 1 & 1 & -1 \\
   0 & 0 & 1 & -1 \\
   1 & -1 & 1 & -1 \\
   -1 & 2 & 0 & -1 \\
   0 & 1 & 0 & -1 \\
   1 & 0 & 0 & -1 \\
   \end{array}
   \right ].
$$
Figure \ref{FigActionsDistribution1} plots the corresponding probability mass function, PMF, of actions $p(m, W=1, n=4)$ together with PMF for $W=1$ and other $n$, demonstrating the limit distribution. The distributions are discrete and lines serve only to a better visibility of dots.
\begin{figure}[!h]
  \centering
  \includegraphics[width=130mm]{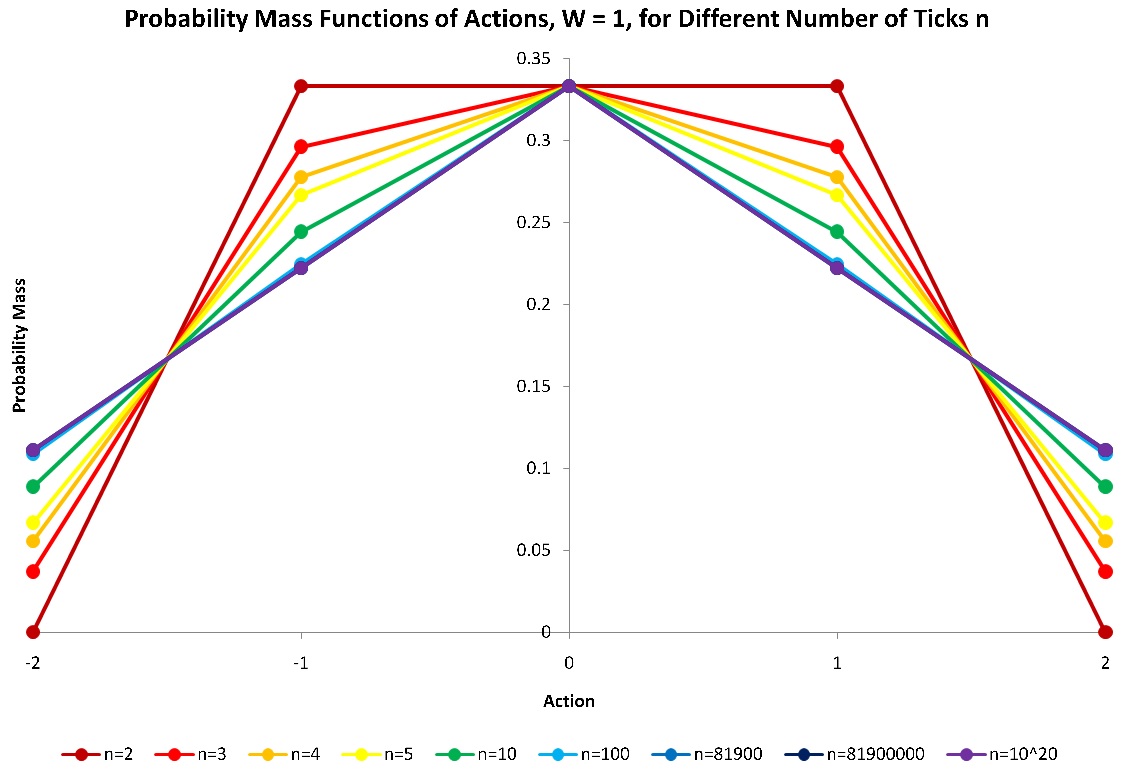}
  \caption[FigActionsDistribution1]
   {Probability mass functions of actions for strategies with position limit $\pm1$ contract. Plotted using Microsoft Excel.}
  \label{FigActionsDistribution1}
\end{figure}
\begin{figure}[!h]
  \centering
  \includegraphics[width=130mm]{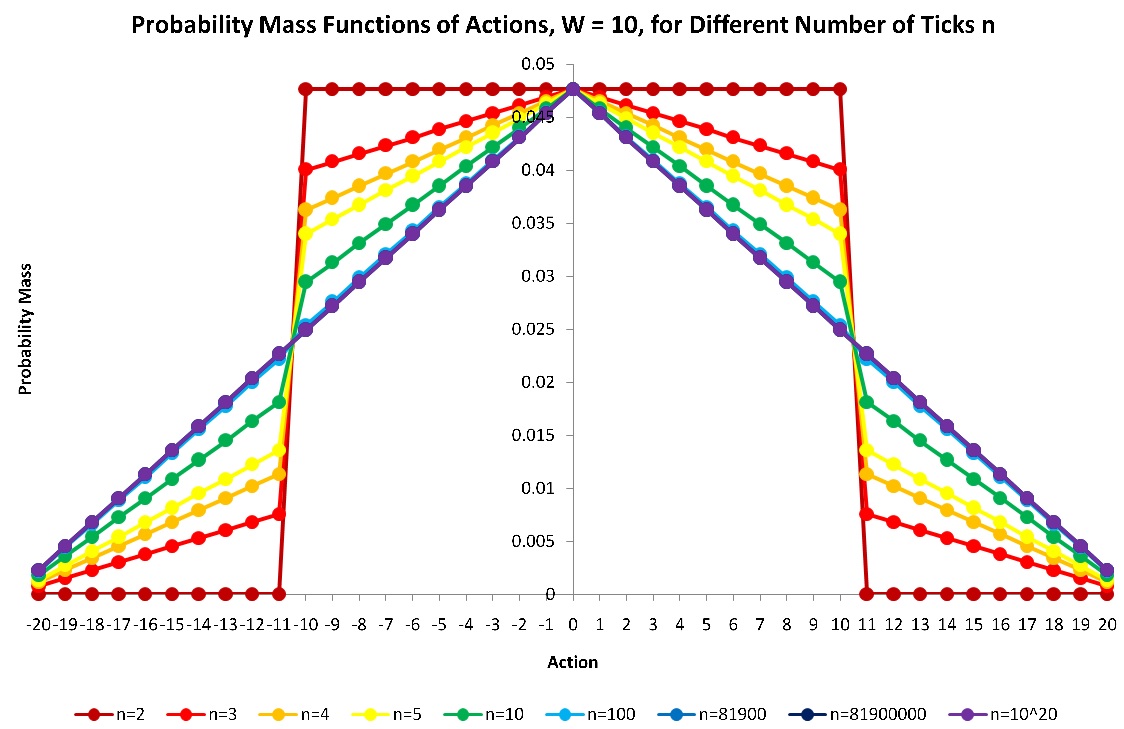}
  \caption[FigActionsDistribution20]
   {Probability mass functions of actions for strategies with position limit $\pm10$ contracts. Plotted using Microsoft Excel.}
  \label{FigActionsDistribution20}
\end{figure}

The $n=81900$ is seconds in the trading session of S\&P 500 E-Mini futures symbolizing the one trade per second frequency. n=81900000 corresponds to a hypothetical case of 1000 trades per second. Figure \ref{FigActionsDistribution20} presents PMF of actions for strategies with position limit $W=10$ and different numbers of ticks.

Dollars paid to the industry as constant costs per contract $C$  for $(2W+1)^{n-1}$ strategies, applied each one time, can be computed using the symmetry of the distribution as weighted actions and then divided by $(2W+1)^{n-1}$ to get $a_1^{PL}$

\begin{equation}
\label{EqIndustryGains}
\begin{split}
\$ = C \left(\sum_{m=-W}^{m=W}|m|Count_A + \sum_{m=-2W}^{m=-W-1}|m|Count_B+\sum_{m=W+1}^{m=2W}|m|Count_B \right)\\
=2C \left( \sum_{m=1}^{m=W}m Count_A + \sum_{m=W+1}^{m=2W}m Count_B \right)\\
=2C \left( \left (\sum_{m=1}^{m=2W}m Count_A \right) - W(2W+1)(3W+1)(2W+1)^{n-3}\right)\\
=2CW(2W+1)^{n-2}\left((2W+1)n - \frac{(n-2)(4W+1)}{3}-(3W+1)\right)\\
=2CW(W+1)(2W+1)(2W+1)^{n-3}\frac{2n-1}{3},\\
a_1^{PL}=-\frac{2CW(W+1)(2n-1)}{3(2W+1)} \equiv \frac{-C\sum_{j=1}^{(2W+1)^{n-1}} abs(\boldsymbol{U_j})}{(2W+1)^{n-1}}.
\end{split}
\end{equation}
Since $\sum_{i=1}^{i=n}i^2=\frac{n(n+1)(2n+1)}{6}$, $W(W+1)(2W+1)$ is divisible by three and six. The last equation is an identity, where the left side is trivial but the right one can "exhaust" any computer.

\paragraph{Strategies generating extreme industry gains.} The minimum zero gain is generated only by one strategy - d.n.s. Only two strategies create the maximum gain $2CW(n-1)$ each. Indeed, the maximum action \textit{reverses} long to short and vice versa extreme positions: $W\rightarrow -W$ or $-W\rightarrow W$. This can be done in ticks $[2, n-1]$. The ticks 1 and $n$ add together maximum $2CW$.

\paragraph{Distribution function.} Following to Eugene Lukacs \cite[pp. 5-6, p. 17]{lukacs1970}, any purely discrete distribution can be written in the form $F(x)=\sum_jp_j\varepsilon(x-\xi_j)$, where $x$, $p_j$, $\xi_j$ are real, $p_j$ satisfy $p_j \ge 0$, and $\sum_jp_j=1$, and
\[ \varepsilon(x) = \left\{ \begin{array}{ll}
         0, & \mbox{if $x < 0$}\\
         1, & \mbox{if $x \ge 0$}. \end{array} \right. \]
$\xi_j$ are discontinuity points of $F(x)$. $p_j$ is the \textit{saltus} at $\xi_j$. Let us enumerate types $m\in[-2W,2W]$ by $j=m+2W+1$, where $\xi_j=j-2W-1=m$, and $p_A=p_A(m,W,n)$, $p_B=p_B(m,W,n)$ are from Equations \ref{EqActionsDistribution}. Then, the $F(x)$ is

\begin{equation}
\label{EqCDF}
\begin{split}
&F(x)=\sum_{j=1}^{j=4W+1}p_j\varepsilon(x-\xi_j)=\sum_{j=1}^{j=4W+1}p_j\varepsilon(x-j+2W+1)=\\
&=\sum_{m=-2W}^{m=-W-1}p_B\varepsilon(x-m)+\sum_{m=-W}^{m=W}p_A\varepsilon(x-m)+\sum_{m=W+1}^{m=2W}p_B\varepsilon(x-m)\\
&=\frac{1}{2W+1}\sum_{m=-2W}^{m=2W}\varepsilon(x-m)-\frac{n-2}{n(2W+1)^2}\sum_{m=-2W}^{m=2W}|m|\varepsilon(x-m)\\
&-\frac{2}{n(2W+1)}\left(\sum_{m=-2W}^{m=-W-1}\varepsilon(x-m) + \sum_{m=W+1}^{m=2W}\varepsilon(x-m) \right).
\end{split}
\end{equation}

\paragraph{Characteristic function} $f(t)=\int_{-\infty}^{\infty}e^{itx}dF(x)$, $i=\sqrt{-1}$, for a purely discrete distribution reduces to the sum $f(t)=\sum_jp_je^{it\xi_j}$ \cite[p. 17]{lukacs1970} yielding
\begin{equation}
\label{EqCF}
\begin{split}
&f(t)=\frac{1}{2W+1}\sum_{m=-2W}^{m=2W}e^{itm}-\frac{n-2}{n(2W+1)^2}\sum_{m=-2W}^{m=2W}|m|e^{itm}\\
&-\frac{2}{n(2W+1)}\left(\sum_{m=-2W}^{m=-W-1}e^{itm} + \sum_{m=W+1}^{m=2W}e^{itm} \right).
\end{split}
\end{equation}
The function is real and even $f(-t)=f(t)$, since in sums the formula contains only pairs $e^{-ix} + e^{ix}=2\cos(x)$. This ensures that the distribution is symmetric \cite[p. 30, Theorem 3.1.2]{lukacs1970}, Figures \ref{FigActionsDistribution1}, \ref{FigActionsDistribution20}. Therefore,
\begin{equation}
\label{EqCFcos}
\begin{split}
&f(t)=\frac{1+2\sum_{m=1}^{m=2W}\cos(tm)}{2W+1}-\frac{2(n-2)\sum_{m=1}^{m=2W}m\cos(tm)}{n(2W+1)^2}\\
&-\frac{4\sum_{m=W+1}^{m=2W}\cos(tm)}{n(2W+1)}.
\end{split}
\end{equation}
\paragraph{Moments.} We compute \textit{beginning moments} of order $s=1, \dots$, if they exist, using \cite[p. 69, Lemma 2, Equation 11]{gnedenko1949} $\alpha_s=\frac{1}{i^s}[\frac{d^s}{dt^s}f(t)]_{t=0}$ and \textit{central moments} using \cite[p. 69, Equation 13]{gnedenko1949} $\mu_s=\frac{1}{i^s}[\frac{d^s}{dt^s}e^{it\alpha_1}f(t)]_{t=0}$. Examples,
\begin{equation}
\label{EqActionsMean}
\begin{split}
&f'(t)=\frac{-2\sum_{m=1}^{m=2W}m\sin(tm)}{2W+1}+\frac{2(n-2)\sum_{m=1}^{m=2W}m^2\sin(tm)}{n(2W+1)^2}\\
&+\frac{4\sum_{m=W+1}^{m=2W}m\sin(tm)}{n(2W+1)},\\
&\mathrm{mean}=\alpha_1=\frac{f'(0)}{i}=0.
\end{split}
\end{equation}
\begin{equation}
\label{EqActionsVariance}
\begin{split}
&f''(t)=\frac{-2\sum_{m=1}^{m=2W}m^2\cos(tm)}{2W+1}+\frac{2(n-2)\sum_{m=1}^{m=2W}m^3\cos(tm)}{n(2W+1)^2}\\
&+\frac{4\sum_{m=W+1}^{m=2W}m^2\cos(tm)}{n(2W+1)},\\
&\alpha_2=-f''(0)=\frac{2\sum_{m=1}^{m=2W}m^2}{2W+1}-\frac{2(n-2)\sum_{m=1}^{m=2W}m^3}{n(2W+1)^2}-\frac{4\sum_{m=W+1}^{m=2W}m^2}{n(2W+1)}\\
&=\frac{2W(4W+1)}{3}-\frac{2(n-2)W^2}{n}-\frac{2W(7W+1)}{3n}=\frac{2W(W+1)(n-1)}{3n},\\
&\mathrm{variance}=\mu_2=\alpha_2-\alpha_1^2=\frac{2W(W+1)(n-1)}{3n}.
\end{split}
\end{equation}

Let us prove $\frac{d^s}{dt^s}\cos(mt)=m^s\cos(mt+\frac{\pi s}{2})$, useful for computing the moments of higher orders. The \textit{induction basis}: for $s=0,1,2$, it is correct $\cos(mt)$, $-m\sin(mt)$, $-m^2\cos(mt)$. Let us assume it is correct for $2 < s$. Then, for $s+1$, it is $m^{s+1}\cos(mt+\frac{\pi}{2}(s+1))=m^{s+1}[\cos(mt)\cos(\frac{\pi s}{2}+\frac{\pi}{2})-\sin(mt)\sin(\frac{\pi s}{2}+\frac{\pi}{2})]=m^{s+1}[-\cos(mt)\sin(\frac{\pi s}{2})-\sin(mt)\cos(\frac{\pi s}{2})]=-m^{s+1}\sin(mt+\frac{\pi s}{2})$. Explicit differentiation yields the same: $\frac{d}{dt}m^s\cos(mt+\frac{\pi s}{2})=-m^{s+1}\sin(mt+\frac{\pi s}{2})$. The \textit{induction step} is completed. $\blacksquare$

We get
\begin{equation}
\label{EqCFcosSDerivative}
\begin{split}
&\left[\frac{d^s}{dt^s}f(t)\right]_{t=0}= \frac{d^s}{dt^s}\left (\frac{1}{2W+1} \right) + \frac{2\cos(\frac{\pi s}{2})\sum_{m=1}^{m=2W}m^s}{2W+1}\\
&-\frac{2(n-2)\cos(\frac{\pi s}{2})\sum_{m=1}^{m=2W}m^{s+1}}{n(2W+1)^2}-\frac{4\cos(\frac{\pi s}{2})\sum_{m=W+1}^{m=2W}m^s}{n(2W+1)};\\
&\frac{d^s}{dt^s}\left (\frac{1}{2W+1} \right) = \frac{1}{2W+1} \; \mathrm{for} \; s = 0 \; \mathrm{or} \; 0 \; \mathrm{for} \; 0 < s.
\end{split}
\end{equation}
The right side is zero for odd $1 \le s=2q+1, q=0,1,\dots$, since $\cos(\frac{\pi s}{2})=\cos(\pi q + \frac{\pi}{2})=-\sin(\pi q)=0$ is the common multiplier. Thus, odd beginning and central, since $\alpha_1=0$, moments are zeros in agreement with symmetry of $S(m,W,n)$ about $m=0$, see \cite[p. 183, 15.8 Measures of skewness and excess]{cramer1962}.

\paragraph{Distributions of actions in time $i$-slices.} Formulas \ref{EqActionsDistribution} count actions $m$ in $\mathfrak{U}$. Slices, by $i=1,\dots,n$, of $S=(2W+1)^{n-1}$ strategies can be interpreted as time $i$-slices and divided on two groups 1) $i=1$, $n$; 2) $i=2,\dots,n-1$. In the 1- and $n$-slice, each action $m \in [-W, W]$ has $\frac{(2W+1)^{n-1}}{2W+1}=(2W+1)^{n-2}$ entries. In each slice of the second group, each action $m\in[-2W,2W]$ occurs $\frac{Count_B(m,W,n)}{n-2}=\frac{[(2W+1)n-(n-2)|m|-2(2W+1)](2W+1)^{n-3}}{n-2}=(2W+1-|m|)(2W+1)^{n-3}$ times. Checking: $\sum_{m=-2W}^{m=2W}(2W+1-|m|)(2W+1)^{n-3}=(2W+1)^{n-3}[(2W+1)(4W+1)-2W(2W+1)]=(2W+1)^{n-1}$. The following sums will be needed

\begin{equation}
\label{EqActionsMoments_1_2}
\begin{split}
&i=1,\dots,n: \; \sum_{j=1}^{j=S}U_{i,j}=0;\\
&i=1, n: \; \sum_{j=1}^{j=S}|U_{i,j}|=\sum_{m=-W}^{m=W}|m|(2W+1)^{n-2}=W(W+1)(2W+1)^{n-2};\\
&i=1, n: \; \sum_{j=1}^{j=S}U_{i,j}^2=\sum_{m=-W}^{m=W}m^2(2W+1)^{n-2}=\frac{W(W+1)(2W+1)^{n-1}}{3};\\
&i=2,\dots,n-1: \; \sum_{j=1}^{j=S}|U_{i,j}|=\sum_{m=-2W}^{m=2W}|m|(2W+1-|m|)(2W+1)^{n-3}=\\
&=2(2W+1)^{n-3}\left ((2W+1)\sum_{m=1}^{m=2W}m - \sum_{m=1}^{m=2W}m^2 \right )=\\
&=\frac{4W(W+1)(2W+1)^{n-2}}{3};\\
&i=2,\dots,n-1: \; \sum_{j=1}^{j=S}U_{i,j}^2=\sum_{m=-2W}^{m=2W}m^2(2W+1-|m|)(2W+1)^{n-3}=\\
&=2(2W+1)^{n-3} \left ((2W+1)\sum_{m=1}^{m=2W}m^2 - \sum_{m=1}^{m=2W}m^3 \right ) =\\
&=2(2W+1)^{n-3} \left (\frac{W(2W+1)^2(4W+1)}{3} - \frac{(2W)^2(2W+1)^2}{4} \right ) =\\
&=\frac{2W(W+1)(2W+1)^{n-1}}{3}.
\end{split}
\end{equation}

\begin{theorem}
\label{TMPositionsSampleCovariances}
$\forall i,l \in [1,n-1]$, $\sum_{j=1}^{j=S} W_{i,j} W_{l,j} = \frac{W(W+1)(2W+1)^{n-1}}{3} \delta_{i,l}$, where the Kronecker delta $\delta_{i,l} = \left\{ \begin{array}{ll} 0, & \mbox{if $i \ne l$}\\ 1, & \mbox{if $i = l$} \end{array} \right.$. The sum is zero, if $i=n \vee l=n$.
\end{theorem}

\begin{proof}
In the 1-slice, each position from $[-W, W]$ is repeated $(2W+1)^{n-2}$ times. Since we consider $S=(2W+1)^{n-1}$ unique vectors of positions, for $n>2$, any $W_1$ associates with $(2W+1)^{n-3}$ positions $l$ of each type $[-W,W]$ in the 2-slice. The sum of products of the constant $W_1$ to these values is zero: $\sum_{l=-W}^{l=W}W_1 l (2W+1)^{n-3}=W_1 (2W+1)^{n-3}\sum_{l=-W}^{l=W}l=0$. $W_1$ is selected arbitrary making the conclusion valid for any $[-W, W]$: $\sum_{j=1}^{j=S}W_{1,j}W_{2,j}=0$. Similar argumentation can be applied to any pair of distinct i-slices, $i=1,\dots,n-1$. Lexicographical sorting of positions by values in two slices, ignoring others, helps to see it. For a pair including $n$-slice, it is trivial because the latter is zero vector. Therefore, $\forall i \ne l \vee i=n \vee l=n, \; \sum_{j=1}^{j=S}W_{i,j} W_{l,j}=0$. $\forall i=l \ne n$, $\sum_{j=1}^{j=S}W_{i,j}W_{l,j}=\sum_{j=1}^{j=S}W_{i,j}^2=(2W+1)^{n-2}\sum_{l=-W}^{l=W}l^2=\frac{W(W+1)(2W+1)^{n-1}}{3}$. To shorten the formula for $i,l=1,\dots,n-1$ using the Kronecker delta is natural.
\end{proof}
In other words, the columns of the transposed position matrix $(\mathcal{W}_{n \times S})^T$ are mutually orthogonal vectors. The sums $\sum_{j=1}^{j=S}U_{i,j}^2$ for $i=1,n$ and $i=2,...,n-1$ are given by Equations \ref{EqActionsMoments_1_2}. They play the role of sample variances of actions in $i$-slices times $(S-1)$ or $S$. In contrast, $\sum_{j=1}^{j=S}U_{i,j}U_{i+l,j}$ for $i=1,\dots,n-1$ and $l=1,\dots,n-i$ play the role of sample covariances times $(S-1)$ or $S$.

\begin{theorem}
\label{TMActionsSampleCovariances}
For  $i=1,\dots,n-1, \; l=1,\dots,n-i$, $\sum_{j=1}^{j=S}U_{i,j}U_{i+l,j}=0$ for $l>1$ and $-\frac{W(W+1)(2W+1)^{n-1}}{3}$ for $l=1$.
\end{theorem}

\begin{proof}
$\sum_{j=1}^{j=S}U_{i,j}U_{i+l,j}=\sum_{j=1}^{j=S}(W_{i,j}-W_{i-1,j})(W_{i+l,j}-W_{i+l-1,j}) = \\
 -\sum_{j=1}^{j=S}W_{i,j}W_{i+l-1,j} + \sum_{j=1}^{j=S}W_{i,j}W_{i+l,j} - 
\sum_{j=1}^{j=S}W_{i-1,j}W_{i+l,j} + \\
\sum_{j=1}^{j=S}W_{i-1,j}W_{i+l-1,j}$. By Theorem \ref{TMPositionsSampleCovariances}, the last three sums are zeros. The first sum is zero for $l>1$ and $-\frac{W(W+1)(2W+1)^{n-1}}{3}$ for $l=1$.
\end{proof}

\begin{theorem}
\label{TMAbsolutePositionsSampleCovariances}
$\forall i,l \in [1,n-1] \wedge i\ne l$, $\sum_{j=1}^{j=S}|W_{i,j}||W_{l,j}|=W^2(W+1)^2(2W+1)^{n-3}$. The sum is zero, if $i=n \vee l=n$.
\end{theorem}

\begin{proof}
In a pair of $i$-, $l$-slices, $i,l \in [1, n-1]$, $i\ne l$, the number of unique pairs $(W_{i,j},W_{l,j})$ taken once is $(2W+1)^2$. For them, $\sum_{W_{i,j}=-W}^{W_{i,j}=W}\sum_{W_{l,j}=-W}^{W_{l,j}=W}|W_{i,j}||W_{l,j}|\\
=\sum_{W_{i,j}=-W}^{W_{i,j}=W}|W_{i,j}|\sum_{W_{l,j}=-W}^{W_{l,j}=W}|W_{l,j}|=4\sum_{W_{i,j}=1}^{W_{i,j}=W}W_{i,j}\sum_{W_{l,j}=1}^{W_{l,j}=W}W_{l,j}=W^2(2W+1)^2$. Each pair is repeated $\frac{(2W+1)^{n-1}}{(2W+1)^2}=(2W+1)^{n-3}$ times, making the total $W^2(W+1)^2(2W+1)^{n-3}$ or zero, if $i=n \vee l=n$: the $n$-slice is zero vector.
\end{proof}

\begin{theorem}
\label{TMAbsoluteActionsSampleCovariances}
For $1 \le W$, the following formulas take place
\begin{equation*}
\begin{split}
&\mathrm{A}: n=2, \; \sum_{j=1}^{j=S}|U_{1,j}||U_{2,j}|=\frac{1}{3}W(W+1)(2W+1);\\
&\mathrm{B}: 2 < n, \; \sum_{j=1}^{j=S}|U_{1,j}||U_{2,j}|=\frac{3}{2}W^2(W+1)^2(2W+1)^{n-3};\\
&\mathrm{C}: 2 < n, \; \sum_{j=1}^{j=S}|U_{1,j}||U_{n,j}|=W^2(W+1)^2(2W+1)^{n-3};\\
&\mathrm{D}: 3 < n, \; 2 < i < n, \; \sum_{j=1}^{j=S}|U_{1,j}||U_{i,j}|=\frac{4}{3}W^2(W+1)^2(2W+1)^{n-3};\\
&\mathrm{E}: 1 < i < n - 1 \; \sum_{j=1}^{j=S}|U_{i,j}||U_{i+1,j}|=\frac{1}{15}W(28W^3+56W^2+27W-1)(2W+1)^{n-3};\\
&\mathrm{F}: 1 < i < n - 2, \; i + 1 < r < n, \; \sum_{j=1}^{j=S}|U_{i,j}||U_{r,j}|=\frac{16}{9}W^2(W+1)^2(2W+1)^{n-3}.
\end{split}
\end{equation*}
\end{theorem}

\begin{proof}
$\forall j\in[1,S] \wedge 2 \le n, \; U_{1,j}=W_{1,j}$. For $n=2$, $U_{2,j}=-U_{1,j}=-W_{1,j}$,
\begin{equation*}
\begin{split}
&\mathrm{A:} \; \sum_{j=1}^{j=S}|U_{1,j}||U_{2,j}|=\sum_{j=1}^{j=S}W_{1,j}^2=\sum_{l=-W}^{l=W}l^2=2\sum_{l=1}^{l=W}l^2=\frac{W(W+1)(2W+1)}{3}.
\end{split}
\end{equation*}
For $n>2$, there are $(2W+1)^{n-2}$ values of $U_{1,j}=W_{1,j}$ of each type $[-W,W]$ and $-W \le U_{1,j}+U_{2,j} \le W$ or $-W-W_{1,j} \le U_{2,j} \le W-W_{1,j}$. The $(2W+1)^{n-3}$ values $U_{1,j}=-W$ are followed once by each $U_{2,j} \in [0,2W]$. The $(2W+1)^{n-3}$ values $U_{1,j}=-W+1$ are followed once by each $U_{2,j} \in [-1,2W-1]$. $\dots$ The $(2W+1)^{n-3}$ values $U_{1,j}=W$ are followed once by each $U_{2,j} \in [-2W,0]$. Thus,
\begin{equation*}
\begin{split}
&\mathrm{B:} \; \sum_{j=1}^{j=S}|U_{1,j}||U_{2,j}|=(2W+1)^{n-3}\left ( \sum_{m=0}^{m=2W}|-W||m| + \dots +\sum_{m=-2W}^{m=0}|W||m| \right )=\\
&=(2W+1)^{n-3}\sum_{l=0}^{l=2W}\sum_{m=-l}^{m=2W-l}|-W+l||m|=\\
&=2(2W+1)^{n-3}\sum_{l=0}^{l=W-1}(W-l)\sum_{m=-l}^{m=2W-l}|m|=\\
&=2(2W+1)^{n-3}\sum_{l=0}^{l=W-1}(W-l)\left (\sum_{m=1}^{m=l}m + \sum_{m=1}^{m=2W-l}m \right )=\\
&=2(2W+1)^{n-3}\sum_{l=0}^{l=W-1}(W-l)\left ( \frac{l(l+1)}{2} + \frac{(2W-l)(2W-l+1)}{2} \right )=\\
&=2(2W+1)^{n-3}\sum_{l=0}^{l=W-1}(-l^3+3Wl^2-(4W^2+W)l+W^2(2W+1))=\\
&=-2(2W+1)^{n-3}\frac{(W-1)^2W^2}{4}+6W(2W+1)^{n-3}\frac{(W-1)W(2W-1)}{6}-\\
&-2(2W+1)^{n-3}W(4W+1)\frac{(W-1)W}{2}+2W^3(2W+1)^{n-2}=\\
&=\frac{3}{2}W^2(W+1)^2(2W+1)^{n-3}.
\end{split}
\end{equation*}
Since $U_{n,j}=-W_{n-1,j}$ is uniformly distributed, each of $2W+1$ values $[-W, W]$ of $U_{1,j}$ represented by $(2W+1)^{n-2}$ times associates with $(2W+1)^{n-3}$ actions from $[-W,W]$ and for $2 < n$
\begin{equation*}
\begin{split}
&\mathrm{C:} \; \sum_{j=1}^{j=S}|U_{1,j}||U_{n,j}|=(2W+1)^{n-3}\sum_{l=-W}^{l=W}\sum_{m=-W}^{m=W}|l||m|=\\
&=(2W+1)^{n-3}\sum_{l=-W}^{l=W}|l|\sum_{m=-W}^{m=W}|m|=W^2(W+1)^2(2W+1)^{n-3}.
\end{split}
\end{equation*}
Using distribution of actions in time $i$-slices, we get for $3 < n$, and $2 < i < n$
\begin{equation*}
\begin{split}
&\mathrm{D:} \; \sum_{j=1}^{j=S}|U_{1,j}||U_{i,j}|=\sum_{l=-W}^{l=W}|l|\frac{\sum_{m=-2W}^{m=2W}|m|(2W+1-|m|)(2W+1)^{n-3}}{2W+1}=\\
&=(2W+1)^{n-4}W(W+1)\left ((2W+1)\sum_{m=-2W}^{m=2W}|m|- \sum_{m=-2W}^{m=2W}m^2\right )=\\
&=\frac{4}{3}W^2(W+1)^2(2W+1)^{n-3}.
\end{split}
\end{equation*}
For $1 < i < n - 1$, each $i$-slice, containing $(2W+1-|l|)(2W+1)^{n-3}$ actions $l$, is followed by a $(i+1)$-slice with the same actions and counts. Actions associations between neighboring slices are not arbitrary. For $W=1$, $U_{i,j}=-2$ is followed by $U_{i+1,j}=0, 1, 2$. $U_{i,j}=-1$ creates $W_{i,j}=-1$ or $W_{i,j}=0$ with $U_{i+1,j}=0, 1, 2$ or $-1, 0, 1$. Lexicographical sorting of strategies by $i$- and $(i+1)$-actions uncovers the association pattern repeated $(2W+1)^{n-4}$ times. Known $\sum_{i=1}^{i=n}n^4=\frac{1}{30}n(n+1)(2n+1)(3n^2+3n-1)$ and the sums of powers 1, 2, 3 of the natural numbers are applied. An elegant method for arbitrary powers is explained by Etherington \cite{etherington1932}.
\begin{equation*}
\begin{split}
&\mathrm{E:} \; \sum_{j=1}^{j=S}|U_{i,j}||U_{i+1,j}|=(2W+1)^{n-4}\sum_{l=-2W}^{l=2W}|l|\sum_{i=0}^{i=2W-|l|}\sum_{m=-i}^{m=2W-i}|m|=\\
&=2(2W+1)^{n-4}\sum_{l=1}^{l=2W}l\sum_{i=0}^{i=2W-l}\left (\sum_{m=1}^{m=i}m + \sum_{m=1}^{m=2W-i}m \right )=\\
&=2(2W+1)^{n-4}\sum_{l=1}^{l=2W}l\sum_{i=0}^{i=2W-l}\left ((i-W)^2+W(W+1)\right )=\\
&=\frac{1}{3}(2W+1)^{n-4}\sum_{l=1}^{l=2W}l(2W+1-l)[2l^2-(2W+1)l+8W(W+1)]=\\
&=\frac{1}{15}W(28W^3+56W^2+27W-1)(2W+1)^{n-3}.
\end{split}
\end{equation*}
For a pair of non-neighboring slices $1 < i < n - 2, \; i+1 < r < n$, actions dependence is "forgotten". Again, lexicographical sorting of strategies by actions in $i$- and $r$-slice uncovers the pattern repeated $(2W+1)^{n-4}$ times
\begin{equation*}
\begin{split}
&\mathrm{F:} \; \sum_{j=1}^{j=S}|U_{i,j}||U_{r,j}|=(2W+1)^{n-4}\sum_{l=-2W}^{l=2W}|l|(2W+1-|l|)\sum_{m=-2W}^{m=2W}|m|(2W+1-|m|)=\\
&=\frac{8}{3}W(W+1)(2W+1)^{n-3}\sum_{l=1}^{l=2W}l(2W+1-l)=\frac{16}{9}W^2(W+1)^2(2W+1)^{n-3}.
\end{split}
\end{equation*}
\end{proof}

For Cartesian products $\{n=1..6\} \times \{W=1..10\}$, and $\{n=7\} \times \{W=1..4\}$, a C++ program directly building the strategies and counting their actions and products has computed the sums of Theorem \ref{TMAbsoluteActionsSampleCovariances} without exceptions corresponding to the formulas A - F. A few illustrations are for your attention.

Theorem \ref{TMAbsoluteActionsSampleCovariances} A, $n=2$, $(W, \sum_{j=1}^{j=S}|U_{1,j}||U_{2,j}|)$: $(1,2)$, $(2,10)$, $(3,28)$, $(4,60)$, $(5,110)$, $(6,182)$, $(7,280)$, $(8,408)$, $(9,570)$, $(10,770)$.

For $n=4$, $W=1$ with the formula letter following the sum value
\[
\begin{array}{cccc}
&U_{2,j}&U_{3,j}&U_{4,j}\\
U_{1,j}&18\;\mathrm{B}&16\;\mathrm{D}&12\;\mathrm{C}\\
U_{2,j}&&22\;\mathrm{E}&16\;\mathrm{D}\\
U_{3,j}&&&18\;\mathrm{B}\\
\end{array}
\]
See also $\mathcal{W}^T$ and $\mathcal{U}^T$ presented earlier for this case . For $n=7$, $W=3$,
\[
\begin{array}{cccccccc}
&U_{2,j}&U_{3,j}&U_{4,j}&U_{5,j}&U_{6,j}&U_{7,j}\\
U_{1,j}&518616\;\mathrm{B}&460992\;\mathrm{D}&460992\;\mathrm{D}&460992\;\mathrm{D}&460992\;\mathrm{D}&345744\;\mathrm{C}\\
U_{2,j}&&643468\;\mathrm{E}&614656;\mathrm{F}&614656\;\mathrm{F}&614656\;\mathrm{F}&460992\;\mathrm{D}\\
U_{3,j}&&&643468\;\mathrm{E}&614656\;\mathrm{F}&614656\;\mathrm{F}&460992\;\mathrm{D}\\
U_{4,j}&&&&643468\;\mathrm{E}&614656\;\mathrm{F}&460992\;\mathrm{D}\\
U_{5,j}&&&&&643468\;\mathrm{E}&460992\;\mathrm{D}\\
U_{6,j}&&&&&&518616\;\mathrm{B}\\
\end{array}
\]
There is an interpretation for remembering location of Formulas A - F. Formula A is applied only for $n=2$. For $4 \le n$, the first row (B, D, ..., D, C) rotates around the "origin" C 90 degrees counterclockwise making a symmetrical right angle. For $n=3$, the angle B-C-B has no D. For $5 \le n$, the second row (E, F, ..., F, D) rotates around the "origin", right most F, 90 degrees counterclockwise also making a symmetrical right angle. For $n=4$, there is no inner angle but single E, see above. Creation of nested angles is repeated until the single E, for even $n$, or last angle E-F-E, for odd $n$. The sum of these $\frac{n(n-1)}{2}$ elements is equal to $2\mathrm{B}+\mathrm{C}+2(n-3)\mathrm{D}+(n-3)\mathrm{E}+\frac{(n-4)(n-3)}{2}\mathrm{F}$. The $n \times n$ matrix is symmetric and the sum of the off-diagonal elements is the double: $4\mathrm{B}+2\mathrm{C}+4(n-3)\mathrm{D}+2(n-3)\mathrm{E}+(n-4)(n-3)\mathrm{F}$. The diagonal is in Equations \ref{EqActionsMoments_1_2}.

\section{Vector Properties of Trading Strategies}

The system of vectors $\mathfrak{U}$ is linearly dependent: one, d.n.s., is $\boldsymbol{0}$ \cite[p. 46, Lemma 14.3]{voevodin1980}, and $1 < n < (2W+1)^{n-1}$ \cite[p. 51, Basis]{voevodin1980}, \cite[p. 14, Theorem 2]{halmos1987}. We can find in $\mathfrak{U}$  a linearly independent system \cite[p. 45, Lemma 14.1]{voevodin1980}.

\begin{lemma}
\label{LAActionsRank}
The rank of the span of $\mathfrak{U}$ is less than $n$.
\end{lemma}

\begin{proof}
$\forall j$, $\sum_{i=1}^{i=n}U_{i,j}=0$. Multiplying it by $0 < P$ yields $P\sum_{i=1}^{i=n}U_{i,j}=\sum_{i=1}^{i=n}PU_{i,j}=\boldsymbol{P}^T\boldsymbol{U}_j=0$. The $n$-dimensional $\boldsymbol{P}$ is orthogonal to \textit{each} $\boldsymbol{U}_j \in \mathfrak{U}$, $\boldsymbol{P} \; \perp \; \mathfrak{U}$ \cite[p. 92 - 94, Orthogonality]{voevodin1980}, and to the span of $\mathfrak{U}$. The rank of the span is less than $n$.
\end{proof}

The $\boldsymbol{P}=(P_1=P, \dots, P_n=P)^T$ is interpreted as a flat price.

\paragraph{Orthogonal vectors of $\mathfrak{U}$.} By Lemma \ref{LAActionsRank}, for $n=2$, all $2W+1$ vectors of $\mathfrak{U}$ are collinear. For $n=3$, all $(2W+1)^2$ vectors of $\mathfrak{U}$ are coplanar with the \textit{orthogonal basis} $\{(1,0,-1)^T,(1,-2,1)^T\}\in \mathfrak{U}$. Let $(*)_{\kappa}$ is a chain empty for $\kappa=0$: $(0)_3=(0,0,0)$, $(1,-1)_2=(1,-1,1,-1)$, $((0)_0,1,(0)_1,-1,(0)_0)=(1,0,-1)$. For $2 \le n$, $\eta=0,\dots,\lfloor \frac{n-2}{2} \rfloor$, the $\lfloor \frac{n-2}{2} \rfloor + 1$ $\eta$-vectors $((0)_{\eta},1,(0)_{n-2-2\eta},-1,(0)_{\eta})^T \in \mathfrak{U}$ are \textit{mutually orthogonal}. For $4 \le n$, $\lambda = 0, \dots, \lfloor \frac{n}{4} \rfloor - 1$, the $\lfloor \frac{n}{4} \rfloor$ $\lambda$-vectors $((0)_{2\lambda},1,-1,(0)_{n-4-4\lambda},-1,1,(0)_{2\lambda})^T \in \mathfrak{U}$ are mutually orthogonal together with the $\eta$-vectors. For $6 \le n$, $\nu=0,...,\lfloor \frac{n-6}{6} \rfloor$, the $\lfloor \frac{n-6}{6} \rfloor + 1$ $\nu$-vectors $((0)_{3\nu},1,-2,1,(0)_{n-6-6\nu},1,-2,1,(0)_{3\nu})^T \in \mathfrak{U}$ is an orthogonal alternative to the $\lambda$-vectors. The $\theta$-vector $((0)_{\lfloor \frac{n-3}{2} \rfloor},1,-2,1,(0)_{\lfloor \frac{n-3}{2} \rfloor})^T \in \mathfrak{U}$ for odd $3\le n=2l+1$, $l=1,\dots$ is orthogonal to the $\eta$-vectors. Examples:
$$
   \begin{array}{rrrrrrrrr}
   n & n-1 & \lfloor \frac{n-2}{2} \rfloor & \lfloor \frac{n}{4} \rfloor & \lfloor \frac{n-6}{6} \rfloor & \eta & \lambda & \nu & \boldsymbol{U}^T \\
   2 & 1 & 0 & n/a & n/a & 0 & n/a & n/a & \eta:(1, -1) \\
   3 & 2 & 0 & n/a & n/a & 0 & n/a & n/a & \eta:(1,0,-1) \\
    &  &  & & &  & & & \theta:(1,-2,1) \\
   4 & 3 & 1 & 1 & n/a & 0 & & n/a & \eta:(1,0,0,-1) \\
    &  &  &  & &1 & & & \eta:(0,1,-1,0) \\
    &  &  &  & & & 0 & & \lambda:(1,-1,-1,1) \\
   5 & 4 & 1 & 1 & n/a & 0 & & n/a & \eta:(1,0,0,0,-1) \\
    &  &  &  & & 1 & & & \eta:(0,1,0,-1,0) \\
    &  &  &  &  & & 0 & & \lambda:(1,-1,0,-1,1) \\
    &  &  & & &  &  & & \theta:(0,1,-2,1,0) \\
   6 & 5 & 2 & 1 & 0 & 0 & & & \eta:(1,0,0,0,0,-1) \\
    &  &  & & & 1 & & & \eta:(0,1,0,0,-1,0) \\
    &  &  & & & 2 & & & \eta:(0,0,1,-1,0,0) \\
    &  &  &  & & & 0 & & \lambda:(1,-1,0,0,-1,1) \\
    &  &  &  & & &  & 0 & \nu:(1,-2,1,1,-2,1) \\
   7 & 6 & 2 & 1 & 0 & 0 & & & \eta:(1,0,0,0,0,0,-1) \\
    &  &  & & & 1 & & & \eta:(0,1,0,0,0,-1,0) \\
    &  &  & & & 2 & & & \eta:(0,0,1,0,-1,0,0) \\
    &  &  &  & & & 0 & & \lambda:(1,-1,0,0,0,-1,1) \\
    &  &  &  & & &  & 0 & \nu:(1,-2,1,0,1,-2,1) \\
    &  &  & & &  &  & & \theta:(0,0,1,-2,1,0,0) \\
   8 & 7 & 3 & 2 & 0 & 0 & & & \eta:(1,0,0,0,0,0,0,-1) \\
    &  &  & & & 1 & & & \eta:(0,1,0,0,0,0,-1,0) \\
    &  &  & & & 2 & & & \eta:(0,0,1,0,0,-1,0,0) \\
    &  &  & & & 3 & & & \eta:(0,0,0,1,-1,0,0,0) \\
    &  &  & & &  & 0 & & \lambda:(1,-1,0,0,0,0,-1,1) \\
    &  &  & & &  & 1 & & \lambda:(0,0,1,-1,-1,1,0,0) \\
    &  &  &  & & &  & 0 & \nu:(1,-2,1,0,0,1,-2,1) \\
   \end{array}
\;
$$
$\{\eta\} \perp \{\lambda\}$, $\{\eta\} \perp \{\nu\}$, $\{\theta\} \perp \{\eta\}$, $\{\theta\} \perp \{\lambda\}$, but $\{\lambda\} \not\perp \{\nu\}$, $\{\theta\} \not\perp \{\nu\}$. 

\paragraph{Rank of $\mathfrak{U}$.} Lemma \ref{LAActionsRank} limits the rank of $\mathfrak{U}$ from above. For $n=2, 3, 4$ it is $n-1=1, 2, 3$. The proofs are the orthogonal bases in $\mathfrak{U}$: $\{\eta:(1,-1)^T\}$, $\{\eta:(1,0,-1)^T$, $\theta:(1,-2,1)^T\}$, $\{\eta:(1,0,0,-1)^T$, $\eta:(0,1,-1,0)^T$, $\lambda:(1,-1,-1,1)^T\}$. For $n=6$, the rank is $n-1=5$: $\{\eta:(1,0,0,0,0,-1)^T$, $\eta:(0,1,0,0,-1,0)^T$, $\eta:(0,0,1,-1,0,0)^T$, $(1,0,-1,-1,0,1)^T$,$(1,-2,1,1,-2,1)^T\}$. The latter two are not $\eta,\lambda,\nu,\theta$-strategies. For $n=8$, the rank is $n-1=7$: $\{\eta:(1,0,0,0,0,0,0,-1)^T$, $\eta:(0,1,0,0,0,0,-1,0)^T$, $\eta:(0,0,1,0,0,-1,0,0)^T$, $\eta:(0,0,0,1,-1,0,0,0)^T$, $(0,1,-1,0,0,-1,1,0)^T$, $(1,-1,-1,1,1,-1,-1,1)^T$, $(1,0,0,-1,-1,0,0,1)^T\}$. The latter four are not $\eta,\lambda,\nu,\theta$-strategies. For $n=5, 7$, $\mathfrak{U}$ contains maximum $n-2=3, 5$ orthogonal vectors. \textit{This is proved by checking all mutually orthogonal combinations using a C++ program.}

Let us notice that $(W_1,W_2,\dots,W_{n-1},0)_j^T \rightarrow (0,W_1,\dots,W_{n-2},W_{n-1})_j^T$ is a cyclic permutation of coordinates and does not change the length of the vector. This is a rotation expressed by $\mathcal{R}\boldsymbol{W}_j$, where the matrix $\mathcal{R}$ is orthogonal
$$
   \mathcal{R} =
   \left [
   \begin{array}{ccccc}
   0 & 0 & \dots & 0 & 1\\
   1 & 0 & \dots & 0 & 0\\
   0 & 1 & \dots & 0 & 0\\
   \dots & \dots & \dots & \dots & \dots\\
   0 & 0 & \dots & 1 & 0\\
   \end{array}
   \right ],
\;
   \mathcal{R}\mathcal{R}^T=
   \left [
   \begin{array}{ccccc}
   1 & 0 & \dots & 0 & 0\\
   0 & 1 & \dots & 0 & 0\\
   0 & 0 & \dots & 0 & 0\\
   \dots & \dots & \dots & \dots & \dots\\
   0 & 0 & \dots & 0 & 1\\
   \end{array}
   \right ] = \mathcal{I}.
$$
An example using the $4 \times 4$ $\mathcal{R}^T$ is found in \cite[p. 69, Exercise 10, matrix $A$]{halmos1987}.
Since $W_{0,j}=W_{n,j}=0$, $\boldsymbol{U}_j=\boldsymbol{W}_j - \mathcal{R}\boldsymbol{W}_j=(\mathcal{I} - \mathcal{R})\boldsymbol{W}_j$, where $\mathcal{I}$ is the identity matrix with ones on the diagonal. Determinant $\det(\mathcal{R})=1$. Thus, the Gramian matrix for $\mathfrak{U}$ is $\mathcal{G}_{\mathcal{U}}=\mathcal{U}^T\mathcal{U}=\mathcal{W}^T(\mathcal{I} - \mathcal{R})^T(\mathcal{I} - \mathcal{R})\mathcal{W}=\mathcal{W}^T(2\mathcal{I} - \mathcal{R}-\mathcal{R}^T)\mathcal{W}$. The square $n\times n$ matrix $2\mathcal{I} - \mathcal{R}-\mathcal{R}^T$ has the main diagonal with twos, sub and super diagonals with -1, and two symmetric corner elements -1. This guarantees that for $2\le n$ sums of rows and columns in the matrix are zero vectors. Applying Theorem 1 about the matrices product rank from \cite[p. 76]{beesack1962} twice, we conclude that $\mathcal{U}^T\mathcal{U}$ has the rank less than $n$. \textit{This is another proof of Lemma \ref{LAActionsRank}.} At the same time, it is exactly $n-1$ for $n=2,3,4,6,8$ and due to $\eta$- and $\lambda$-vectors not less than $\lfloor \frac{n-2}{2} \rfloor + 1 + \lfloor \frac{n}{4} \rfloor$ for $4 < n$.

\begin{theorem}
\label{TMActionsRank}
The rank of $\mathfrak{U}$ is $n - 1$.
\end{theorem}

\begin{proof}
For $2 \le n$ and any $W$, $\mathfrak{U}$ has $n - 1$ strategies with only two ordered transactions: buy one at $1 \le i < n$ followed by sell one at the last $n$th tick:
$$
   \left [
   \begin{array}{cccccc}
   \lceil 1 & 0 & 0 & \dots & 0 \rceil & -1\\
   |0 & 1 & 0 & \dots & 0 | & -1\\
   |0 & 0 & 1 & \dots & 0 | & -1\\
   \dots & \dots & \dots & \dots & \dots & \dots\\
   \lfloor 0 & 0 & 0 & \dots & 1 \rfloor & -1\\
   \dots & \dots & \dots & \dots & \dots & \dots\\
   \end{array}
   \right ].
$$
The top left submatrix of $\mathcal{U}^T$, always obtainable after a suitable rearrangement of rows, is diagonal $(n-1)\times(n-1)$ identity matrix with determinant one. Any greater minor of order $n$ is zero due to Lemma \ref{LAActionsRank}. Then, the highest order of non-zero minor of $\mathcal{U}$ is $n - 1$. This is the rank \cite[p. 132]{voevodin1980}.
\end{proof}

\textit{Theorem \ref{TMActionsRank} means that the span of the trading strategies $\mathfrak{U}$, containing d.n.s., is the hyperplane of the linear space created, spanned, by arriving ticks.}

\paragraph{Buy and hold.} This strategy is a popular investment benchmark. "Holding" is "never selling" a purchased security or real estate. All futures, and many bonds and options expire. Examples of perpetual and long paying interest bonds are the Dutch Water Bonds dated by 1624, the British Consoles issued first time in 1751 \cite{andrews2016}, \cite{moore2014}, some perpetual debt in France \cite{bauneau2014}. The perpetual debt financial instruments is not only a history \cite{atkins2015}, \cite{stanton2017}. The lookback Russian put option \cite{shepp1993}, \cite{duffie1993} with "reduced regret" has no expiration date.

In practice, "never selling" is "holding for a long time". For futures, with well known expiration date and time, "hold" might mean "up to the expiration". For a chain of prices "hold" might mean "until the last tick". Even, if one "holds" or does not sell what has been purchased, one way to estimate its value is to assume that it is sold at the price $P_n$ and cost $C_n$ - mark-to-market or "fair" value. This adds an artificial sell transaction to each valuation tick of interest after buying. All strategies $j$ in $\mathfrak{U}$ exit the market with $W_{n,j}=0$, if they enter it. The d.n.s with $W_{n,d.n.s}=0$ never enters the market.

For comparison with an investment, the buy is assumed coinciding with the beginning of the investment. Here, the single strategy buying at the beginning and marking-to-market at the end $U_{b.a.h.}=(1,0,\dots,0,-1)^T$ is "buy and hold", b.a.h. This is the $\eta$-strategy. The strategies from the proof of Theorem \ref{TMActionsRank} buying at $i$ and selling at the end is "buy, hold, and sell", b.h.s. The b.a.h. is b.h.s but not necessarily vice versa. The system of $n-1$ b.h.s. is the base of $\mathfrak{U}$: it is linearly independent, and any other strategy in $\mathfrak{U}$ is a linear combination of b.h.s. For instance, $(1,-2,1)^T=(1,0,-1)^T-2(0,1,-1)^T$; $(1,-1,0)^T=(1,0,-1)^T-(0,1,-1)^T$. A system of vectors may have several bases. All of them are \textit{equivalent} \cite[pp. 47 - 50]{voevodin1980} and in our case have $n-1$ vectors. The $n-1$ b.h.s. are not mutually orthogonal. Each has the Euclidean length $\sqrt{2}$. The $\mathfrak{U}$ is a subsystem of a linear space over the fields of rational, real, or complex numbers. Its span is a \textit{hyperplane} in one of these spaces with the $n-1$ b.h.s. serving as hyperplane non-orthogonal basis.

\paragraph{Entry-wise operations.} The $\textrm{abs}(\boldsymbol{U}_j)$ is the entry-wise absolute value function. The author did not find a suitable notation to express this. \cite[p. 88]{horn1990}: \textit{"The Hadamard product of two matrices $A=[a_{ij}]$ and $B=[b_{ij}]$ with the same dimensions (not necessarily square) with entires in a given ring is the entry-wise product $A \circ B = [a_{ij}b_{ij}]$, which has the same dimensions as $A$ and $B$."} This is also known as \textit{Schur product} \cite{chandler1962}. The history of names Schur and Hadamard product is in \cite[pp. 92 - 95, Historical remarks]{horn1990}. Entry-wise Hadamard powers and square roots are denoted $A^{\circ 2}, A^{\circ 3}, A^{\circ \frac{1}{2}}$ \cite{reams1999}. If we take the non-negative square root values, then $\textrm{abs}(\boldsymbol{U}_j) = (\boldsymbol{U}_j \circ \boldsymbol{U}_j)^{\circ \frac{1}{2}}$ and
\begin{equation}
\label{EqPLMatrix}
\mathcal{PL}_{q \times S}=-k (\mathcal{P}_{n \times q})^T\mathcal{U}_{n \times S} - (\mathcal{C}_{n \times q})^T(\mathcal{U}_{n \times S} \circ \mathcal{U}_{n \times S})^{\circ \frac{1}{2}},
\end{equation}
where $\mathcal{P}_{n \times q}$ is the price matrix with $q$ scenarios, price column-vectors of size $n$, $\mathcal{U}_{n \times S=(2W+1)^{n-1}}$ is the strategies matrix for the set $\mathfrak{U}$, $\mathcal{PL}_{q \times S}$ is the profit and loss matrix with $S$ columns of size $q$ corresponding to $q$ price scenarios and the set $\mathfrak{U}$, $\mathcal{C}_{n \times q}$ is the cost matrix with $q$ scenarios, cost column-vectors of size $n$. If the cost per share is a fixed non-negative fraction $f$ of price, "equity case", then $(\mathcal{C}_{n \times q})^T=kf(\mathcal{P}_{n \times q})^T$. If the cost per contract is the constant $C$, "futures case", then $(\mathcal{C}_{n \times q})^T=(C\mathcal{J}_{n\times q})^T$, where $\mathcal{J}_{n\times q}$ is the \textit{Hadamard identity matrix} with all elements equal to one. For $q=1$, the "full" matrix form is reduced to $\mathcal{PL}_{1 \times S}$, the row-vector $\boldsymbol{PL}$ of size $S$. This is a sample of $S=(2W+1)^{n-1}$ $PL$ values for $\mathfrak{U}$ and one price scenario $\boldsymbol{P}$.

\section{Means and Variances of Profits and Losses}

The distribution of actions in $\mathfrak{U}$, Equations \ref{EqActionsDistribution}, Figures \ref{FigActionsDistribution1}, \ref{FigActionsDistribution20}, mean $a_1^{PL}$, Equations \ref{EqIndustryGains}, Theorem \ref{TMSampleMeanPL}, do not depend on $\boldsymbol{P}$. A distribution of $PL$ for $\mathfrak{U}$ depends on $\boldsymbol{P}$ and $\boldsymbol{C}$, Equation \ref{EqPLMatrix}. Without losing generality, the $n-1$ basis strategies can be ordered as $\boldsymbol{U}_1^{b.h.s.}=(1,0,\dots,0,-1)^T$, $\boldsymbol{U}_2^{b.h.s.}=(0,1,\dots,0,-1)^T$, $\dots$, $\boldsymbol{U}_{n-1}^{b.h.s.}=(0,0,\dots,1,-1)^T$, where the first $n-1$ coordinates are zeros, except 1 at $i$ and -1 at $n$. Due to these and hyperplane properties of $\mathfrak{U}$, any strategy $\boldsymbol{U}_j=(U_{1,j},U_{2,j},\dots,U_{n-1,j},U_{n,j})^T \in \mathfrak{U}$ in the basis is $\boldsymbol{U}_j=\sum_{i=1}^{i=n-1}U_{i,j}\boldsymbol{U}_i^{b.h.s.}=(U_{1,j},U_{2,j},\dots,U_{n-1,j},-\sum_{i=1}^{i=n-1}U_{i,j})^T$. Multiplication and summation yield correct $n$th coordinate because $\sum_{i=1}^{i=n}U_{i,j} = 0$.

The first component $PL^{I}$ of the $PL$ distribution is values for $j=1,\dots,(2W+1)^{n-1}$: $-k\boldsymbol{P}^T\boldsymbol{U}_j=-k\boldsymbol{P}^T\sum_{i=1}^{i=n-1}U_{i,j}\boldsymbol{U}_i^{b.h.s.}=k\sum_{i=1}^{i=n-1}U_{i,j}(P_n-P_i)$. Their sum is zero, Theorem \ref{TMSampleMeanPL}. The second component $PL^{II}$ of the $PL$ distribution is values for $j=1,\dots,(2W+1)^{n-1}$: $-\boldsymbol{C}^T(\boldsymbol{U}_j \circ \boldsymbol{U}_j)^{\circ \frac{1}{2}}=-\boldsymbol{C}^T(|U_{1,j}|,\dots,|U_{n,j}|)^T$. Their mean sum for constant cost $C$ is $a_1^{PL}$, Equation \ref{EqIndustryGains}. The $PL^{II}$ list has repeated values. Depending on $\boldsymbol{P}$, the $PL^I$ list may have repeated values too.

Thus, each $PL^I$ value is a linear combination of $k(P_n - P_i)$, $i=1,\dots,n-1$. Each $PL^{II}$ value is a \textit{corresponding} linear \textit{absolute} combination of $-2C$. It is either zero for single d.n.s. or \textit{even} negative multiple of $C$. For fixed $\boldsymbol{P}$ and $C$, corresponding values in two lists relate each to other due to integer coefficients of linear combinations. \textit{For $q=1$, Equation \ref{EqPLMatrix} converts a sample distribution of $n$ prices $P_i$ into a sample distribution of $(2W+1)^{n-1}$ values $PL_j$.}

\paragraph{Sample distributions of $P_i$ and $\Delta P_i$.} A chain of $n$ numbers $\{P_1,\dots,P_i,\dots,P_n\}$ creates chains of $n-1$ adjacent absolute differences $\{\Delta P_2=P_2-P_1,\dots,\Delta P_i = P_i-P_{i-1},\dots,\Delta P_n=P_n-P_{n-1}\}$, relative differences $\{\frac{\Delta P_2}{P_1},\dots,\frac{\Delta P_i}{P_{i-1}},\dots,\frac{\Delta P_n}{P_{n-1}}\}$, and popular \textit{log-returns} $\{\ln(\frac{P_2}{P_1}),\dots,\ln(\frac{P_i}{P_{i-1}}),\dots,\ln(\frac{P_n}{P_{n-1}})\}$. The latter two, well defined for $0 < \frac{P_i}{P_{i-1}}$, are close for $\frac{\Delta P_i}{P_{i-1}} \rightarrow 0$. Futures prices are positive. Many theories and speculations are devoted to these chains, when $P_i$ are prices or rates. Some focus on increments. Other pay attention to prices.

Indeed, a Brownian motion is about increments, their independence, Gaussian properties, fundamental proportionality of their variance to elapsed time \cite{bachelier1900}, \cite{einstein1905}, \cite{neftci1996}, \cite{rogers2000}. Its sophisticated combinations are popular in pricing derivatives \cite{neftci1996}, \cite{hull1997}, \cite{hunt2000}, \cite{duffie2001}. In contrast, the trading pattern \textit{head and shoulders} \cite[Figure 3a, pp. 559 - 561]{neftci1991}, \cite[p. 236]{downes1995}, \cite[pp. 74, 76, 108 - 110, 153 - 155]{murphy1999}, \cite[pp. 108 - 110]{kaufman2005} appearing also in coin-tossing experiments, not possessing predictive power \cite[p. 131]{malkiel2007}, cares about price levels in time combinations resembling a top of a human body. For one, trading on such patterns is \textit{astrology}. However, markets are people and programs created by people trading something. If participants "believe" into such matters and trade based on their "beliefs", then a trading feedback can affect markets transforming "beliefs" to reality. The impact should depend on the fraction of trading "believers". How this fraction may form, based on the properties of random \textit{prospects}, is suggested in \cite[p. 33, Hypothesis]{salov2015}. For a pragmatist, a "working pattern" is more important than "why it is working". \textit{Dynkin-Neftci times} decompose a situation on 1) evaluation that an event has occurred using available, "not from the future", information, and 2) gathering statistics how frequently the event is followed by a certain scenario in the past.

The \textit{Ornstein-Uhlenbeck process} \cite{uhlenbeck1930} combines variable's increments with its single value playing a special role. When the variable crosses zero line the drift for increments reverses its sign. The farer from the level, the greater drift magnitude is. Being accompanied by random shocks, it directs the variable back to the level. All continues on the opposite side. This ability to fluctuate around an attractor level is reused by the \textit{mean-reversion} models of interest rates, where the level is shifted from zero to a positive value \cite[pp. 418 - 419]{hull1997}. A care should be taken to avoid rates going to the negative territory similar to the original Bachelier's price assuming Gaussian properties for absolute increments. Cox, Ingersoll, Ross standard deviation of random term proportional to the square root of rate and log-normal properties of latter is one way to ensure positivity \cite[p. 418]{hull1997}. The author has enjoyed the brief and thorough review \cite[p. 271]{jacobsen1996}: \textit{"... his [VS: Laplace's] formidable intuition has led him to a differential equation which is entirely justifiable, and is in fact the Fokker-Plank equation for a one-dimensional Ornstein-Uhlenbeck process, which appears as the weak limit of the Bernoulli-Laplace urn models"}.

The author has formulated the "chicken and egg question" of what is more fundamental prices or their increments \cite{salov2011}, \cite[pp. 34 - 35]{salov2013}. A hybrid approach relies on both. How many price levels should be taken into consideration? Are these levels permanent? One of the problems is non-stationarity of markets \cite[pp. 194]{lipton2001}: \textit{"There are many reasons for considering nonstationary markets, the most obvious of which is that the economic conditions keep changing and this change cannot be adequately captured by stationary models"}. In \cite[p. 35]{salov2013} the author has expressed his view: \textit{"... markets have many modes replacing each other in time, where prices or increments get varying accents"}.

Sample means of prices $a_1^P=\frac{\sum_{i=1}^{i=n}P_i}{n}$ and increments $a_1^{\Delta P}=\frac{\sum_{i=2}^{i=n} \Delta P_i}{n - 1}$ with fixed $P_1$ are interdependent \cite[p. 11, Equation 4]{salov2017}: $a_1^P=P_1+\frac{n^2-1}{n}a_1^{\Delta P} - \frac{\sum_{i=2}^{i=n}i\Delta P_i}{n}$. While each statistics does not depend on the order of sample numbers, together they are bound by the term $\sum_{i=2}^{i=n}i \Delta P_i$, where products $i \Delta P_i$ are sensitive to the order due to the multiplier $i$. We expect that other sample statistics for both sets are interdependent, in general. For sample variances $(S_{n-1}^P)^2=\frac{\sum_{i=1}^{i=n}(P_i-a_1^P)^2}{n-1}$ and $(S_{n-1}^{\Delta P})^2=\frac{\sum_{i=2}^{i=n}(\Delta P_i-a_1^{\Delta P})^2}{n-2}$, we get

\begin{equation}
\label{EqSampleVariances}
\begin{split}
&(S_{n-1}^P)^2=\frac{\sum_{i=1}^{i=n}(P_i-a_1^P)^2}{n-1}=\frac{P_1^2 + \sum_{i=2}^{i=n}P_i^2 - n (a_1^P)^2}{n-1};\\
&(S_{n-1}^{\Delta P})^2=\frac{\sum_{i=2}^{i=n}(\Delta P_i-a_1^{\Delta P})^2}{n-2}=\frac{\sum_{i=2}^{i=n}(\Delta P_i)^2 - (n - 1)(a_1^{\Delta P})^2}{n-2}=\\
&=\frac{P_1^2 - P_n^2 + 2\sum_{i=2}^{i=n}P_i \Delta P_i - (n-1)(a_1^{\Delta P})^2}{n-2};\\
&(S_{n-1}^P)^2=\frac{(n-2)(S_{n-1}^{\Delta P})^2+P_n^2 + \sum_{i=2}^{i=n}(P_i^2 -2P_i \Delta P_i)-n (a_1^P)^2}{n-1}+\\
&+(a_1^{\Delta P})^2.
\end{split}
\end{equation}

\paragraph{Distributions of $PL^I$, $PL^{II}$, and $PL=PL^I+PL^{II}$} can be considered separately. Their sample means depend neither on $P_i$ nor $\Delta P_i$, $i=1,\dots,n$, and, for constant $C$, are $a_1^{PL^I}=0$, $a_1^{PL}=a_1^{PL^{II}}=-\frac{2CW(W+1)(2n-1)}{3(2W+1)}$. Variances of $PL$ and $PL^I$ should depend on $\boldsymbol{P}$. For $S=(2W+1)^{n-1}$, since $|U_{i,j}| \le 2W$, 

\begin{equation}
\label{EqPLIVarianceEstimate}
\begin{split}
&0 \le (S_{n-1}^{PL^I})^2=\frac{\sum_{j=1}^{j=S}(PL_j^I - a_1^{PL^{I}})^2}{S-1}=\frac{\sum_{j=1}^{j=S}(PL_j^I)^2}{S-1}=\\
&=\frac{\sum_{j=1}^{j=S}(-k\sum_{i=1}^{i=n}P_i U_{i,j})^2}{S-1} \le \frac{4k^2 W^2 S}{S-1} (\sum_{i=1}^{i=n}P_i)^2 < 4k^2 W^2 (\sum_{i=1}^{i=n}P_i)^2.
\end{split}
\end{equation}
It is known and easy to prove using mathematical induction that $(\sum_{i=1}^{i=n}a_i)^2=\sum_{i=1}^{i=n}a_i^2+2\sum_{l=1}^{l=n}\sum_{i=1}^{i=l-1}a_i a_l=\sum_{i=1}^{i=n}a_i^2+2\sum_{l=1}^{l=n} a_l\sum_{i=1}^{i=l-1}a_i=\sum_{i=1}^{i=n}a_i^2+2\sum_{l=2}^{l=n} a_l\sum_{i=1}^{i=l-1}a_i=\sum_{i=1}^{i=n}a_i^2+2\sum_{l=1}^{l=n-1}a_l\sum_{i=l+1}^{i=n}a_i$. With Equations \ref{EqActionsMoments_1_2},
\begin{equation*}
\begin{split}
&(S_{n-1}^{PL^I})^2=\frac{\sum_{j=1}^{j=S}(-k\sum_{i=1}^{i=n}P_i U_{i,j})^2}{S-1}=\\
&=\frac{k^2}{S-1} \sum_{j=1}^{j=S} \left ( \sum_{i=1}^{i=n} P_i^2 U_{i,j}^2 + 2 \sum_{l=2}^{l=n} P_l U_{l,j} \sum_{i=1}^{i=l-1}P_i U_{i,j} \right )=\\
&=\frac{k^2}{S-1} \left ( \sum_{i=1}^{i=n}P_i^2 \sum_{j=1}^{j=S} U_{i,j}^2 + 2 \sum_{j=1}^{j=S} \sum_{l=2}^{l=n} P_l U_{l,j} \sum_{i=1}^{i=l-1}P_i U_{i,j}  \right )=\\
&=\frac{k^2}{S-1} \left ( \frac{W(W+1)S}{3} \left ( 2\sum_{i=1}^{i=n}P_i^2 - P_1^2 - P_n^2 \right ) + 2 \sum_{j=1}^{j=S} \sum_{l=1}^{l=n-1} P_l U_{l,j} \sum_{i=l+1}^{i=n}P_i U_{i,j} \right ).
\end{split}
\end{equation*}
In the first summand, $\sum_{i=1}^{i=n}P_i^2$ is the square of the price vector $\boldsymbol{P}$ length. The second summand can be expressed as the double sum $\sum_{j=1}^{j=S}$ of $\frac{n(n-1)}{2}$ terms

$$
   2 \sum_{j=1}^{j=S}
   \left (
   \begin{array}{cccccccc}
   P_1P_2U_{1,j}U_{2,j} & + & P_1P_3U_{1,j}U_{3,j} & + & \dots & + & P_1P_nU_{1,j}U_{n,j}&+\\
    & + & P_2P_3U_{2,j}U_{3,j} & + & \dots & + & P_2P_nU_{2,j}U_{n,j}&+\\
   & & & + & \dots & + & \dots&+\\
   & & & & & + & P_{n-1}P_nU_{n-1,j}U_{n,j}&\\
   \end{array}
\right ).
$$
By Theorem \ref{TMActionsSampleCovariances}, all sums above the bottom diagonal are zeros. After summation by $j$, the diagonal terms get the common multiplier $-2\frac{W(W+1)(2W+1)^{n-1}}{3}$, and
\begin{equation}
\label{EqPLIVariance}
\begin{split}
&(S_{n-1}^{PL^I})^2=\frac{k^2W(W+1)S}{3(S-1)} \left (2\sum_{i=1}^{i=n}P_i^2-2\sum_{i=1}^{i=n-1}P_i P_{i+1} - P_1^2 - P_n^2 \right )=\\
&=\frac{k^2W(W+1)(2W+1)^{n-1}}{3((2W+1)^{n-1}-1)}\sum_{i=2}^{i=n}(\Delta P_i)^2.
\end{split}
\end{equation}

The variance of $PL^{II}$ for constant $C$ is
\begin{equation*}
\begin{split}
&(S_{n-1}^{PL^{II}})^2=\frac{\sum_{j=1}^{j=S}(PL_j^{II}-a_1^{PL^{II}})^2}{S-1}=\frac{\sum_{j=1}^{j=S}(PL_j^{II})^2-S(a_1^{PL^{II}})^2}{S-1}=\\
&=\frac{C^2\sum_{j=1}^{j=S}(\sum_{i=1}^{i=n}|Ui,j|)^2-S(a_1^{PL^{II}})^2}{S-1}, \; \mathrm{where}\\
&\sum_{j=1}^{j=S} (\sum_{i=1}^{i=n}|Ui,j|)^2= \sum_{j=1}^{j=S} \left ( \sum_{i=1}^{i=n}U_{i,j}^2+2\sum_{l=1}^{l=n}|U_{l,j}|\sum_{i=1}^{i=l-1}|U_{i,j}| \right )=\\
&=\sum_{i=1}^{i=n}\sum_{j=1}^{j=S}U_{i,j}^{2}+2\sum_{l=1}^{l=n}\sum_{i=1}^{i=l-1}\sum_{j=1}^{j=S}|U_{l,j}||U_{i,j}|.
\end{split}
\end{equation*}
Equations \ref{EqActionsMoments_1_2} "evaluate" the left term-sum: two equal sums for $i=1,n$, plus $n-2$ equal intermediate sums for $i=2,\dots, n-1$: $\frac{2}{3}W(W+1)(2W+1)^{n-1}+\frac{(n-2)2}{3}W(W+1)(2W+1)^{n-1}=\frac{2(n-1)}{3}W(W+1)(2W+1)^{n-1}$. Theorem \ref{TMAbsoluteActionsSampleCovariances} "evaluates" the right term-sum: $2\mathrm{A}$ for $n=2$, or $4\mathrm{B}+2\mathrm{C}+4(n-3)\mathrm{D}+2(n-3)\mathrm{E}+(n-4)(n-3)\mathrm{F}$ for $3 \le n$. Thus,
\begin{equation}
\label{EqPLIIVariance}
\begin{split}
&n = 2: \; (S_{n-1}^{PL^{II}})^2=\frac{4}{3}W^2(W+1)^2(2W+1)^2; \; 3 \le n: \; (S_{n-1}^{PL^{II}})^2=\\
&\frac{4C^2W(W+1)(2W+1)^{n-3}(6n(2W^2+2W+1)-11W(W+1)-3)}{45((2W+1)^{n-1}-1)}.
\end{split}
\end{equation}
\begin{theorem}
\label{TMPLVariance}
For $\mathfrak{U}$, $(S_{n-1}^{PL})^2$ = $(S_{n-1}^{PL^I})^2 + (S_{n-1}^{PL^{II}})^2$.
\end{theorem}

\begin{proof}
$PL_j=PL_j^{I}+PL_j^{II}=-k\sum_{i=1}^{i=n}P_iU_{i,j}-\sum_{i=1}^{i=n}C_i|U_{i,j}|$. $a_1^{PL}=a_1^{PL^{I}}+a_1^{PL^{II}}=0+a_1^{PL^{II}}=a_1^{PL^{II}}$. The $a_1^{PL^{II}}$ does not assume a particular constant case $C_i=C$ but more general vector $\boldsymbol{C}$.

\begin{equation*}
\begin{split}
&(S_{n-1}^{PL})^2=\frac{\sum_{j=1}^{j=S}(PL_j-a_1^{PL})^2}{S-1}=\frac{\sum_{j=1}^{j=S}(PL_j)^2-S(a_1^{PL})^2}{S-1}=\\
&=\frac{\sum_{j=1}^{j=S}(PL_j^I+PL_j^{II})^2-S(a_1^{PL^{II}})^2}{S-1}=\\
&=\frac{\sum_{j=1}^{j=S}[(PL_j^I)^2+2PL_j^IPL_j^{II} +(PL_j^{II})^2]-S(a_1^{PL^{II}})^2}{S-1}=\\
&=(S_{n-1}^{PL^I})^2 + (S_{n-1}^{PL^{II}})^2+2\frac{\sum_{j=1}^{j=S}PL_j^IPL_j^{II}}{S-1}, \; \mathrm{where}\\
&\sum_{j=1}^{j=S}PL_j^IPL_j^{II}=-k\sum_{j=1}^{j=S}(P_1U_{i,j}+\dots+P_nU_{i,j})(C_1|U_{i,j}|+\dots+C_n|U_{i,j}|).
\end{split}
\end{equation*}
Opening brackets under the sum yields the terms $-k\sum_{j=1}^{j=S}P_iU_{i,j}C_l|U_{l,j}|=\\$
$-kP_iC_l\sum_{j=1}^{j=S}U_{i,j}|U_{l,j}|$. However, $\forall_{i,l}$, $\sum_{j=1}^{j=S}U_{i,j}|U_{l,j}|=0$. Indeed, $\mathfrak{U}$ consists of a d.n.s. and pairs $(\boldsymbol{U}_j,\boldsymbol{U}_{j'}=-\boldsymbol{U}_j)$, strategies and their "mirror reflections" in the index $i$, time, axis. $U_{i,d.n.s.}|U_{l,d.n.s.}|=0$. In any pair, $U_{i,j}|U_{l,j}|+(-U_{i,j})|-U_{l,j}|=(U_{i,j}-U_{i,j})|U_{l,j}|=0$.
\end{proof}

\section{Algebraic properties of trading positions}

Similar to $\mathfrak{U}$, the set of positions $\mathfrak{W}$ consists of the \textit{do nothing position}, d.n.p, $\boldsymbol{W}_{d.n.p.}=(0, \dots, 0)^T$ and pairs of mirror reflections $(\boldsymbol{W}_j,-\boldsymbol{W}_j)$ in index $i$, time, axis. Let us define on $\mathfrak{W}$ the binary operation denoted $\oplus_W$, a pairwise arithmetic addition of coordinates $\boldsymbol{W}_j \oplus_W \boldsymbol{W}_l = (W_{1,j} \oplus_W W_{1,l}, \dots, W_{n,j} \oplus_W W_{n,l})^T$, so that each coordinate sum $ > W$ is replaced with $W$ and $< -W$ with $-W$. This, otherwise ordinary addition, ensures that for any pair of position vectors the vector-result belongs to $\mathfrak{W}$, the \textit{closure property}.

Following to Cayley \cite[p. 41]{cayley1854}, \cite[pp. 144 - 153]{cayley1889}, we illustrate the operation for $W=3$ using the table, named today after him, for coordinates of positions
$$
\begin{array}{rrrrrrrr}
 \oplus_3 |&-3&-2&-1&0&1&2&3\\
 --&--&--&--&--&--&--&--\\
-3|&\mathit{-3}&\mathit{-3}&\mathit{-3}&\boldsymbol{-3}&\boldsymbol{-2}&\boldsymbol{-1}&\boldsymbol{0}\\
-2|&\mathit{-3}&\mathit{-3}&\boldsymbol{-3}&\boldsymbol{-2}&\boldsymbol{-1}&\boldsymbol{0}&\boldsymbol{1}\\
-1|&\mathit{-3}&\boldsymbol{-3}&\boldsymbol{-2}&\boldsymbol{-1}&\boldsymbol{0}&\boldsymbol{1}&\boldsymbol{2}\\
0|&\boldsymbol{-3}&\boldsymbol{-2}&\boldsymbol{-1}&\boldsymbol{0}&\boldsymbol{1}&\boldsymbol{2}&\boldsymbol{3}\\
1|&\boldsymbol{-2}&\boldsymbol{-1}&\boldsymbol{0}&\boldsymbol{1}&\boldsymbol{2}&\boldsymbol{3}&\mathit{3}\\
2|&\boldsymbol{-1}&\boldsymbol{0}&\boldsymbol{1}&\boldsymbol{2}&\boldsymbol{3}&\mathit{3}&\mathit{3}\\
3|&\boldsymbol{0}&\boldsymbol{1}&\boldsymbol{2}&\boldsymbol{3}&\mathit{3}&\mathit{3}&\mathit{3}\\
\end{array}
$$
The first, left, and second, right, elements are selected by column and row. The result is on the row and column intersection.
Italic numbers show "underflows" $-3 \oplus_3 -2 = \mathit{-3}$ and "overflows" $1 \oplus_3 3 = \mathit{3}$. The bold numbers correspond to usual addition of integers. The number of table entries, pairs of the Cartesian product $[-W,W] \times [-W,W]$, is $(2W+1)^2$. The number of "underflows" and "overflows" is $W(W+1)$. The number of ordinary additions is $(2W+1)^2 - W(W+1)=3W^2+3W+1$. For $1 \le W$, the greatest $\frac{3W^2+3W+1}{(2W+1)^2} = \frac{7}{9}$ is for $W=1$. $\lim_{W \rightarrow \infty}\frac{3W^2+3W+1}{(2W+1)^2}=\frac{3}{4}$. The ratio monotonically decreases with the growing $W$ because of the negative "derivative" $\frac{-1}{(2W+1)^3}$. The "drop" to the asymptotic level is $\frac{7}{9} - \frac{3}{4}=\frac{1}{36}$.

The operation is \textit{commutative}: $\boldsymbol{W}_j \oplus_W \boldsymbol{W}_l = \boldsymbol{W}_l \oplus_W \boldsymbol{W}_j$. Its, symmetric with respect to the main diagonal, Cayley table shows this well. Such tables are "less friendly" for conclusions about \textit{associativity} requiring three elements and two sequential operations. For some values, associativity holds: $(1 \oplus_1 1) \oplus_1 1 = 1 \oplus_1 (1 \oplus_1 1) = 1$. \textit{In general, $\oplus_W$ is not associative}: $(1 \oplus_1 1) \oplus_1 -1 = 0$ but  $1 \oplus_1 (1 \oplus_1 -1) = 1$. An example of a \textit{commutative not associative operation} is the mean: $\frac{a+b}{2}$. The game \textit{Rock–Paper–Scissors} also illustrates a commutative not associative operation: $(RP)S=S$, $R(PS)=R$, $RP=PR=P$, etc.

$\forall \; \boldsymbol{W}\in \mathfrak{W}$, $\boldsymbol{W} \oplus_W \boldsymbol{W}_{d.n.p.}=\boldsymbol{W}_{d.n.p.} \oplus_W \boldsymbol{W} = \boldsymbol{W}$. Hence, $\boldsymbol{W}_{d.n.p.}$ in $\mathfrak{W}$ is the \textit{two-sided identity element} or simply \textit{identity} \cite[p. 67]{malcev1973}. Every $\boldsymbol{W} \in \mathfrak{W}$ is \textit{invertible} with the unique \textit{inverse element} $-\boldsymbol{W} \in \mathfrak{W}$: $\boldsymbol{W} \oplus_W -\boldsymbol{W} = -\boldsymbol{W} \oplus_W \boldsymbol{W} = \boldsymbol{W}_{d.n.p.}$. The identity is own inverse.

There might be several solutions of $\boldsymbol{W}_j \oplus_W \boldsymbol{X} = \boldsymbol{W}_l$ or $W_{i,j} \oplus_W X_i = W_{i,l}$. If $W_{i,l}=\pm W$, then, depending on $W_{i,j}$, several $X_i$ can be good: 1) $2 \oplus_3 x = 3$, $x=1,2,3$; 2) $1 \oplus_3 x = 3$, $x=2,3$; 3) $-2 \oplus_3 x = -3$, $x=-3,-2,-1$. Choosing a solution by absolute minimum value ensures uniqueness: 1) 1; 2) 2; 3) -1.

Not every equation $W_{i,j} \oplus_W X_i = W_{i,l}$ has a solution: $-2 \oplus_3 x = 3$ requires \textit{forbidden} $x=5 \not\in [-3, 3]$. In spite of unique invertibility of every $\boldsymbol{W} \in \mathfrak{W}$, the equivalent, due to commutativity of $\oplus_W$, equations $W_{i,j} \oplus_W X_i = W_{i,l}$ and $X_i \oplus_W W_{i,j} = W_{i,l}$ have no, or $|W_{i,j}|+1$ (one or several) solutions:
$$
\begin{array}{ccccc}
&&\mathrm{DIAGRAM}&&\\
&&W_{i,j} \oplus_W X_i = W_{i,l}&&\\
|W_{i,l} - W_{i,j}| >W&\swarrow&&\searrow&|W_{i,l} - W_{i,j}| \le W\\
\{X_i\}=\O&&&&\{X_i\} \ne \O\\
&&|W_{i,l}|<W&\swarrow&\downarrow \; |W_{i,l}|=W\\
&&\mathrm{unique \; solution}&&|W_{i,j}|+1 \; \mathrm{solutions}:\\
&&X_i = W_{i,l}-W_{i,j}&&X_i \in [W-W_{i,j}, W]\\
&&&&\mathrm{for} \; W_{i,l}=W;\\
&&&&X_i \in [-W, -W-W_{i,j}]\\
&&&&\mathrm{for} \; W_{i,l}=-W.\\
\end{array}
$$
With several solutions $X_i$, $X_i=W_{i,l}-W_{i,j}$ has the least absolute value.

\textit{Financial sense behind $(\mathfrak{W}, \oplus_W)$ is: applying strategies to an account and single futures type the sums of corresponding positions cannot exceed by absolute value a level determined by margin requirements and/or position limits}.

\paragraph{Classification.} Due to the closure, the \textit{algebraic structure} $(\mathfrak{W}, \oplus_W)$ is a \textit{magma} \cite[p. 1, LAWS OF COMPOSITION, Definition 1]{bourbaki1974}: \textit{not associative} \cite[p. 4, ASSOCIATIVE LAWS, Definition 5]{bourbaki1974},  \textit{commutative} \cite[p. 7, PERMUTABLE ELEMENTS, COMMUTATIVE LAWS, Definitions 7, 8]{bourbaki1974}, \textit{initial}, because has the identity \cite[p. 12, IDENTITY ELEMENT; CANCELLABLE ELEMENTS; INVERTIBLE ELEMENTS, Definition 2]{bourbaki1974}, and with a unique \textit{inverse element} for each element in $\mathfrak{W}$. A term, interchangeably applied with magma, is \textit{groupoid} \cite[p. 67]{malcev1973}, \cite[p. 90]{rosenfeld1968}, \cite[p. 6, Definition 1]{belousov1967}, \cite[p. 1]{bruck1971}, \cite[p. 1]{sabinin1999}.

If $W_{i,j} \oplus_W X_i = W_{i,l}$ with the commutative $\oplus_W$ would have a unique solution $X_i$ for any pair $(W_{i,j}, W_{i,l})$ of the Cartesian product $[-W, W] \times [-W, W]$, then the groupoid $(\mathfrak{W}, \oplus_W)$ would be a \textit{quasigroup} \cite[p. 72]{malcev1973}, \cite[p. 9]{bruck1971}, \cite[p. 6, Definition 1]{belousov1967}, \cite[p. 23, 1.3. Definition]{sabinin1999}. Moreover, since $(\mathfrak{W}, \oplus_W)$ has the identity $\boldsymbol{W_{d.n.p.}}$, it would be a \textit{loop} \cite[p. 73]{malcev1973}, \cite[p. 15]{bruck1971}, \cite[p. 8, Definition 4]{belousov1967}, \cite[p. 24, 1.6. Definition]{sabinin1999}. For completeness, an associative loop is a \textit{group} and associative commutative loop is an \textit{Abelian group}. \textit{Our "loop" differs.}

While uniqueness of $X_i$ is achieved by selecting from a finite set of solutions the one with the absolute minimum, not every pair $(W_{i,j}, W_{i,l})$ has a solution: pairs for which $|W_{i,l} - W_{i,j}|>W$ require forbidden values $\not \in [-W, W]$.

Malcev \cite{malcev1973}, Belousov \cite{belousov1967}, and Sabinin \cite{sabinin1999} define quasigroup not only as a groupoid with a need to solve a system of two (or one for two-sided case) equations but alternatively as an \textit{algebra}. The latter includes the set, the main binary operation, and two (or one for two-sided case) binary inverse operations. In our case, the algebraic structure includes: the set $\mathfrak{W}$, the \textit{total} (defined for all pairs of elements) not associative commutative binary operation $\oplus_W$, the identity element $\boldsymbol{W}_{d.n.p.}$, the inverse element $-\boldsymbol{W}$ for each element $\boldsymbol{W}$. It can be added a \textit{partial} inverse binary operation $\ominus_W$. The adjective "partial" has the traditional meaning: "defined for some but not all pairs of elements". \textit{The operation $\oplus_W$ is total. The inverse operation $\ominus_W$ is partial.} Malcev illustrates such a possibility using the set of natural numbers including zero, binary arithmetic addition defined for each pair of numbers, and partial binary arithmetic subtraction defined only for pairs $(a, b)$, where $a \ge b$ \cite[p. 30]{malcev1973}.

The Cayley table for the coordinate $i$ subtraction $W_{i,l} \ominus_W W_{i,j}$ is
$$
\begin{array}{rrrrrrrr}
 \ominus_3|&-3&-2&-1&0&1&2&3\\
 --&--&--&--&--&--&--&--\\
-3|&\boldsymbol{0}&\boldsymbol{-1}&\boldsymbol{-2}&\boldsymbol{-3}&\mathrm{n/a}&\mathrm{n/a}&\mathrm{n/a}\\
-2|&\boldsymbol{1}&\boldsymbol{0}&\boldsymbol{-1}&\boldsymbol{-2}&\boldsymbol{-3}&\mathrm{n/a}&\mathrm{n/a}\\
-1|&\boldsymbol{2}&\boldsymbol{1}&\boldsymbol{0}&\boldsymbol{-1}&\boldsymbol{-2}&\boldsymbol{-3}&\mathrm{n/a}\\
0|&\boldsymbol{3}&\boldsymbol{2}&\boldsymbol{1}&\boldsymbol{0}&\boldsymbol{-1}&\boldsymbol{-2}&\boldsymbol{-3}\\
1|&\mathrm{n/a}&\boldsymbol{3}&\boldsymbol{2}&\boldsymbol{1}&\boldsymbol{0}&\boldsymbol{-1}&\boldsymbol{-2}\\
2|&\mathrm{n/a}&\mathrm{n/a}&\boldsymbol{3}&\boldsymbol{2}&\boldsymbol{1}&\boldsymbol{0}&\boldsymbol{-1}\\
3|&\mathrm{n/a}&\mathrm{n/a}&\mathrm{n/a}&\boldsymbol{3}&\boldsymbol{2}&\boldsymbol{1}&\boldsymbol{0}\\
\end{array}
$$
The operation $\ominus_W$ is not commutative (the table is antisymmetric with respect to zero diagonal), not associative (for instance, $(1 \ominus_3 (-1)) \ominus_3 (-1)=3$ but $1 \ominus_3 ((-1)) \ominus_3 (-1)=1$), partial. The number of pairs for which the result is not available is $W(W+1)$. The total number of pairs is $(2W+1)^2$. The number of pairs with the defined subtraction is the difference $3W^2+3W+1$.

The terms \textit{partial magma}, \textit{partial loop} are applied in cases, where the major operation is partial \cite{nesterov2006}. While the main properties of the algebraic system $(\mathfrak{W}, \oplus_W, \ominus_W)$ defined on the finite set of trading positions $\mathfrak{W}$ with total binary not associative commutative addition $\oplus_W$, and partial not associative, not commutative subtraction $\ominus_W$, including domain pairs counting, are described, the author feels uncomfortable to name it a \textit{partial loop} because $\oplus_W$ is total and viewed as the main operation.

Non-associativity of the main algebraic operation creates a link to works on non-associative algebras including the contribution of Etherington \cite{etherington1949}. Their focus is on different and often more complicated algebraic structures than one discussed. Following to the Schafer's remark \cite[p. 1]{schafer1961}, emphasizing that \textit{nonassociative algebra} does not assume associativity, while \textit{not associative algebra} means that associativity is not satisfied, we say \textit{not associative} $\oplus_W$.

\section{Algebraic properties of trading strategies}

Let $\boldsymbol{W}=\boldsymbol{W}_j \oplus_W \boldsymbol{W}_l$. By Theorem \ref{TMStrategiesPositionsBijection}, with $W_0=W_{0,j}=W_{0,l}=0$, $\boldsymbol{W} \leftrightarrow \boldsymbol{U}$, $\boldsymbol{W}_j \leftrightarrow \boldsymbol{U}_j$, $\boldsymbol{W}_l \leftrightarrow \boldsymbol{U}_l$. What is a corresponding binary operation $\boldsymbol{U}=\boldsymbol{U}_j \circ \boldsymbol{U}_l$?

The following recipe exists. Using $W_{i,j}=\sum_{r=1}^{r=i}U_{r,j}$ and $W_{i,l}=\sum_{r=1}^{r=i}U_{r,l}$, convert $\boldsymbol{U}_j$ to $\boldsymbol{W}_j$ and $\boldsymbol{U}_l $ to $\boldsymbol{W}_l$. Then, "add" positions $\boldsymbol{W}=\boldsymbol{W}_j \oplus_W \boldsymbol{W}_l$ and applying adjacent difference convert $\boldsymbol{W}$ to $\boldsymbol{U}$. The latter should be recognized as the result of $\boldsymbol{U}_j \circ \boldsymbol{U}_l$.

We have seen from distributions of positions and actions that, while positions in steps $i-1$ and $i$ are combined independently and uniformly,  for actions it is not so because they have to ensure that positions are within the limits. The author believes that it is impossible in a general case $1 < i < n$ to compute $U_i$ given $U_{i,j}$ and $U_{i,l}$. Information from step $i-1$ is needed.

For $\oplus_W$, the Cayley table is \textit{antisymmetric}, $\times (-1)$, with respect to the second zero diagonal.

\begin{theorem}
\label{TMOperationAntisymmetric}
$\forall \; a, b \in [-W, W]$, $-(a \oplus_W b)=(-a) \oplus_W (-b)$.
\end{theorem}

\begin{proof}
Multiplying both $a$ and $b$ by $-1$ corresponds to the reflection in the second zero diagonal of the Cayley table. But the table is antisymmetric, $\times (-1)$, with respect to this reflection.
\end{proof}

We can write $U_i=W_i-W_{i-1}=(W_{i,j} \oplus_W W_{i,l})-(W_{i-1,j} \oplus_W W_{i-1,l})=([W_{i-1,j} + U_{i,j}] \oplus_W [W_{i-1,l} + U_{i,l}])-(W_{i-1,j} \oplus_W W_{i-1,l})$. For $i=1$, this yields $U_1=U_{1,j} \oplus_W U_{1,l}$ and $\circ \equiv \oplus_W$. For $i=n$, $U_{n,j}=0-W_{n-1,j}=-W_{n-1,j}$, $U_{n,l}=0-W_{n-1,l}=-W_{n-1,l}$, and $U_n=-(W_{n-1,j} \oplus_W W_{n-1,l})=\{$ by Theorem \ref{TMOperationAntisymmetric} $\}=((-W_{n-1,j}) \oplus_W (-W_{n-1,l}))=U_{1,j} \oplus_W U_{1,l}$ and $\circ \equiv \oplus_W$. Thus, the coordinate wise operation of $\boldsymbol{U}=\boldsymbol{U}_j \circ \boldsymbol{U}_l$ for $1 \le i \le n$ is
\begin{equation}
\label{EqActionsOperation}
U_i=([W_{i-1,j} + U_{i,j}] \oplus_W [W_{i-1,l} + U_{i,l}])-(W_{i-1,j} \oplus_W W_{i-1,l}).
\end{equation}

\section{The maximum profit trading strategies}

Today, for given $\boldsymbol{P}$, $\boldsymbol{C}$, $W=1$, computing $3^{134908}$ $PL$ values by Equation \ref{EqPL}, in order to select the maximum profit strategy, MPS, is impossible. A \textit{quantum computer} \cite{feynman1982}, \cite{deutsch1985} would need $\lceil \log_2(3^{134908}) \rceil = \lceil 134908 \frac{\ln(3)}{\ln(2)} \rceil = 213825$ \textit{qbits} to represent the corresponding \textit{coherent superposition quantum states}. Plus, quantum algorithms are required \cite{shor1994}, \cite{shor2000}. Recent successes are 2000 qbits \textit{D Wave 2000Q} computer for \textit{annealing simulation}, an \textit{analog computer}, \cite{brad2017}, and 51 qbits \textit{generic computational device} created at Harvard University \cite{fossbytes2017}, \cite{reynolds2017}.

Since d.n.s. with $PL=0$ is available, MPS cannot lose. In \cite{salov2007}, the author has developed the l- and r-algorithms (left and right) with linear complexity $O(n)$. This is faster than \textit{genetic algorithms} \cite{goldberg1989}, which do not guarantee maximum. Given $\boldsymbol{P}$, $\boldsymbol{C}$, and $W \in \mathbb{N}$, it returns $\boldsymbol{U}$ with maximum $PL$. Without loosing generality, $W=1$. This strategy, denoted MPS0 and not reinvesting profits, is a foundation for MPS1 and MPS2 reinvesting them. Similar to MPS0, MSP1 reverses long and short positions. In contrast with MPS0, MPS1 adds to positions while switching, if initial and maintenance futures margins permit. MPS2 extracts the absolute maximum reinvesting immediately, if it is profitable. Discrete MPS0, MPS1, MPS2 have been studied \cite{salov2007}, \cite{salov2008}, \cite{salov2011},  \cite{salov2011b}, \cite{salov2012}, \cite{salov2013}, \cite{salov2017}. MPS0, MPS1, MPS2 are \textit{objective} market properties. Not all elements of MPS are \textit{Markov times}.

\paragraph{"Markov time".} Neftci is, probably, first who applied Markov times to formalize Technical Analysis \cite[p. 553]{neftci1991}: \textit{"... one contribution this article makes is to recognize the importance of Markov times as a tool to pick well-defined rules for issuing signals at market turning points. Let $\{X_t\}$ be an asset price ... Let $\{I_t\}$ be the sequence of information sets (sigma-algebras) generated by the $X_t$ and possibly by other data observed up to time $t$. ... a random variable $\tau$ is a Markov time if the event $A_t=\{\tau \le t\}$ is $I_t$-measurable - that is, whether or not $\tau$ is less than $t$ can be decided given $I_t$"}. Giles says \cite[p. 175]{giles2000}: \textit{"... Neftci's Markov Times approach..."}. This does not mean that Neftci introduced "Markov times" but suggested an approach using them. The definition of a Markov time is in the first English edition of the Shiryaev's textbook \cite{shiryaev1984}, cited by Neftci \cite[p. 556, Theorem]{neftci1991} with a typo: Springer's year is 1984 but not 1985. The primary source defines "Markov time" \cite[p. 469, Definition 3, Russian 1979]{shiryaev1984}. \textit{Who introduced the term "Markov time"? Which Markov?}

Let us review two phrases. 1) Carl Boyer about "Bernoulli": \textit{"No family in the history of mathematics has produced as many celebrated mathematicians as did the Bernoulli family ..."} \cite[p. 415]{boyer1989}, \cite{salov2014}. 2) August Wilhelm von Hofmann about "Nikolay Nikolaevich Zinin": \textit{"If Zinin has nothing more than to convert nitrobenzene into aniline, even then his name should be inscribed in golden letters in the history of chemistry"} \cite{hofmann1880}. Admirers of American indigo blue, 2,2'-Bis(2,3-dihydro-3-oxoindolyliden) $C_{16}H_{10}N_2O_2$, jeans are indebted to Zinin for synthesis of benzeneamine $C_6H_5NH_2$. It is well known that Zinin was a private teacher of chemistry to young Alfred Nobel. It is less frequently cited that Zinin had brilliant mathematical skills remarked by astronomer Ivan Mikhailovich Simonov and geometer Nikolay Ivanovich Lobachevsky. Zinin was their pupil at the mathematics branch of the philosophical department of the Kazan University, 1830 - 1833. His graduation thesis "Perturbation Theory" written on "On perturbations of elliptic motions of planets", the topic suggested by Simonov, was awarded by the golden medal \cite[p. 25 - 33]{gumilevskii1965}.

Boyer's phrase reminds about celebrities with non-unique names. Andrey Andreevich Markov, father (06/2/(14)/1856 - 07/20/1922), Vladimir Andreevich Markov, younger brother of father (05/07(19)/1871 - 01/18/(30)/1897), Andrey Andreevich Markov, son of father (09/09(22)/1903 - 10/11/1979), and Alexander Alexandrovich Markov, related by profession (03/24/1937 - 10/23/1994) are first class mathematicians. Neglect of history is the road to misunderstanding: \textit{"Markov chains"} honors the father and \textit{"Markov algorithms"} is about the son's contribution \cite{kushner2006}. To emphasize achievements of Markov father, the author rephrases Hofmann's words and believes that chemist Zinin, being also the first class mathematician, would agree: \textit{"If Markov [father] has nothing more than to create chains named after him, even then his name should be inscribed in golden letters in the history of mathematics"}.

Markov: \textit{"In my opinion, the cases of variables linked into a chain so that when the value of one of them becomes known, subsequent variables appear independent on the variables preceding it, deserve attention"} \cite[p. 365, VS's translation]{markov1951}. Interesting facts about Markov are in Oscar Borisovich Sheynin's \cite{sheynin1989}, \cite{sheynin2007}. From \textit{"Theory of Probability. An Elementary Treatise against a Historical Background"} (English and Russian manuscripts kindly provided by O.B.S. to the author in an email), the author has known: \textit{"'Markov chain' first appeared (in French) in 1926 (Bernstein 1926, first line of § 16)"}, \cite{sheyninUnpublished}. 1926 and 1927 are the years of submission and publication of \cite[p. 40]{bernstein1927}, Figure \ref{FigBernstein}.
\begin{figure}[!h]
  \centering
  \includegraphics[width=115mm]{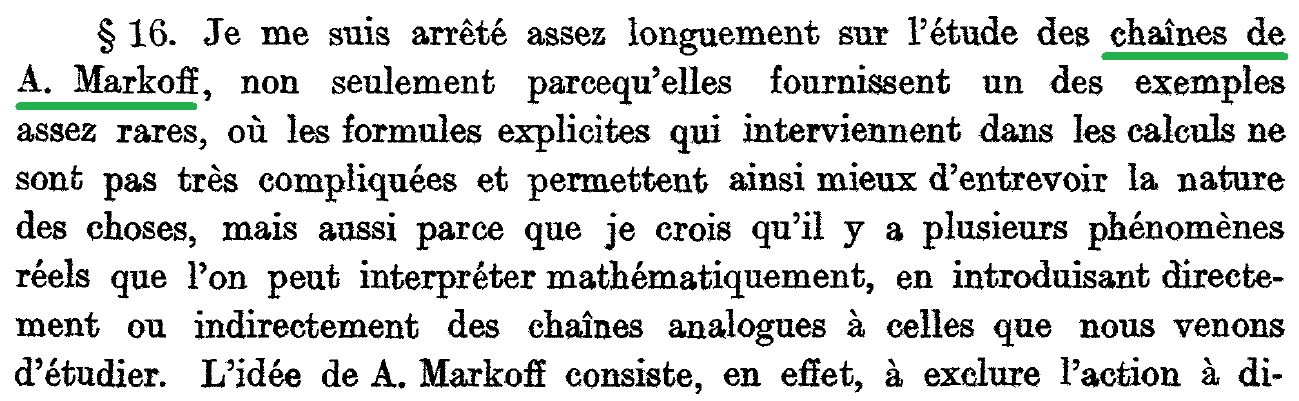}
  \caption[FigBernstein]
   {University Library. University of Illinois at Urbana-Champaign. A photocopy of the fragment of page 40 of Sergei Natanovich Bernstein's paper \cite{bernstein1927} (LEXILOGOS's translation \url{http://www.lexilogos.com/english/french_translation.htm}): \textit{"I have dwelt quite a long time on the study of the chains of A. Markoff, not only because they provide a rather rare example, where the explicit formulas which intervene in the calculations are not very complicated and thus allow better to glimpse the nature of things, but also because I believe that there are several real phenomena that can be interpreted mathematically, by introducing directly or indirectly chains similar to those we have just studied"}. }
  \label{FigBernstein}
\end{figure}
\textit{Reading Sheynin's manuscripts, the author was thinking about losses of the mathematical community following from the fact that they are unpublished.}

These are \textit{Markov chains} - not \textit{times}. Howard Taylor III applies "Markov times" since 1968 \cite[p. 1333 \textit{"Markov time or stopping time"}]{taylor1968} but not in 1965 \cite{taylor1965}. The former article cites Dynkin's \cite{dynkin1963}. The 1965's paper has no references to Russian works. Dynkin says "Markov moment" in \cite[p. 150]{dynkin1968}. This is not a \textit{"moment of distribution"} routinely occurring in statistical literature but synonym of \textit{"time"}, best expressed by English nouns \textit{"time"} or \textit{"instant"}. In the monograph \cite[p. 54]{dynkin1966}, Dynkin and Yushkevich  write "Markov moment" with the meaning of "Markov time".

Three \textit{independent} significant works of 1963 on stopping times are \cite{dynkin1963}, \cite{shiryaev1963}, \cite{chow1963}. The latter two have no the words "Markov times" but Dynkin's paper is translated in English using the "Markov instant". \textit{"Doklady ..."} received it on December 12, 1962. His monograph \cite[p. 142]{dynkin1963b} with Preface dated by March 31, 1962 defines \textit{a random variable independent on the future} and names it "Markov moment" with English synonyms "Market instant", "Market time". While the fundamental properties of the Markov and \textit{strong Markov} processes described in the latter monograph were presented by Dynkin earlier \cite{dynkin1956}, \cite{dynkin1959}, \cite{dynkin1959b}, the author did not find in there the words "Markov moment", "Markov instant", "Markov time". \textit{This investigation suggests that it was Dynkin who first published the name "Markov time" in the paper \cite{dynkin1963} and monograph \cite{dynkin1963b} in 1963, Figure \ref{FigDynkin}.}
\begin{figure}[!h]
  \centering
  \includegraphics[width=115mm]{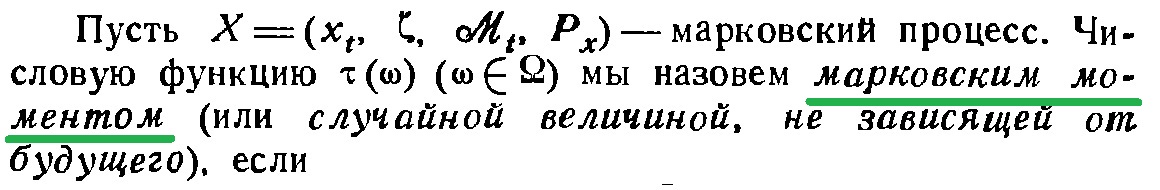}
  \caption[FigDynkin]
   {A photocopy of the fragment of page 142 of Eugene Borisovich Dynkin's monograph \cite[\textit{Strictly Markov Processes}]{dynkin1963b} (VS's translation): \textit{"Let $X=(x_t, \zeta, \mathcal{M}_t, P_x)$ - Markov process. A real valued function $\tau(\omega)$ $(\omega \in \Omega)$ we shall call Markov moment (or random variable, independent on future), if ..."}}
  \label{FigDynkin}
\end{figure}
\paragraph{"Markov time"?} An inventor has the right to name his or her invention. Markov chains and times imply commonality due to the probability spaces setup based on the theories of sets and measure. Accents differ. The Markov property of chains associate with conditional probabilities of future events dependent on the current realized state and independent on the past. The Markov times associate with the past and current events, which can be determined from the information already available. Neftci considers the latter property essential for the signals evaluation. The only reason, why current $I_t$ but not past information is needed, is that $I_t$ includes the past information.

Let us review the \textit{mature} 1968 Dynkin's definition \cite[p. 150]{dynkin1968} (VS's translation): \textit{"Let an expanding sequence of $\sigma$-algebras $F_0 \subseteq F_1 \subseteq \dots \subseteq F_n \subseteq \dots$ is given. Let us name the random variable $\tau$, taking non-negative integer values and the value $+\infty$, Markov moment, if for any finite $n$ $\{\tau=n\} \in F_n$. If, in addition, $\tau < +\infty$ with probability $1$, then let us speak that $\tau$ - stopping moment. Let $X_n$ - a random variable, measurable relative to $F_n$. The stopping time $\tau$ is named optimal, if the value $\boldsymbol{M}X_{\tau}$ is maximal".} Here, the Russian "moment" is a synonym of "time", "instant". This definition implies a set of elementary events and its \textit{nesting} sigma algebras forming expanding, due to $F$, measurable spaces. It adds measures - probabilities. These are traditional expanding probability spaces. \textit{Is there something in this definition binding $\tau$ to the Markov chains except the common setup? The author does not see it.}

Imagine a study that highlights interesting moments in Markov processes. Naming them "Markov times" is reasonable, if other \textit{X processes} are also in scope. Then, "Markov times", "X times" indicate that the instants are from different processes. Dynkin's definition has a wider sense but the word "Markov" associates with Bernstein's "Markov chains" and narrows it. \textit{It is Dynkin time}.

Neftci selects an existing wider concept with a narrower name to extract instants in price time series independent on the future. \textit{Contribution of Salih Neftci \cite{neftci1991} is in that, he has suggested to apply Markov times, defined to distinguish the class of strictly Markov processes in works of Dynkin and Yushkevich \cite{dynkin1956}, \cite{dynkin1959} and used by the modern theory of stochastic  processes \cite{shiryaev1996}, to formalize Technical Analysis widely exploited by traders \cite{murphy1999} and ignored \cite{malkiel2007} or studied \cite{brock1992} by academicians.} In this paper, the words \textit{Dynkin-Neftci time} are chosen to avoid an impression that prices are considered a priori as Markov chains. The time helps to detect trading patterns of Technical Analysis algorithmically by computers and statistically estimate their significance or uselessness. \textit{Will such statistics continue in the future is assumptions: differential equations describe a ballistic trajectory but  an anti-missile system can change it unpredictably}.

While Markov chains are widely applied, Markov himself writes \cite[p. 397]{markov1951}: \textit{"Our conclusions can be expanded also on the complex chains in which each number is directly connected not with one but several preceding it numbers"}. Let us notice that Dynkin times can be used within such a framework and likely for non-Markov processes. The assumption about dependence of a next state only on the last one undoubtedly simplifies simulations. \textit{However, a statement axiomatically and a priori postulating independence of the future on the past can be a speculation influencing on proper understanding markets.}

\paragraph{What may happen, if a model applies non-Dynkin-Neftci times?} Some trading simulators get \textit{delayed quotes}. They are useful for training as long as a student has no access to non-delayed ticks. "Looking in the future" 10 minutes ahead makes liquid S\&P 500 E-mini futures a "boring money machine". The hindsight, for which the simulator is not responsible, does not teach. After switching back to a non-delayed mode, emotions  return. Without mathematics it is clear that information from the future creates arbitrage in time illustrated in \textit{Back to the Future Part II, created by Robert Zemeckis, Bob Gale, 1989}.

There is a concept related to Dynkin-Neftci times. Discussing a filtration $F_i$, a history of the stock until time $i$ on the tree of prices states, \cite[p. 32]{baxter1996} defines \textit{"a previsible process ... on the same tree whose value at any given node at time-tick $i$ is dependent only on the history up to one time-tick earlier, $F_{i-1}$"}.

It is considered imperative, that mathematics of pricing must eliminate theoretical arbitrage. The same theorem, proving a necessary and sufficient condition of the absence of arbitrage, is in \cite[pp. 7 - 9, Theorem 1.7]{hunt2000}, \cite[p.4, Theorem]{duffie2001}. This rationally completes otherwise insufficient stochastic price models yielding a unique value of a derivative. Combining a strategy replicating portfolio with the absence of arbitrage yields option values \cite[p. 3]{hunt2000}.

The \textit{insider trading} of the frozen orange juice futures with tremendous profits for the main personages of the American comedy \textit{Trading Places, directed by John Landis, 1983}, is an \textit{unlawful} arbitrage. \textit{In the fiction}, the "instrument of revenge" is a nearby April contract. Currently, the expiration months are January F, March H, May K, July N, September U, November X. Since the events were developing during the holiday season in December, January or March contracts could be realistic. The wall clock is approaching 9:00:00 am - the opening. Currently, the market opens at 8:00:00 am. After opening at 102 cents per pound, prices move up: $102 \rightarrow 105 \rightarrow 108 \rightarrow 116 \rightarrow  117 \rightarrow  129 \rightarrow 132 \rightarrow  139 \rightarrow  142$, also due to Duke brothers buying on a stolen but falsified crop report. Winthorpe: "Now. Sell thirty April - one forty two!" or in another hearing "Now. Sell two hundred April at one forty two!" This triggers the \textit{opposite trend}: $142 \rightarrow 140 \rightarrow 137 \rightarrow 132 \rightarrow 130 \rightarrow 125 \rightarrow 120 \rightarrow 114 \rightarrow 108 \rightarrow 102$ right before the orange crop TV report. "Ladies and gentlemen, the orange crop estimates for the next year." Silence. The report is bearish: the cold winter is not apparently affecting the orange harvest. The real panic is: $96 \rightarrow 85 \rightarrow 77 \rightarrow 56 \rightarrow 46 \rightarrow 38 \rightarrow 35 \rightarrow 30 \rightarrow 29$ last. One contract is 15,000 pounds: one point is \$150. The move is $(142 - 29) * \$150 = \$16,950$. At the end, Winthorpe is busy closing the short positions and three times saying "hundred". 300 contracts could profit \$5,085,000 before commissions, taxes, and in 1983. \textit{It is worth noticing that during 1979 - 1983 there were no such low frozen orange juice prices. In real life, it would attract attention of the existing since 1974 Commodity Futures Trading Commission, CFTC.}

While the no-arbitrage theoretically links stochastic price processes with unique options values, the author believes that for trading futures with high leverage and large positions, a model accurately simulating discrete prices is more practical than the condition of no-arbitrage needed to rationally price derivatives. After reviewing this section, the author's younger son-student Dmitri has "invented" the joke: \textit{"To be a trader, one does not have to be successful"}.

\paragraph{Optimal trading elements, OTE.} Between entering and exiting the market, MPS0 reverses long to short positions \cite[pp. 25 - 26, Property 4]{salov2007}. Several MPS0 may generate the same $PL$. For constant $C$, times of reversal transactions and some \textit{local} price minimums and maximums coincide. The net action $\sum_{i=1}^{i=n} U_{i,MPS0} = 0$, MPS0 $\in \mathfrak{U}$. The time of the last transaction of MPS0 is non-Dynkin-Neftci. It can change after arriving new information. This artificial transaction marks $PL$ to market. All transactions before the last one associate with Dynkin-Neftci times and can be used as signals for building real trading rules \cite{salov2008}. \textit{The later can lose money.}

MPS0 with $W=1$ creates \textit{optimal trades} adjacent in time. \cite[p. 39]{salov2011b} defines the \textit{optimal trading element}: \textit{"a collective name for properties associated with an optimal trade returned by an MPS"}. This does not limit the number of properties. The key is their association with the MPS0 optimal trades. The latter depend on the \textit{filtering cost $FC$}. Perspective properties are \cite[p. 39]{salov2011b}: a) trade direction  - a buy or sell to initiate the trade, b) profit - optimal trades always profit, c) duration - time length of the trade, d) number of ticks including the first and last transaction of the trade, e) volume - total market volume during the trade, f) empirical distribution of a-increments - waiting times between neighboring ticks, g) empirical distribution of b-increments - price increments between neighboring ticks, h) empirical distribution of price and/or volume. Once a MPS0 with a filtering cost as a tool is applied to a chain of ticks and OTE are evaluated, the next analytical step is computing statistics of OTE. Due to reversal properties, a buying OTE, BOTE, is followed by a selling OTE, SOTE, and vice versa. The mean b-increment (price increment) of a BOTE is always positive. The mean b-increment of a SOTE is always negative.

\paragraph{OTE by example.} Figure \ref{FigOTE} represents eight OTEs for filtering cost \$100 found in trading sessions on April 10, 2017 for ESM17, ESU17, and ESZ17. The red up and blue down parallel lines shadow areas above time intervals of the optimal trades.

\begin{figure}[!h]
  \centering
  \includegraphics[width=110mm]{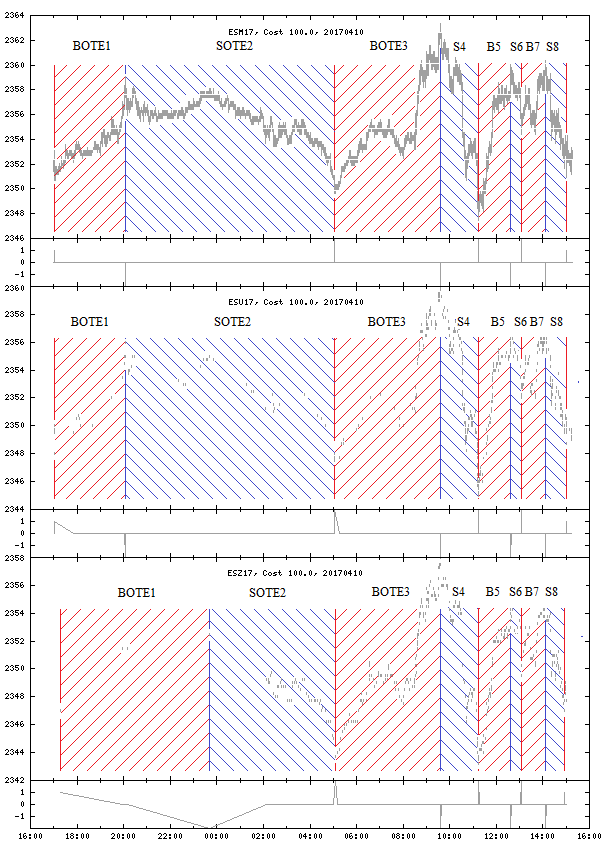}
  \caption[FigOTE]
   {MPS0, Filtering cost \$100, E-mini S\&P 500 Futures Time \& Sales Globex, \url{http://www.cmegroup.com/}, transaction prices of ESM17, ESU17, and ESZ17 for the time range [Sunday April 9, 2017, 17:00:00 - Monday April 10, 2017, 15:15:00], CST. Plotted using custom C++  and Python programs and gnuplot \url{http://www.gnuplot.info/}.}
  \label{FigOTE}
\end{figure}
\begin{center}
\begin{longtable}{|r|r|r|r|r|r|r|r|}
\caption[Guessed Counts]{ESM17, Session OTEs, $FC$ = \$100, $C$ = \$4.68, $W=1$.} \label{TblESM17OTE} \\
 \hline
 \multicolumn{1}{|c}{\#} &
 \multicolumn{1}{|c}{$t_{start}$} &
 \multicolumn{1}{|c}{$P_{start}$} &
 \multicolumn{1}{|c}{$t_{end}$} &
 \multicolumn{1}{|c}{$P_{end}$} &
 \multicolumn{1}{|c}{$\Delta t, s$} &
 \multicolumn{1}{|c}{$PL$, \$} &
 \multicolumn{1}{|c|}{Type} \\
 \hline 
 \endfirsthead
 \multicolumn{8}{c}%
 {\tablename\ \thetable{} -- continued from previous page} \\
 \hline
 \multicolumn{1}{|c}{\#} &
 \multicolumn{1}{|c}{$t_{start}$} &
 \multicolumn{1}{|c}{$P_{start}$} &
 \multicolumn{1}{|c}{$t_{end}$} &
 \multicolumn{1}{|c}{$P_{end}$} &
 \multicolumn{1}{|c}{$\Delta t, s$} &
 \multicolumn{1}{|c}{$PL$, \$} &
 \multicolumn{1}{|c|}{Type} \\
 \hline 
 \endhead
 \hline \multicolumn{8}{|r|}{{Continued on next page}} \\ \hline
 \endfoot
 \hline
 \endlastfoot
1 & 2017-04-09 17:02:54 & 2350.75 & 2017-04-09 20:04:00 & 2359.00 & 10866 & 403.14 & BOTE\\
2 & 2017-04-09 20:04:00 & 2359.00 & 2017-04-10 05:03:55 & 2349.75 & 32395 & 453.14 & SOTE\\
3 & 2017-04-10 05:03:55 & 2349.75 & 2017-04-10 09:35:48 & 2363.25 & 16313 & 665.64 & BOTE\\
4 & 2017-04-10 09:35:48 & 2363.25 & 2017-04-10 11:14:41 & 2347.50 &  5933 & 778.14 & SOTE\\
5 & 2017-04-10 11:14:41 & 2347.50 & 2017-04-10 12:37:14 & 2360.00 &  4953 & 615.64 & BOTE\\
6 & 2017-04-10 12:37:14 & 2360.00 & 2017-04-10 13:04:45 & 2354.50 &  1651 & 265.64 & SOTE\\
7 & 2017-04-10 13:04:45 & 2354.50 & 2017-04-10 14:06:28 & 2360.25 &  3703 & 278.14 & BOTE\\
8 & 2017-04-10 14:06:28 & 2360.25 & 2017-04-10 15:00:07 & 2351.00 &  3219 & 453.14 & SOTE\\
\end{longtable}
\end{center}
The eight profits and durations from Table \ref{TblESM17OTE} form two sample distributions
\begin{verbatim}
PL distribution
Mean                = 489.0775
Samples size        = 8
Maximum value       = 778.14
Maximum value count = 1
Minimum value       = 265.64
Minimum value count = 1
Variance            = 33590.9598
Std. deviation      = 183.278367
Skewness            = 0.322282476
Excess kurtosis     = -1.39538524
0 (233.442, 311.256] 2
1 (311.256, 389.07] 0
2 (389.07, 466.884] 3
3 (466.884, 544.698] 0
4 (544.698, 622.512] 1
5 (622.512, 700.326] 1
6 (700.326, 778.14] 1

Trade time distribution
Mean                = 9879.125
Samples size        = 8
Maximum value       = 32395
Maximum value count = 1
Minimum value       = 1651
Minimum value count = 1
Variance            = 105625105
Std. deviation      = 10277.4075
Skewness            = 1.82711153
Excess kurtosis     = 1.68683095
0 (0, 3239.5] 2
1 (3239.5, 6479] 3
2 (6479, 9718.5] 0
3 (9718.5, 12958] 1
4 (12958, 16197.5] 0
5 (16197.5, 19437] 1
6 (19437, 22676.5] 0
7 (22676.5, 25916] 0
8 (25916, 29155.5] 0
9 (29155.5, 32395] 1
\end{verbatim}

\paragraph{Start and birth times of OTE.} When the first tick arrives, nothing is known with respect to the MPS0 and OTE, unless the previous trading sessions are considered. The \textit{OTE start time} $t_s^{OTE}$ remains unknown until the \textit{first} price arrived will mark a move exceeding $2FC$ at least by one $\delta_{ESM17}=0.25$. This event, \textit{the OTE birth time} $t_b^{OTE}$, is in the future making $t_s^{OTE}$ non-Dynkin-Neftci. The $2FC$ are counted from a local minimum or maximum price. $t_b^{OTE}$ is Dynkin-Neftci. Just only such a price drop or rise occurs, $t_s^{OTE}$ is fixed. After this, the start time is Dynkin-Neftci but the current OTE end time $t_e^{OTE}$ coinciding with the next OTE start time are unknown - non-Dynkin-Neftci. Notice, all MPS0 start, birth, and end times prior the current just fixed start time cannot change. They are Dynkin-Neftci times. Figure \ref{FigOTEStartBirthEnd} is a zoom in to Figure \ref{FigOTE} for ESM17 SOTE $\#4$, Table \ref{TblESM17OTE}. Again, the $\#4$ end time can be determined only after arriving the $\#5$ birth time.

$|\Delta P| =\frac{2\times \$100}{\$50}+0.25=4.25$. The last maximum (our case) price is 2363.25 observed at 09:35:48. Would the price go higher, it would become the new trailing high. Subtracting $|\Delta P|$ yields the \textit{birth price} 2359.00. It arrives at 09:59:13. This closes the previous BOTE and creates the next SOTE$(P_s^{S4}=2363.25, t_s^{S4}=09:35:48; P_b^{S4}=2359.00, t_e^{S4}=09:59:13; P_e^{S4}=?, t_e^{S4}=?)$.

\begin{figure}[!h]
  \centering
  \includegraphics[width=125mm]{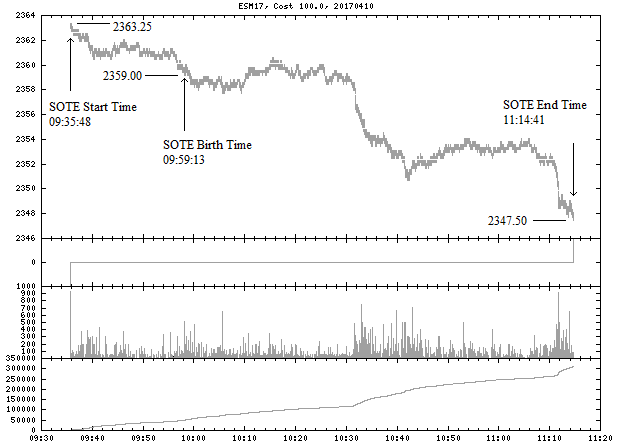}
  \caption[FigOTEStartBirthEnd]
   {Time \& Sales Globex, \url{http://www.cmegroup.com/}, ESM17, Sunday April 9, 2017, 17:00:00 - Monday April 10, 2017, 15:15:00. MPS0, $FC = \$100$, transaction prices, SOTE \#4, Table \ref{TblESM17OTE} start, birth, and end times. Plotted using custom C++ and Python programs and gnuplot \url{http://www.gnuplot.info/}.}
  \label{FigOTEStartBirthEnd}
\end{figure}

By definition, OTEs include arbitrary properties associated with optimal MPS trades. Figure \ref{FigOTEDistr} illustrates four properties of SOTE $\#4$. Let us notice, that if price b-increments would be i.i.d log-normal (or normal), then distribution of prices would be the same but with a different mean and variance. The sample distribution density of prices has three maximums and does not correspond to a unimodal log-normal (or normal) distribution.

\begin{figure}[!h]
  \centering
  \includegraphics[width=125mm]{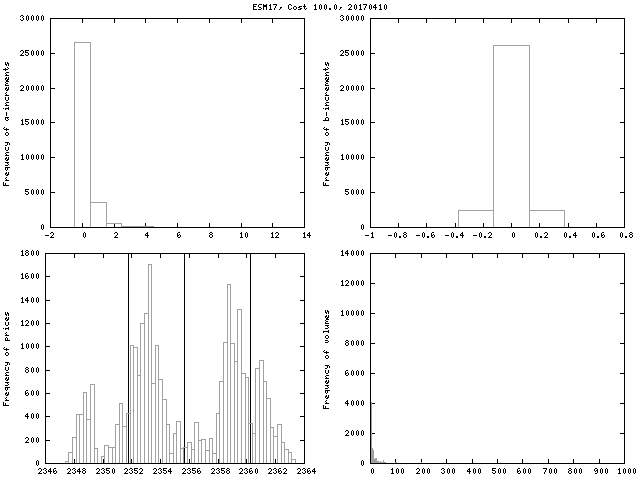}
  \caption[FigOTEDistr]
   {Time \& Sales Globex, \url{http://www.cmegroup.com/}, ESM17, Sunday April 9, 2017, 17:00:00 - Monday April 10, 2017, 15:15:00. MPS0, $FC = \$100$,  sample distributions of a-increments, b-increments, prices, and volumes, SOTE \#4, Table \ref{TblESM17OTE}. Plotted using custom C++ and Python programs and gnuplot \url{http://www.gnuplot.info/}.}
  \label{FigOTEDistr}
\end{figure}

\paragraph{Three scenarios for OTE.} Once a new OTE is born, three exclusive scenarios exist: 1) the profit of the current OTE will grow at least by one $\delta$, \textit{absolute minimal price fluctuation}; 2) a next opposite type OTE will replace the current; 3) the trading session will terminate. For a non-intraday trader, the chains of prices and OTEs are continued. \textit{Technically, one can apply MPS, as a tool, using any chain of prices such as daily last prices.}

Figure \ref{FigOTEStartBirthEnd} illustrates the first scenario for SOTE4, $S4$, from Table \ref{TblESM17OTE}. After the birth time $t_b=$09:59:13, the price never moved more than $4.25=17\delta_{ESM17}$ points up on the right of the last lowest price until the $B5$ birth price.

 Table \ref{TblESM17OTE2} collects basic OTE properties after switching to $FC=\$74.99$. ESM17 prices are the same. $S4$ on Figures \ref{FigOTEScenario2}, \ref{FigOTEStartBirthEnd2} presents the second scenario. The influence of $FC$ under other equal conditions on the number of OTEs, their durations and $PL$ distributions, spectra is studied in \cite{salov2013}, \cite{salov2017}.
\begin{center}
\begin{longtable}{|r|r|r|r|r|r|r|r|}
\caption[Guessed Counts]{ESM17, Session OTEs, $FC$ = \$74.99, $C$ = \$4.68, $W=1$.} \label{TblESM17OTE2} \\
 \hline
 \multicolumn{1}{|c}{\#} &
 \multicolumn{1}{|c}{$t_{start}$} &
 \multicolumn{1}{|c}{$P_{start}$} &
 \multicolumn{1}{|c}{$t_{end}$} &
 \multicolumn{1}{|c}{$P_{end}$} &
 \multicolumn{1}{|c}{$\Delta t, s$} &
 \multicolumn{1}{|c}{$PL$, \$} &
 \multicolumn{1}{|c|}{Type} \\
 \hline 
 \endfirsthead
 \multicolumn{8}{c}%
 {\tablename\ \thetable{} -- continued from previous page} \\
 \hline
 \multicolumn{1}{|c}{\#} &
 \multicolumn{1}{|c}{$t_{start}$} &
 \multicolumn{1}{|c}{$P_{start}$} &
 \multicolumn{1}{|c}{$t_{end}$} &
 \multicolumn{1}{|c}{$P_{end}$} &
 \multicolumn{1}{|c}{$\Delta t, s$} &
 \multicolumn{1}{|c}{$PL$, \$} &
 \multicolumn{1}{|c|}{Type} \\
 \hline 
 \endhead
 \hline \multicolumn{8}{|r|}{{Continued on next page}} \\ \hline
 \endfoot
 \hline
 \endlastfoot
 1 & 2017-04-09 17:02:54 & 2350.75 & 2017-04-09 20:04:00 & 2359.00 & 10866 & 403.14 & BOTE\\
 2 & 2017-04-09 20:04:00 & 2359.00 & 2017-04-10 05:03:55 & 2349.75 & 32395 & 453.14 & SOTE\\
 3 & 2017-04-10 05:03:55 & 2349.75 & 2017-04-10 06:37:43 & 2355.50 & 5628 & 278.14 & BOTE\\
 4 & 2017-04-10 06:37:43 & 2355.50 & 2017-04-10 07:59:57 & 2352.50 & 4934 & 140.64 & SOTE\\
 5 & 2017-04-10 07:59:57 & 2352.50 & 2017-04-10 09:35:48 & 2363.25 & 5751 & 528.14 & BOTE\\
 6 & 2017-04-10 09:35:48 & 2363.25 & 2017-04-10 10:42:06 & 2350.75 & 3978 & 615.64 & SOTE\\
 7 & 2017-04-10 10:42:06 & 2350.75 & 2017-04-10 10:53:17 & 2354.00 & 671 & 153.14 & BOTE\\
 8 & 2017-04-10 10:53:17 & 2354.00 & 2017-04-10 11:14:41 & 2347.50 & 1284 & 315.64 & SOTE\\
 9 & 2017-04-10 11:14:41 & 2347.50 & 2017-04-10 12:37:14 & 2360.00 & 4953 & 615.64 & BOTE\\
10 & 2017-04-10 12:37:14 & 2360.00 & 2017-04-10 13:04:45 & 2354.50 & 1651 & 265.64 & SOTE\\
11 & 2017-04-10 13:04:45 & 2354.50 & 2017-04-10 13:15:52 & 2358.25 & 667 & 178.14 & BOTE\\
12 & 2017-04-10 13:15:52 & 2358.25 & 2017-04-10 13:40:14 & 2354.50 & 1462 & 178.14 & SOTE\\
13 & 2017-04-10 13:40:14 & 2354.50 & 2017-04-10 14:06:28 & 2360.25 & 1574 & 278.14 & BOTE\\
14 & 2017-04-10 14:06:28 & 2360.25 & 2017-04-10 15:00:07 & 2351.00 & 3219 & 453.14 & SOTE\\
15 & 2017-04-10 15:00:07 & 2351.00 & 2017-04-10 15:02:13 & 2354.25 & 126 & 153.14 & BOTE\\
16 & 2017-04-10 15:02:13 & 2354.25 & 2017-04-10 15:14:30 & 2351.25 & 737 & 140.64 & SOTE\\
\end{longtable}
\end{center}

\begin{figure}[!h]
  \centering
  \includegraphics[width=125mm]{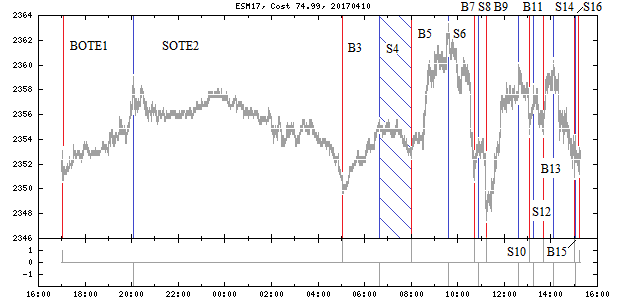}
  \caption[FigOTEScenario2]
   {MPS0, Filtering cost \$74.99, E-mini S\&P 500 Futures Time \& Sales Globex, \url{http://www.cmegroup.com/}, transaction prices of ESM17 for the time range [Sunday April 9, 2017, 17:00:00 - Monday April 10, 2017, 15:15:00], CST. Plotted using custom C++  and Python programs and gnuplot \url{http://www.gnuplot.info/}.}
  \label{FigOTEScenario2}
\end{figure}
\begin{figure}[!h]
  \centering
  \includegraphics[width=125mm]{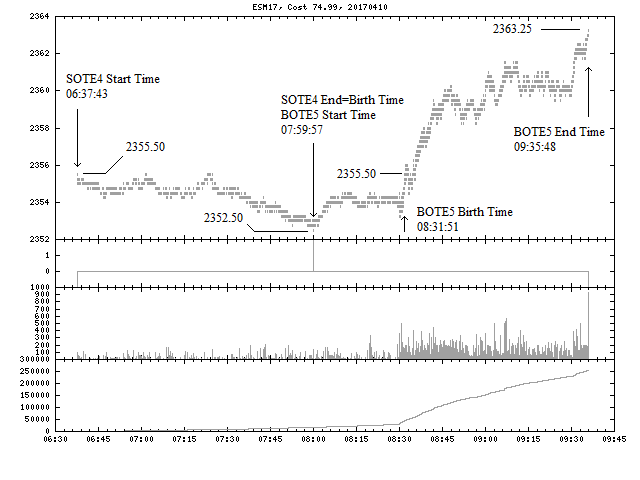}
  \caption[FigOTEStartBirthEnd2]
   {Time \& Sales Globex, \url{http://www.cmegroup.com/}, ESM17, Sunday April 9, 2017, 17:00:00 - Monday April 10, 2017, 15:15:00. MPS0, $FC = \$74.99$, transaction prices, SOTE \#4, BOTE \#5 Table \ref{TblESM17OTE2} start, birth, and end times. Plotted using custom C++ and Python programs and gnuplot \url{http://www.gnuplot.info/}.}
  \label{FigOTEStartBirthEnd2}
\end{figure}
In the second scenario, the price does not go a $\delta$ in the profit direction of the OTE type. Selling short $S4$ at the OTE birth price 2352.50, 07:59:57 and buying at the next $B5$ birth price 2355.50, 08:31:51 loses $(-2352.50 + 2355.50) \times \$50 -\$9.36 = -\$159.36$. A simple trading rule - enter/exit and revert position by buying BOTE and selling SOTE at the OTE birth price - resembles the \textit{Alexander's filter}, see details in \cite[p. 71, pp. 92 - 93]{salov2013}. Applying it to $(S4, B5)$, with the $B5$'s profit $([2363.25 - 3.0] - 2355.50) \times \$50 - 2 \times \$4.68 = \$228.14$, would compensate the $S4$'s loss $-\$159.36$ and be profitable $\$228.14 -\$159.36 =\$68.78$, Figure \ref{FigOTEStartBirthEnd2}. This rule can start losing, if the second scenario continues in a chain of consecutive OTEs. Empirical distributions of the OTE $PL$ \cite[p. 27, Figure Profit Frequencies]{salov2012}, \cite[p. 95, Figure 35]{salov2013} help estimating the mean $PL$. It was found negative in a range of $FC$ and $C=\$4.66$. A mathematical expectation of $PL$ can be not the only criterion influencing on trading decisions \cite{salov2015}.

The third scenario implies that the current OTE did not get any development relative to $FC$ with regard to the profit given the time prior the trading sessions is closed. The SOTE \#16 from Table \ref{TblESM17OTE2} is an example, Figure \ref{FigOTEStartBirthEnd3}.

\begin{figure}[!h]
  \centering
  \includegraphics[width=125mm]{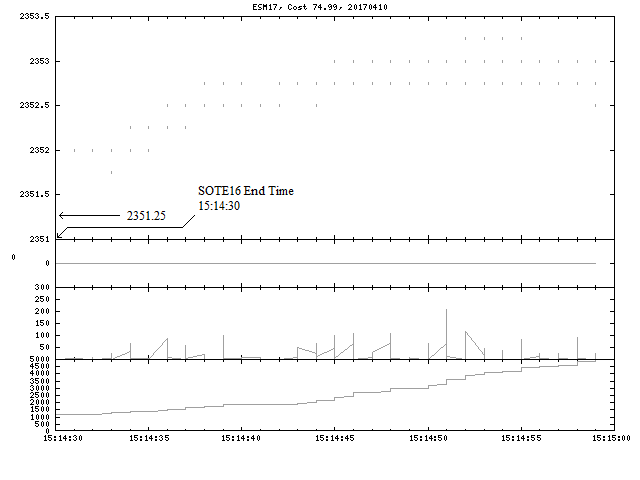}
  \caption[FigOTEStartBirthEnd3]
   {Time \& Sales Globex, \url{http://www.cmegroup.com/}, ESM17, Sunday April 9, 2017, 17:00:00 - Monday April 10, 2017, 15:15:00. MPS0, $FC = \$74.99$, transaction prices, after SOTE \#16 Table \ref{TblESM17OTE2}. Plotted using custom C++ and Python programs and gnuplot \url{http://www.gnuplot.info/}.}
  \label{FigOTEStartBirthEnd3}
\end{figure}

\paragraph{The maximum loss strategy, MLS.} Such a strategy could be a risk estimator for given $\boldsymbol{P}$ and $\boldsymbol{C}$. Here, we only mention that $\boldsymbol{U}_{MLS} \in \mathfrak{U}$ is not necessarily $-\boldsymbol{U}_{MPS0} \in \mathfrak{U}$. For instance, for $\boldsymbol{P}=(P, \dots, P)^T$, $P > 0$, $\boldsymbol{C}=(C, \dots, C)^T$, $C>0$, $\boldsymbol{U}_{MPS0} = \boldsymbol{U}_{d.n.s.}$, losing nothing, while many other strategies lose more due to transactions costs. From section "Strategies generating extreme industry gains", it follows that under such conditions, for even $3 < n$, $\boldsymbol{U}_{MLS}=(W, -2W, \dots, 2W, -W)^T$ with the loss $-2CW(n-1)$ (valid for any $1<n$). Another and only candidate would be $-\boldsymbol{U}_{MLS}=(-W, 2W, \dots, -2W, W)^T$.

\paragraph{MPS studies, 2007 - 2017.} Potential profit, as a number computed without accounting transaction costs, and its simple, under this condition, algorithm were suggested by Robert Pardo \cite[pp. 125 - 126]{pardo1992}. His words \textit{"the measurement of the potential profit that a market offers is not a widely understood idea"} attracted the author, who thought that transaction costs would complicate the algorithm, yield a rich concept of the maximum profit strategy, vector, expand research and application horizons: MPS is \textit{another} face of the \textit{same} market. The unpublished formulation of MPS and related algorithms were developed by the author in 1994. They were not applied or widely reported in that time and eventually have been followed by the studies and publications listed below.

\begin{enumerate}
\item[1)]
The l- and r-algorithms for MPS0, and algorithms for the first MPS1 and second MPS2 P\&L reserve strategies \cite{salov2007}, \cite[p. 203, Appendix F, A Minor Correction]{salov2013}.
\item[2)]
Fist published evaluations of MPS0, MPS1, MPS2 \cite[Chapter 7]{salov2007}. Using MPS0 as a performance benchmark \cite[pp. 151 - 152]{salov2007}, for comparing different intervals of a single market \cite[p. 152]{salov2007}, for comparing different markets \cite[p. 153, Chapter 10]{salov2007}, as moving indicators \cite[p. 153]{salov2007}.
\item[3)]
First published statistics of MPS0 optimal trades \cite[Chapter 9]{salov2007}.
\item[4)]
Proposal to apply MPS0 for options on potential profit \cite[p. 155]{salov2007}.
\item[5)]
Proposal to consider MSP0 as a quantitative alternative to not well defined trend and volatility \cite[pp. 153 - 154]{salov2007}, \cite[Alternative analysis]{salov2011b}.
\item[6)]
Using MPS0 for filtering events \cite[pp. 154 - 155]{salov2007}, and defining trading signals of real trading rules and strategies \cite{salov2008}.
\item[7)]
Introducing the a-b-c-increments classification \cite{salov2011}. Statistical studies of the a-b-c-increments including MPS0 optimal trades \cite{salov2011}, \cite{salov2011b}, \cite{salov2012}, \cite{salov2013}, \cite{salov2017}. Chain reactions \cite[pp. 101 - 105]{salov2013}, \cite[pp. 52 - 56, Extreme b-increments]{salov2017}.
\item[8)]
Proposal of new notation for iteration of functions and iteral of functions to describe a-b-c-processes \cite{salov2012b}, \cite{salov2012c}.
\item[9)]
Introducing OTE, BOTE, SOTE \cite{salov2011b}. Studying OTE statistics and expectations for trading based on the OTE birth time and price \cite{salov2011b}, \cite{salov2012}, \cite[p. 95, Figure 35]{salov2013}.
\item[10)]
New (not in \cite{salov2007}, \cite{salov2008}, \cite{salov2011}, \cite{salov2011b}) discrete, spectral MPS0 properties and software framework \cite{salov2013}, \cite{salov2017}. Relationships between MPS0 and trading volume \cite[pp. 26 - 30, Figures 18 - 23]{salov2017}. Mathematical expectations as not the only reason for trading decisions \cite{salov2014}, \cite[p. 12, Livermore about the hope and fear, pp. 28 - 31, The role of time. Trading and speculation]{salov2015}.
\end{enumerate}

\section{Patterns}

\paragraph{Strategies vs. rules.} "Trading strategies" have two meanings: 1) records of trading actions like chains - vectors of $\mathfrak{U}$, and 2) reasons causing the actions. Here, the latter are named "trading rules". Automation of trading depends on the formalization of rules. Formalization of a rule is valuable, if it yields an algorithm or program, which can be evaluated by a human being or computer. For this, it can rely on Dynkin-Neftci times distinguishing events, which can be determined without looking to the future.

Given available information required by a rule, the latter is evaluated and a conclusion is made whether an event takes place. The number of events within a time interval, associates with the frequency of the market offers \textit{connected} to the rule. The next is to \textit{estimate} what may happen after the event. This estimation can be subjective or imply a preliminary objective statistical research, a search of dependencies between an event and market follow up, \textit{if any}. Such a research assumes evaluation of the past up to day information and may have meaning, if the future repeats the past in something. Finding this something is one of the goals. Due to non-stationary market conditions, once determined useful regularities may stop working and trigger a new research. The question discussed in the last section, especially logical if prices are totally unpredictable, is \textit{why do speculative markets exist}?

MPS0 indicates local price minimums and maximums on a given time interval. The last extreme depends on the future, non-Dynkin-Neftci time. Other are Dynkin-Neftci times. A MPS0 builds a chain of OTEs such as on Figure \ref{FigOTE} $(B1,S2,B3,S4,B5,S6,B7,S8)$. An OTE is characterized by the start and end OTE times and prices $t_s$, $P_s$, $t_e$, $P_e$. Its birth time and price $t_b$, $P_b$ are practically important being deal with a Dynkin-Neftci time.

\paragraph{How can MPS0 and OTE define patterns?} Let us consider a hypothetical chain $(B1,S2,B3,S4,B5,S6)$. By the MPS0 properties, $P_e^{B1}=P_s^{S2}$, $P_e^{S2}=P_s^{B3}$, $P_e^{B3}=P_s^{S4}$, $P_e^{S4}=P_s^{B5}$, $P_e^{B5}=P_s^{S6}$. Since $S6$ is current, its birth price $P_b^{S6} < P_s^{S6}$ is realized and $P_s^{S6} - P_b^{S6} \ge 2FC + \delta$, where $2FC$ is converted to full points. The known head and shoulders pattern can be defined by the condition: $(P_s^{B1}<P_s^{B3}) \&\& (P_s^{B3}==P_s^{B5})\&\&(P_e^{B1}<P_e^{B3})\&\&(P_e^{B5}<P_e^{B3})\&\&(P==P_b^{B5})$. The logical equality $==$ and less $<$ comparisons assume tolerances. These tolerances together with $FC$ are the pattern optimization parameters.

\paragraph{Algorithmic optimization.} Since $B1, \dots, B5$, $P_s^{S6}$, and $P_b^{S6}$ are fixed, the logical expression in the previous paragraph requires only a one time evaluation of all individual comparisons but $(P==P_b^{B5})$, which must be monitored. A recognition of this pattern can be efficient and does not require reevaluation all six OTEs for each arriving price $P$.

\paragraph{Interval $[t_s,t_b]$.} In contrast with theories of continuous prices, real prices are discrete \cite[pp. 32 - 33]{salov2013}, \cite[pp. 3 - 10]{salov2017}. In particular, futures prices are products of natural numbers and $\delta$. For constant $FC$, the price change to monitor is expressed in $\delta$ as $|P_s^{OTE} - P_b^{OTE}|=|\delta N_s^{OTE} - \delta N_b^{OTE}| = \delta |\Delta N^{OTE}| > \frac{2FC}{k}$ and $N_{s,b}^{OTE} = \lfloor \frac{2FC}{\delta k}\rfloor + 1$. Example, $\lfloor \frac{2 \times \$100}{0.25 \times \$50}\rfloor + 1 = 17$ deltas or 4.25 full points; $\lfloor \frac{2 \times \$74.99}{0.25 \times \$50}\rfloor + 1 = 12$ deltas or 3 full points. Notice, that for $FC = 0$ the formula returns 1: MPS extracts only profitable but not break even optimal trades.

It is possible that due to a \textit{price gap} the price will jump over $P_b^{OTE}$. This still indicates that the new OTE is born and actual price $P_s^{OTE}$ is fixed in the price chain. The time of the gapped price is the born time $t_b^{OTE}$.

Now, the current OTE segment $[t_s, t_b]$ is fixed. The $[t_b, t_{current}]$ is developing following one of the three scenarios. The l-, r-algorithms routinely extract all previous times $t_s$ and $t_b$ and allow to study what happens at these times and in the past adjacent intervals $[t_s, t_b]$ and $[t_b, t_e]$ or entire OTE intervals $[t_s, t_e]$. The next paragraph is an example of a study.

\paragraph{Empirical Cumulative Distribution Functions, ECDF, of OTE profits.} Once $t_b$ is detected, the profit of $[t_s,t_b]$ cannot be realized since it is in the past. It is interesting how far the price can move in the profit direction after $t_b$. In other words, with fixed $C = \$4.68$ and $FC$ how many OTEs exceed this interval by $\delta, 2\delta, 3\delta, \dots$ etc.
\begin{figure}[!h]
  \centering
  \includegraphics[width=125mm]{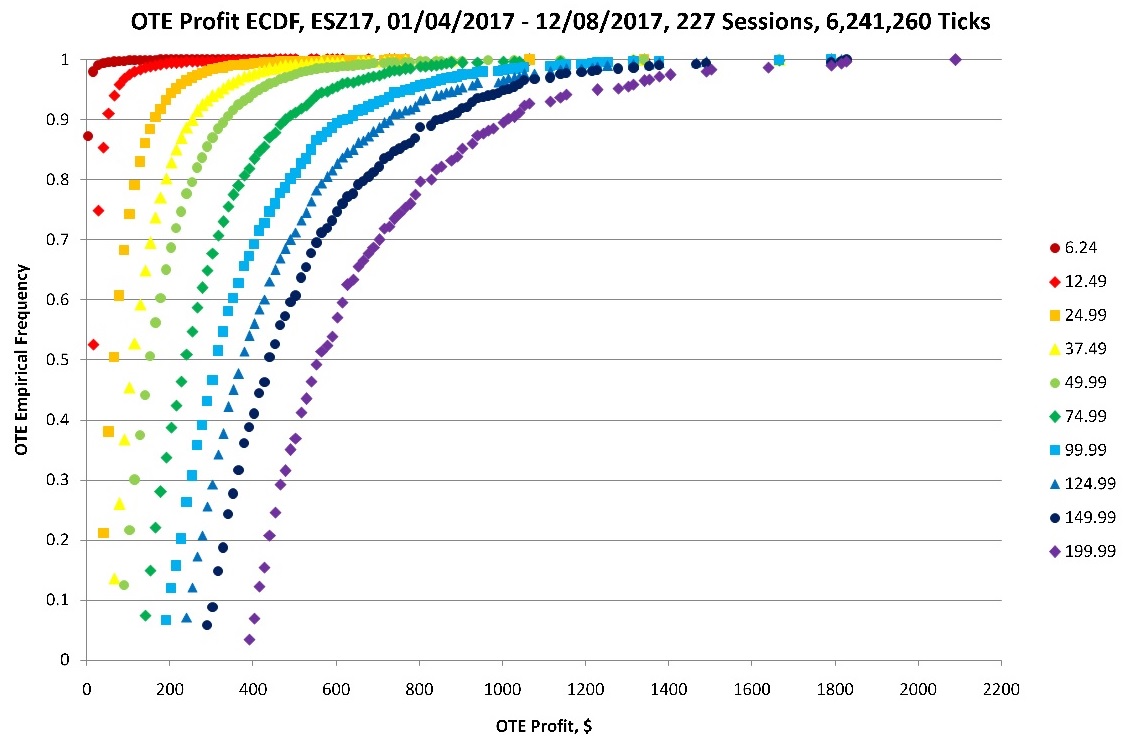}
  \caption[FigOTEProfitECDF]
   {Time \& Sales Globex, \url{http://www.cmegroup.com/}, ESZ17, Wednesday January 4, 2017 - Friday December 8, 2017, 227 trading sessions with the time range 17:00:00 (previous day) - 15:15:00 (closing day), 6,241,260 transaction ticks, $C = \$4.68$. MPS0 OTE ECDF for $FC = 6.24, 12.49, 24.99, 37.49, 49.99, 74.99, 99.99, 124.99, 149.99, 199.99$ dollars. Plotted using custom C++, Python, AWK programs, and Microsoft Excel.}
  \label{FigOTEProfitECDF}
\end{figure}
Figure \ref{FigOTEProfitECDF} depicts ECDF of OTE profits depending on $FC$. The results are for individual trading sessions assuming that a trader does not leave open positions for next sessions. The ranges 15:30:00 - 16:00:00 are ignored. Table \ref{TblOTEProfitStatistics} summarizes sample statistics.
\begin{center}
\begin{longtable}{|r|r|r|r|r|r|r|r|r|r|}
\caption[Guessed Counts]{ESZ17 OTE Profit Sample Statistics, $C=\$4.68$.} \label{TblOTEProfitStatistics} \\
 \hline
 \multicolumn{1}{|c}{$FC$} &
 \multicolumn{1}{|c}{$N$} &
 \multicolumn{1}{|c}{Mean} &
 \multicolumn{1}{|c}{Min} &
 \multicolumn{1}{|c}{$N_{Min}$} &
 \multicolumn{1}{|c}{Max} &
 \multicolumn{1}{|c}{$N_{Max}$} &
 \multicolumn{1}{|c}{StDev} &
 \multicolumn{1}{|c}{Skew.} &
 \multicolumn{1}{|c|}{E-Kurt.} \\
 \hline 
 \endfirsthead
 \multicolumn{10}{c}%
 {\tablename\ \thetable{} -- continued from previous page} \\
 \hline
 \multicolumn{1}{|c}{$FC$} &
 \multicolumn{1}{|c}{$N$} &
 \multicolumn{1}{|c}{Mean} &
 \multicolumn{1}{|c}{Min} &
 \multicolumn{1}{|c}{$N_{Min}$} &
 \multicolumn{1}{|c}{Max} &
 \multicolumn{1}{|c}{$N_{Max}$} &
 \multicolumn{1}{|c}{StDev} &
 \multicolumn{1}{|c}{Skew.} &
 \multicolumn{1}{|c|}{E-Kurt.} \\
 \hline 
 \endhead
 \hline \multicolumn{10}{|r|}{{Continued on next page}} \\ \hline
 \endfoot
 \hline
 \endlastfoot
6.24 & 860375 & 5.41 & 3.14 & 751179 & 1340.64 & 1 & 10.0 & 20.8 & -0.82\\
12.49 & 69715 & 31.15 & 15.64 & 36654 & 1340.64 & 1 & 32.3 & 7.14 & 10.5\\
24.99 & 12342 & 92.88 & 40.64 & 2601 & 1340.64 & 1 & 68.5 & 3.51 & 25.2\\
37.49 & 5816 & 146.40 & 65.64 & 793 & 1665.64 & 1 & 96.5 & 3.38 & 25.4\\
49.99 & 3542 & 195.03 & 90.64 & 439 & 1665.64 & 1 & 119.9 & 2.88 & 16.3\\
74.99 & 1749 & 288.59 & 140.64 & 130 & 1790.64 & 1 & 156.6 & 2.43 & 11.4\\
99.99 & 1034 & 379.66 & 190.64 & 68 & 1790.64 & 1 & 198.4 & 2.18 & 7.03\\
124.99 & 734 & 450.82 & 240.64 & 52 & 1815.64 & 1 & 223.2 & 2.07 & 6.00\\
149.99 & 534 & 524.11 & 290.64 & 31 & 1828.14 & 1 & 243.5 & 1.94 & 5.05\\
199.99 & 317 & 658.23 & 390.64 & 11 & 2090.64 & 1 & 280.7 & 1.95 & 4.61\\
\end{longtable}
\end{center}
Let us review Table \ref{TblOTEProfitStatistics} and $FC=\$49.99$ as an example. Buying BOTE or selling SOTE at $t_b$ makes the part of the OTE profit $2 \times \$49.99 = \$99.98$ unavailable. At the first glance, this is attractive because the mean OTE profit is $\$195.03$ and the difference $\$195.03 - \$99.98 = \$95.05$ is positive. However, "on the way back to a next $t_b$" additional $\$99.98$ has to be subtracted creating the mean loss $\$95.05 - \$99.98 = -\$4.93$. At the same time, taking the OTE mean profit always at a corresponding price ignores less frequent but more profitable offers. In addition, not all OTE for a given $FC$ reach the mean OTE profit. These factors will negatively influence on the mean $PL$ results. For one interested in details, Table \ref{TblOTEProfits} contains empirical profits and corresponding mass and cumulative frequencies for the curve $FC=\$49.99$ on Figure \ref{FigOTEProfitECDF}.

\begin{center}
\begin{longtable}{|r|r|r|r|r|r|}
\caption[ESZ17 OTE Sample Profits and Empirical Frequencies]{ESZ17 OTE EPMF, ECDF, $C=\$4.68$, $FC=\$49.99$.} \label{TblOTEProfits} \\
 \hline
 \multicolumn{1}{|c}{$i$} &
 \multicolumn{1}{|c}{Profit, \$} &
 \multicolumn{1}{|c}{$N_i$} &
 \multicolumn{1}{|c}{$\sum_{j=1}^{j=i} N_j$} &
 \multicolumn{1}{|c}{$\frac{N_i}{\sum_{j=1}^{j=66} N_j}$} &
 \multicolumn{1}{|c|}{$\frac{\sum_{j=1}^{j=i} N_j}{\sum_{j=1}^{j=66} N_j}$} \\
 \hline 
 \endfirsthead
 \multicolumn{6}{c}%
 {\tablename\ \thetable{} -- continued from previous page} \\
 \hline
 \multicolumn{1}{|c}{$i$} &
 \multicolumn{1}{|c}{Profit, \$} &
 \multicolumn{1}{|c}{$N_i$} &
 \multicolumn{1}{|c}{$\sum_{j=1}^{j=i} N_j$} &
 \multicolumn{1}{|c}{$\frac{N_i}{\sum_{j=1}^{j=66} N_j}$} &
 \multicolumn{1}{|c|}{$\frac{\sum_{j=1}^{j=i} N_j}{\sum_{j=1}^{j=66} N_j}$} \\
 \hline 
 \endhead
 \hline \multicolumn{6}{|r|}{{Continued on next page}} \\ \hline
 \endfoot
 \hline
 \endlastfoot
1 & 90.64 & 439 & 439 & 0.124 & 0.124\\
2 & 103.14 & 325 & 764 & 0.092 & 0.216\\
3 & 115.64 & 297 & 1061 & 0.084 & 0.300\\
4 & 128.14 & 263 & 1324 & 0.074 & 0.374\\
5 & 140.64 & 239 & 1563 & 0.067 & 0.441\\
6 & 153.14 & 229 & 1792 & 0.065 & 0.506\\
7 & 165.64 & 197 & 1989 & 0.056 & 0.562\\
8 & 178.14 & 147 & 2136 & 0.042 & 0.603\\
9 & 190.64 & 165 & 2301 & 0.047 & 0.650\\
10 & 203.14 & 134 & 2435 & 0.038 & 0.687\\
11 & 215.64 & 112 & 2547 & 0.032 & 0.719\\
12 & 228.14 & 98 & 2645 & 0.028 & 0.747\\
13 & 240.64 & 107 & 2752 & 0.030 & 0.777\\
14 & 253.14 & 67 & 2819 & 0.019 & 0.796\\
15 & 265.64 & 86 & 2905 & 0.024 & 0.820\\
16 & 278.14 & 61 & 2966 & 0.017 & 0.837\\
17 & 290.64 & 63 & 3029 & 0.018 & 0.855\\
18 & 303.14 & 53 & 3082 & 0.015 & 0.870\\
19 & 315.64 & 52 & 3134 & 0.015 & 0.885\\
20 & 328.14 & 36 & 3170 & 0.010 & 0.895\\
21 & 340.64 & 40 & 3210 & 0.011 & 0.906\\
22 & 353.14 & 35 & 3245 & 0.010 & 0.916\\
23 & 365.64 & 30 & 3275 & 0.008 & 0.925\\
24 & 378.14 & 26 & 3301 & 0.007 & 0.932\\
25 & 390.64 & 18 & 3319 & 0.005 & 0.937\\
26 & 403.14 & 27 & 3346 & 0.008 & 0.945\\
27 & 415.64 & 19 & 3365 & 0.005 & 0.950\\
28 & 428.14 & 16 & 3381 & 0.005 & 0.955\\
29 & 440.64 & 19 & 3400 & 0.005 & 0.960\\
30 & 453.14 & 10 & 3410 & 0.003 & 0.963\\
31 & 465.64 & 18 & 3428 & 0.005 & 0.968\\
32 & 478.14 & 9 & 3437 & 0.003 & 0.970\\
33 & 490.64 & 11 & 3448 & 0.003 & 0.973\\
34 & 503.14 & 9 & 3457 & 0.003 & 0.976\\
35 & 515.64 & 10 & 3467 & 0.003 & 0.979\\
36 & 528.14 & 6 & 3473 & 0.002 & 0.981\\
37 & 540.64 & 9 & 3482 & 0.003 & 0.983\\
38 & 553.14 & 7 & 3489 & 0.002 & 0.985\\
39 & 565.64 & 3 & 3492 & 0.001 & 0.986\\
40 & 578.14 & 4 & 3496 & 0.001 & 0.987\\
41 & 590.64 & 3 & 3499 & 0.001 & 0.988\\
42 & 603.14 & 1 & 3500 & 0.000 & 0.988\\
43 & 615.64 & 3 & 3503 & 0.001 & 0.989\\
44 & 628.14 & 1 & 3504 & 0.000 & 0.989\\
45 & 640.64 & 3 & 3507 & 0.001 & 0.990\\
46 & 653.14 & 5 & 3512 & 0.001 & 0.992\\
47 & 665.64 & 3 & 3515 & 0.001 & 0.992\\
48 & 678.14 & 3 & 3518 & 0.001 & 0.993\\
49 & 690.64 & 1 & 3519 & 0.000 & 0.994\\
50 & 715.64 & 2 & 3521 & 0.001 & 0.994\\
51 & 728.14 & 2 & 3523 & 0.001 & 0.995\\
52 & 740.64 & 3 & 3526 & 0.001 & 0.995\\
53 & 765.64 & 2 & 3528 & 0.001 & 0.996\\
54 & 778.14 & 1 & 3529 & 0.000 & 0.996\\
55 & 803.14 & 1 & 3530 & 0.000 & 0.997\\
56 & 815.64 & 1 & 3531 & 0.000 & 0.997\\
57 & 840.64 & 1 & 3532 & 0.000 & 0.997\\
58 & 853.14 & 2 & 3534 & 0.001 & 0.998\\
59 & 878.14 & 1 & 3535 & 0.000 & 0.998\\
60 & 903.14 & 1 & 3536 & 0.000 & 0.998\\
61 & 965.64 & 1 & 3537 & 0.000 & 0.999\\
62 & 1028.14 & 1 & 3538 & 0.000 & 0.999\\
63 & 1140.64 & 1 & 3539 & 0.000 & 0.999\\
64 & 1315.64 & 1 & 3540 & 0.000 & 0.999\\
65 & 1340.64 & 1 & 3541 & 0.000 & 1.000\\
66 & 1665.64 & 1 & 3542 & 0.000 & 1.000\\
\end{longtable}
\end{center}

We say an "empirical probability mass function", EPMF, of the OTE profits: with fixed $C < FC$ the profits are discrete due to discrete prices $P_i=N_i\delta$. "Permitted" profits are \cite[p. 93 Formulas 55]{salov2013}: $PL_{min}^{OTE}=k\delta (\lfloor \frac{2FC}{k\delta}\rfloor +1)-2C$; $PL_{i}^{OTE}=PL_{min}^{OTE} + k\delta i$, $i=0, 1,\dots$. Indeed, Equation \ref{EqPL} is for $\boldsymbol{U}$ with $W_0=0$ but not only $\boldsymbol{U} \in \mathfrak{U}$ with $W_n=0$. If $W_n \ne 0$, then the position is marked to the market using $P_n$. This allows to rewrite Equation \ref{EqPL} for growing $j$
\begin{equation}
\label{EqPLDiscrete}
PL_j=k\delta\sum_{i=1}^{i=j}U_i(N_j - N_i)-\sum_{i=1}^{i=j}C_i|U_i|-C_n|\sum_{i=1}^{i=j}U_i|, \; j=1,\dots,n.
\end{equation}
For $n=j=i=1$, it returns $-2C_1|U_1|$ and for $U_1\ne 0$ the fees are paid at taking the position and due to marking-to-market.
For an MPS0 OTE $[t_s,t_e]$, $U_i=0$ for $s < i < e$, $|U_{i=s}|=|U_{i=e}|=W$, $U_s=-U_e$ and
\begin{equation}
\label{EqPLOTE}
PL_{s,e}^{OTE}=k\delta W |N_e - N_s|-W(C_s + C_e).
\end{equation}
With $C_s=C_e=C<FC$, $W=1$, and $|N_e-N_s|\ge\lfloor \frac{2FC}{k\delta} \rfloor + 1$, we always can express $|N_e-N_s|=\lfloor \frac{2FC}{k\delta} \rfloor + 1 + i$ and come to the discrete $PL_i^{OTE}$.

\begin{center}
\begin{longtable}{|r|r|r|r|r|r|r|r|r|r|}
\caption[Guessed Counts]{ESZ17 BOTE and SOTE Profit Sample Statistics, $C=\$4.68$, $FC=\$49.99$.} \label{TblBSOTEProfitStatistics} \\
 \hline
 \multicolumn{1}{|c}{Type} &
 \multicolumn{1}{|c}{$N$} &
 \multicolumn{1}{|c}{Mean} &
 \multicolumn{1}{|c}{Min} &
 \multicolumn{1}{|c}{$N_{Min}$} &
 \multicolumn{1}{|c}{Max} &
 \multicolumn{1}{|c}{$N_{Max}$} &
 \multicolumn{1}{|c}{StDev} &
 \multicolumn{1}{|c}{Skew.} &
 \multicolumn{1}{|c|}{E-Kurt.} \\
 \hline 
 \endfirsthead
 \multicolumn{10}{c}%
 {\tablename\ \thetable{} -- continued from previous page} \\
 \hline
 \multicolumn{1}{|c}{Type} &
 \multicolumn{1}{|c}{$N$} &
 \multicolumn{1}{|c}{Mean} &
 \multicolumn{1}{|c}{Min} &
 \multicolumn{1}{|c}{$N_{Min}$} &
 \multicolumn{1}{|c}{Max} &
 \multicolumn{1}{|c}{$N_{Max}$} &
 \multicolumn{1}{|c}{StDev} &
 \multicolumn{1}{|c}{Skew.} &
 \multicolumn{1}{|c|}{E-Kurt.} \\
 \hline 
 \endhead
 \hline \multicolumn{10}{|r|}{{Continued on next page}} \\ \hline
 \endfoot
 \hline
 \endlastfoot
BOTE & 1786 & 199.3 & 90.64 & 221 & 1340.64 & 1 & 120.5 & 2.52 & 11.1\\
SOTE & 1756 & 190.68 & 90.64 & 218 & 1665.64 & 1 & 119.2 & 3.26 & 22.1\\
BOTH & 3542 & 195.03 & 90.64 & 439 & 1665.64 & 1 & 119.9 & 2.88 & 16.3\\
\end{longtable}
\end{center}

Figure \ref{FigBSOTEProfitECDF} illustrates BOTE and SOTE ECDFs. Deviations are small. Both are close to ECDF on Figure \ref{FigOTEProfitECDF} for $FC=\$49.99$. See statistics in Table \ref{TblBSOTEProfitStatistics}.
\begin{figure}[!h]
  \centering
  \includegraphics[width=125mm]{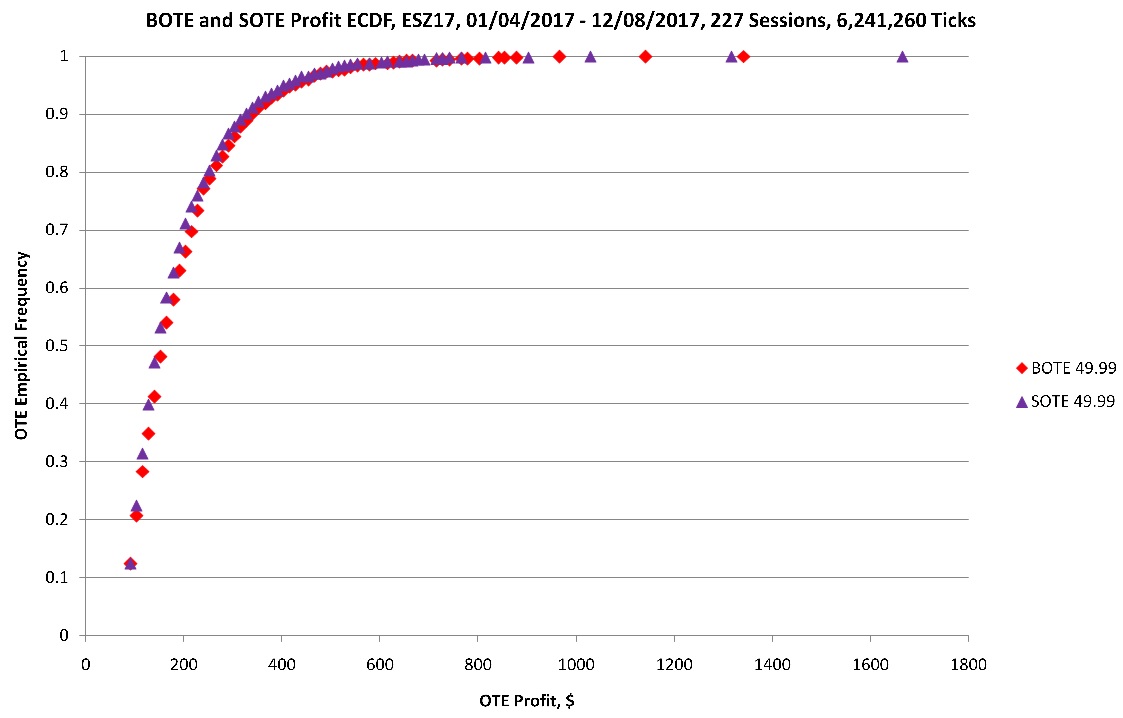}
  \caption[FigBSOTEProfitECDF]
   {Time \& Sales Globex, \url{http://www.cmegroup.com/}, ESZ17, Wednesday January 4, 2017 - Friday December 8, 2017, 227 trading sessions with the time range 17:00:00 (previous day) - 15:15:00 (closing day), 6,241,260 transaction ticks, $C = \$4.68$. MPS0 BOTE and SOTE ECDF for $FC = \$49.99$. Plotted using custom C++, Python, AWK programs, and Microsoft Excel.}
  \label{FigBSOTEProfitECDF}
\end{figure}
The author notices that after identical scaling of Figure \ref{FigOTEProfitECDF} and Figure 35 from \cite[p. 95]{salov2013}, the corresponding ECDF curves almost coincide. \textit{This indicates that empirical statistics of the OTE profits observed for ESZ17 in 227 sessions in 2017 and ESZ13 in 184 sessions in 2013 are close.}

\section{Why do speculative markets exist?}

A mortgage-backed security value may depend on dynamics of 360 monthly interest rates \cite[pp. 199 - 209]{davidson1994}, \cite[pp. 705 - 709, Exhibit 1]{fabozzi1995}, while the majority of them is well correlated \cite[pp. 93 - 107, 3.2 Principal Component Analysis]{golub2000}, PCA. Theoretical throwing between an increasing number of random factors in a model and further attempts to reduce the dimension of its space by means of PCA resemble incessant fluctuations of prices laughing at the idea of market equilibrium. Significant and repeated price fluctuations define market opportunities. MPS is their objective measure \cite[pp. 6 - 7]{salov2013}.

\textit{One may like, hate, or ignore speculation but whether the price moves are random, chaotic, trendy, or not, substantial and recurrent market offers, which can be objectively accounted for, applying MPS, are an essential condition of the existence of speculative markets and the thriving interest in them. Emphasizing this quantitative role of MPS, the author understands that we may talk about an essential but insufficient condition of existence and that free markets and relation to them assume a certain economic, political, and cultural society basis.}

MPS and OTE express market states in terms of positions, actions, profits understood by both traders and academicians. MPS explains trader's aspiration. The "universe" of strategies is measured by the "galactic" number $(2W+1)^{n-1}$ routinely reaching $3^{134908}$. \textit{Importance of MPS for the economy is in combining transactions costs, prices, and actions for measuring market offers up to this day.} For instance, it would be interesting to see more evidences, if the ECDFs of OTE profits, expressed in one currency, Figures \ref{FigOTEProfitECDF}, \ref{FigBSOTEProfitECDF}, persist in years.

Steven Strahler cites \cite{strahler2003} Leo Melamed, the legendary founder of the financial futures and Globex, about exchanges and electronic technology behind: \textit{"We have not yet gone into the galaxy, but we're thinking about it"}. After 11 years, Melamed, congratulating the Futures Magazine with the 500th issue, confirms \cite{melamed2014}: \textit{"... the future of futures markets is limited only by our own imagination"}. The \textit{scale} of the MPS is suitable for this journey.

\paragraph{Acknowledgments.}  I am grateful to Oscar Sheynin for sharing \cite{sheyninUnpublished} and valuable opinions on the history of probability theory and statistics.

\bigskip

\noindent\textbf{Valerii Salov} received his M.S. from the Moscow State University, Department of Chemistry in 1982 and his Ph.D. from the Academy of Sciences of the USSR, Vernadski Institute of Geochemistry and Analytical Chemistry in 1987.  He is the author of the articles on analytical, computational, and physical chemistry, the book Modeling Maximum Trading Profits with C++, \textit{John Wiley and Sons, Inc., Hoboken, New Jersey}, 2007, and papers in \textit{Futures Magazine} and \textit{ArXiv}.

\noindent\textit{v7f5a7@comcast.net}


\begin{thebibliography}{9}
\bibitem{andrews2016} Andrews, Evan. 6 Longstanding Debts from History. \textit{History.com}, Publisher A+E Networks, December 2, 2016, \url{http://www.history.com/news/history-lists/6-longstanding-debts-from-history}.
\bibitem{arnold2004} Arnold, Vladimir, I. A.N. Kolmogorov and Natural Science. \textit{Uspehi Matematicheskih Nauk}, Volume 59, No. 1, January - February 2004, pp. 25 - 44 (in Russian); and Russ. Math. Surv., 2004, Volume 59, pp. 27 - 46.
\bibitem{atkins2015} Atkins, Ralph, Hale, Thomas. Companies issue record levels of perpetual debt. \textit{The Financial Times}, June 15, 2015, \url{https://www.ft.com/content/e2734340-111f-11e5-a8b1-00144feabdc0}.
\bibitem{bachelier1900} Bachelier, Louis. Th\'eorie de la Sp\'eculation. \textit{Annales scientifiques de l'Ecole Normale Sup\'erieure}, 3e s\'erie, tome 17, 1900, pp. 21 - 86 (in French). English translation and facsimile of original Bachelier's thesis can be found in \cite{davis2006}.
\bibitem{bauneau2014} Baubeau, Patrice: War Finance (France), in: 1914-1918-online. \textit{International Encyclopedia of the First World War}, ed. by Ute Daniel, Peter Gatrell, Oliver Janz, Heather Jones, Jennifer Keene, Alan Kramer, and Bill Nasson, issued by Freie Universität Berlin, Berlin 2014-10-08. DOI: 10.15463/ie1418.10022.
\bibitem{baxter1996} Baxter, Martin, Rennie, Andrew. \textit{Financial Calculus. An introduction to derivative pricing}. Cambridge: Cambridge University Press, 1996.
\bibitem{beesack1962} Beesack, Paul, R. On the Rank of a Matrix. \textit{Mathematics Magazine}, Volume 35, No. 2, March 1962, pp. 73 - 77.
\bibitem{belousov1967} Belousov, Valentin, D. \textit{Foundations of quasigroups and loops}. Moscow: Nauka, 1967 (in Russian).
\bibitem{bernstein1927} Bernstein, Sergei, N. Sur l’extension du théorème limite du calcul des probabilités aux sommes de quantités dépendantes. \textit{Mathematische Annalen}, Volume 97, Issue 1, 1927, pp. 1 - 59. Citation is found in \cite{sheyninUnpublished} and reviewed (in French). See also Bernsetein, Sergei, N. Extending the limit theorem of the theory of probability on sums of dependent variables. \textit{Uspehi Matematicheskih Nauk}, No. 10, 1944, pp. 65 - 114 (in Russian)
\bibitem{black1973} Black, Fischer, Scholes, Myron. The Pricing of Options and Corporate Liabilities. \textit{The Journal of Political Economy}, Volume 81, May - June 1973, pp. 637 - 659.
\bibitem{bourbaki1974} Bourbaki, Nicolas. \textit{Elements of Mathematics. Algebra I. Chapters 1 - 3.} Paris: Hermann, Publishing in Arts and Science, Massachusetts: Addison-Wesley Publishing Company, 1974.
\bibitem{boyer1989} Boyer, Carl, B. \textit{A History of Mathematics.} Second Edition, Revised by Merzbach, Uta, C. Foreword by Isaac Asimov. New York: John Wiley \& Sons, Inc., 1989.
\bibitem{brad2017} Brad, Jones. D-Wave sold its first 2,000-qubit quantum annealer to a cybersecurity firm. Posted on January 24, 2017 4:46 pm, \url{https://www.digitaltrends.com/computing/d-wave-2000q-launch/}.
\bibitem{brock1992} Brock, William, Lakonishok, Josef, and LeBaron, Blake. Trading Rules and the Stochastic Properties of Stock Returns. \textit{The Journal of Finance}, Volume 47, No. 5, December 1992, pp. 1731 - 1764.
\bibitem{bruck1971} Bruck, Richard, H. \textit{A Survey of Binary Systems}. Third Edition. New York: Springer-Verlag Berlin Beideiberg, 1971.
\bibitem{cayley1854} Cayley, Arthur. On the Theory of Groups, as depending on the Symbolic Equation $\theta^n=1$. \textit{Philosophical Magazine and Journal of Science}, Volume VII - Fourth series, Number XLII. - January 1854, pp. 40 - 47.
\bibitem{cayley1889} Cayley, Arthur. On the Theory of Groups. \textit{American Journal of Mathematics}, Volume 11, No. 2, January 1889, pp. 139 - 157.
\bibitem{chandler1962} Chandler, Davis. The Norm of the Schur Product Operation. \textit{Numerishe Mathematik}, Volume 4, 1962, pp. 343 - 344.
\bibitem{chow1963} Chow, Yu., S., Robbins, Herbert. On Optimal Stopping Rules. \textit{Zeitschrift fur Wahrscheinlichkeitstheorie und Verwandte Gebiete}, Volume 2, No. 1, 1963, pp. 33 - 49.
\bibitem{cramer1962} Cramer, Harald. \textit{Mathematical Methods of Statistics.} First Indian Edition, Under agreement with the Princeton University Press. Bombay, New York: Asia Publishing House, 1962.
\bibitem{davidson1994} Davidson, Andrew, S., Herskovitz, Michael, D. \textit{Mortgage-Backed Securities. Investment Analysis \& Advanced valuation Techniques}. Chicago: Probus Publishing Company, 1994.
\bibitem{davis2006} Davis, Mark, Etheridge, Alison. \textit{Louis Bachelier's Theory of Speculation. The Origins of Modern Finance}. Princeton and Oxford: Princeton University Press, 2006.
\bibitem{deutsch1985} Deutsch, David. Quantum Theory, the Church-Turing Principle and Universal Quantum Computer. \textit{Proceedings of the Royal Society of London. Series A, Mathematical and Physical Sciences}, Volume 400, No. 1818, July 8, 1985, pp. 97 - 117.
\bibitem{downes1995} Downes, John, Goodman, Jordon, Elliot. \textit{Dictionary of Finance and Investment Terms}.  Fourth Edition. New York: Barron's, 1995.
\bibitem{duffie1993} Duffie, Darell, J., Harrison, Michael, J. Arbitrage Pricing of Russian Options and Perpetual Lookback Options. \textit{The Annals of Applied Probability}, Volume 3, No. 3, August 1993, pp. 641 - 651.
\bibitem{duffie2001} Duffie, Darrell. \textit{Dynamic Asset Pricing Theory}. Third Edition. Princeton and Oxford: Princeton University Press, 2001.
\bibitem{dynkin1956} Dynkin, Eugene, B., Yushkevich, Aleksandr, A. Strong Markov Processes. \textit{Teoriya Veroyatnostei i ee Primeneniya}, Volume 1, No. 1, 1956, pp. 149 - 155 (in Russian).
\bibitem{dynkin1959} Dynkin, Eugene, B. One-Dimensional Continuous Strong Markov Processes. \textit{Teoriya Veroyatnostei i ee Primeneniya}, Volume 4, No. 1, 1959, pp. 3 - 54 (in Russian).
\bibitem{dynkin1959b} Dynkin, Eugene, B. \textit{The Foundations of the Theory of Markov Processes}. Moscow: PHYSMATGIZ, 1959 (in Russian). \textit{Theory of Markov Processes}. Mineola, New York: Dover Publications, Inc., 2006. The Dover's edition is an unabridged republication of the work published by Prentice-Hall, Inc., Englewood Cliffs, New Jersey, 1961, translated from by Brown, D. E.
\bibitem{dynkin1963} Dynkin, Eugene, B. Optimal selection of a stopping time for a Markov process. \textit{Doklady Akademii Nauk SSSR}, Volume 150, 1963, pp. 238 - 240 (in Russian). The author has seen only the English translation in \cite{dynkin2000} "The optimum choice of the instant for stopping a Markov process", pp. 485 - 488.
\bibitem{dynkin1963b} Dynkin, Eugene, B. \textit{Markov Processes.} Moscow: PHYSMATGIZ, 1963 (in Russian).
\bibitem{dynkin1966} Dynkin, Eugene, B., Yushkevich, Aleksandr, A. \textit{Theory of Probability and Markov Processes}. Moscow: Mir, 1966 (in Russian).
\bibitem{dynkin1968} Dynkin, Eugene, B. Sufficient statistics for the optimal stopping problem. \textit{Teoriya Veroyatnostei i ee Primeneniya}, Volume 13, No. 1, 1968, pp. 150 - 152 (in Russian).
\bibitem{dynkin2000} Dynkin, Eugene, B. Editors: Yushkevich, Aleksandr, A., Seitz, G., M., Onishchik, A., L. \textit{Selected Papers of E. B. DYNKIN with Commentary}. American Mathematical Society, Providence, Rhode Island, International Press, Cambridge, Massachusetts, 2000.
\bibitem{einstein1905} Einstein, Albert. \"{U}ber die von der molekularkinetischen Theorie der W\"arme geforderte Bewegung von in ruhenden Fl\"ussigkeiten suspendierten Teilchen (On the Movement of Small Particles Suspended in a Stationary Liquid Demanded by the Molecular-Kinetic Theory of Heat), Annalen der Physik, Volume 322, No. 8, 1905, pp. 549 - 560 (in German). See English translation: \textit{Investigations on the Theory of the Brownian Motion by Albert Einstein}, Edited with notes by R F\"{u}rth, Translated by A.D. Cowper, New York: Dover Publications, Inc., 1956.
\bibitem{engle2000} Engle, Robert. The Econometrics of Ultra-High-Frequency Data. \textit{Econometrica, Vol. 68, N. 1}, 2000, pp. 1 - 22.
\bibitem{etherington1932} Etherington, Ivor, M., H. A simple method of finding sums of powers of the natural numbers. \textit{Edinburgh Mathematical Notes}, Volume 27, 1932, pp. xvi - xix.
\bibitem{etherington1949} Etherington, Ivor, M., H. Non-Associative Arithmetics. \textit{Proceedings of the Royal Society of Edinburgh Section A: Mathematics}, Volume 62, Issue 4, 1949, pp. 442 - 453.
\bibitem{fabozzi1995} Fabozzi, Frank, J. Editor. \textit{The Handbook of Mortgage Backed Securities}. Chicago: Probus Publishing, 1995.
\bibitem{feynman1982} Feynman, Richard. Simulating Physics with Computers. \textit{International Journal of Theoretical Physics}, Volume 21, No. 6/7, 1982, pp. 467 - 488.
\bibitem{fossbytes2017} FOSSBYTES. World’s Most Advanced Quantum Computer Created By A Team Of US And Russian Scientists. July 15, 2017, \url{https://fossbytes.com/most-advanced-quantum-computer-51-qubit/}.
\bibitem{giles2000} Giles, Ronald, L. Direction Indicators in Financial Modeling, pp. 171 - 180. In: Bonilla, Maria, Casasus, Trinidad, Sala, Ramon. Editors. \textit{Financial Modeling, Deutsche Bank Research.} New York: Springer-Verlag Berlin Heidelberg GmbH, 2000.
\bibitem{glasserman2003} Glasserman, Paul. \textit{Monte Carlo Methods in Financial Engineering}, New York: Springer, 2003.
\bibitem{gnedenko1949} Gnedenko, Boris, V., Kolmogorov, Andrey, N. \textit{Limit Distributions for Sums of Independent Random Variables}. Moscow, Leningrad: Technico-Theoretical Literature Governmental Press, 1949 (in Russian). Translated to English by K.L. Chung, Cambridge, Mass.: Addison-Wesley, 1954.
\bibitem{gnedenko1988} Gnedenko, Boris, V. \textit{The Probability Theory. [Kurs Teorii Veroyatnostei].} Moscow: Nauka, 1988 (in Russian).
\bibitem{goldberg1989} Goldberg, David, E. \textit{Genetic Algorithms in Search, Optimization \& Machine Learning}. New York: Addison-Wesley, 1989.
\bibitem{golub2000} Golub, Bennett, W., Tilman, Leo, M. \textit{Risk Management. Approaches for Fixed Income Markets}. New York: John Wiley \& Sons, Inc., 2000.
\bibitem{goodhart1997} Goodhart, Charles, A., E., O'Hara, Maureen. High frequency data in financial markets: Issues and applications. \textit{Journal of Empirical Finance}, Volume 4, 1997, pp. 73 - 114.
\bibitem{gumilevskii1965} Gumilevskii, Lev. \textit{Zinin}, Series the life of remarkable people. Issue 9 (404). Moscow: Molodaya Gvardiya, 1965 (in Russian).
\bibitem{halmos1987} Halmos, Paul, R. \textit{Finite-Dimensional Vector Spaces.} Reprint of the 2d edition published by Van Nostrand, Princeton, N.J. New York: Springer-Verlag New York Inc., 1987.
\bibitem{hanson2007} Hanson, Floyd, B. \textit{Applied Stochastic Processes and Control for Jump-Diffusions. Modeling, Analysis, and Computation}, Philadelphia: SIAM, 2007.
\bibitem{hofmann1880} Hofmann, August Wilhelm. Nekrolog auf N. Zinin (VS: These two pages have no title and the one given here only reflects the contents.) \textit{Berichte der Deutschen Chemischen Gesellschaft.} Sitzun vom. 8, Marz 1880, Volume 13, Issue 1, 1880, pp. 449 - 450 (In German).
\bibitem{horn1990} Horn, Roger, A. The Hadamard Product, \textit{Matrix Theory and Applications. Proceedings of Symposia in Applied Mathematics}, Editor Johnson, Charles, R., Volume 40, American mathematical Society, Providence, Rhode Island, 1990, pp. 87 - 169.
\bibitem{hull1997} Hull, John, C. \textit{Options, Futures, and Other Derivatives}. 3rd Ed. Upper Saddle River, NJ: Prentice Hall, 1997.
\bibitem{hunt2000} Hunt, Phil, J., Kennedy, Joanne, E. \textit{Financial Derivatives in Theory and Practice}. New York: John Wiley \& Sons, LTD, 2000.
\bibitem{iso2017} \textit{ISO/IEC JTC1 SC22 WG21 Working Draft, Standard for Programming Language C++}, N4659, March 21, 2017.
\bibitem{jacobsen1996} Jacobsen, Martin. Laplace and the origin of the Orsntein-Uhlenbeck process. \textit{Bernoulli}, Volume 2, No. 3, 1996, pp. 271 - 286.
\bibitem{jones1999} Jones, Ryan. \textit{The Trading Game. Playing by the Numbers to Make MILLIONS}. New York: John Wiley \& Sons, Inc, 1999.
\bibitem{karr1993} Karr, Alan, F. \textit{Probability}. New York: Springer-Verlag, 1993.
\bibitem{kaufman2005} Kaufman, Perry, J. \textit{New trading systems and methods}. 4th Ed. New York: John Wiley \& Sons, Inc., 2005.
\bibitem{kelly1956} Kelly, John, L., Jr. A New Interpretation of Information Rate. \textit{Bell System Technical Journal}, Volume 35, No. 4, July 1956, pp. 917 - 926.
\bibitem{kloeden1999} Kloeden, Peter, E., Platen, Eckhard. \textit{Numerical Solution of Stochastic Differential Equations}, Berlin: Springer, 1999.
\bibitem{kolmogorov1982} Kolmogorov, Andrey, N., Zhurbenko, Igor, G., Prokhorov, Alexander, V. \textit{Introduction to probability theory.} Moscow: Nauka, 1982 (in Russian).
\bibitem{kolmogorov1983} Kolmogorov, Andrey, N. Combinatorial foundations of information theory and the calculus of probabilities. \textit{Uspehi Matematicheskih Nauk}, Volume 38, No. 4, July - August 1983, pp. 27 - 36 (in Russian); \textit{Russian Mathematical Surveys}, Volume 38, No. 4, 1983, pp. 29 - 40. 
\bibitem{korn1968} Korn, G., Korn T., \textit{Mathematical Handbook for Scientists and Engineers. Definitions, Theorems, and Formulas for Reference and Review}. 2nd ed., New York: McGraw-Hill Book Company, 1968.
\bibitem{kushner2006} Kushner, Boris, A. The Constructive Mathematics of A. A. Markov. \textit{The American Mathematical Monthly}, Volume 113, No. 6, June - July 2006, pp. 559 - 566.
\bibitem{lando2007} Lando, Sergei, K. \textit{Lectures on Generating Functions}. Third Editions, Moscow: MCNMO, 2007.
\bibitem{lefevre1923} Lef\`{e}vre, Edwin. \textit{Reminiscences of a Stock Operator}. New York: John Wiley \& Sons, Inc., 1993. Copyright 1993, 1994 by Expert Trading, Ltd. Originally published  in 1923 by George H.Doran and Company.
\bibitem{lipton2001} Lipton, Alexander. \textit{Mathematical Methods for Foreign Exchange. A Financial Engineer's Approach}. New Jersey: World Scientific Publishing Co. Pte. Ltd., 2001.
\bibitem{livermore1940} Livermore, Jesse, L. \textit{How To Trade in Stocks. The Livermore Formula for Combining Time Element and Price}. New York: Duell, Sloan \& Pearce, 1940.
\bibitem{lukacs1970} Lukacs, Eugene. \textit{Characteristic Functions}. 2nd edition. London: Charles Griffin \& Company Limited, 1970.
\bibitem{malcev1973} Malcev, Anatoly, I. \textit{Algebraic Systems}. New York: Springer-Verlag, 1973.
\bibitem{malkiel2007} Malkiel, Burton. G. \textit{A random walk down Wall Street. The Time-Tested Strategy for Successful Investing}. 9th Ed. New York: W.W.Norton \& Company, 2007.
\bibitem{markov1951} Markov, Andrey, A. Extension of limit theorems for computing probabilities on the sum of variables linked into chain. Reported at the meeting of Physics-Mathematics department, December 5, 1907. Records of Academy of Sciences on Physics-Mathematics Department, VIII series, Volume 22, No. 9. In \textit{A.A. Markov. Selected Works. Theory of Numbers. Theory of Probability.} Moscow: Academy of Sciences Publisher, 1951, pp. 363 - 397 (in Russian).
\bibitem{melamed2014} Melamed, Leo. A historical flash back with Leo Melamed. \textit{Futures Magazine}, 500th issue, August 24, 2014.
\bibitem{merton1973} Merton, Robert C. Theory of rational option pricing. \textit{Bell Journal of Economics and Management Science}, 1973, Volume 4, No 1, pp. 141 - 183.
\bibitem{milstein1995} Milstein, Grigory, N. \textit{Numerical Integration of Stochastic Differential Equations}. Dordrecht: Kluwer Academic Publishers, 1995.
\bibitem{moore2014} Moore, Elaine. UK to repay tranche of perpetual war loans. \textit{Financial Times}, October 31, 2014, \url{https://www.ft.com/content/94653f60-60e8-11e4-894b-00144feabdc0}.
\bibitem{murphy1999} Murphy, John, J. \textit{Technical Analysis of the Financial Markets. A Comprehensive Guide to Trading Methods and Applications}. New York: New Work Institute of Finance, 1999.
\bibitem{neftci1991} Neftci, Salih N. Naive Trading Rules in Financial Markets and Wiener-Kolmogorov Prediction Theory: A Study of "Technical Analysis". \textit{Journal of Business}, Volume 64, No. 4, October 1991, pp. 549 - 571.
\bibitem{neftci1996} Neftci, Salih N. \textit{An Introduction to the Mathematics of Financial Derivatives}. San Diego: Academic Press, 1996.
\bibitem{nesterov2006} Nesterov, Alexander, I. Gravity within the Framework of Nonassociative Geometry. Chapter 24, pp. 299 - 311. In Sabinin, Lev, V., Sbitneva, Larissa, Shestakov, Ivan. \textit{Non-Associative Algebra and Its Applications}. New York: Chapman \& Hall/CRC Taylor \& Francis Group, 2006.
\bibitem{pardo1992} Pardo, Robert. \textit{Design, Testing, and Optimization of Trading Systems}. New York: John Wiley \& Sons, Inc., 1992.
\bibitem{pelc2011} Pelc, Anderzej. Why Do We Believe Theorems? pp. 358 - 369 in \textit{The BEST WRITING on MATHEMATICS 2010}, Pitici, Mircea, Editor, Princeton: Princeton University Press, 2011.
\bibitem{protter2004} Protter, Philip, E. \textit{Stochastic Integration and Differential Equations}. Second Edition. Berlin: Springer, 2004.
\bibitem{reams1999} Reams, Robert. Hadamard inverses, square roots and products of almost semidefinite matrices. \textit{Linear Algebra and its Applications}, Volume 288, 1999, pp. 35 - 43.
\bibitem{rehfeld2011} Rehfeld, K., Marwan, N., Heitzig, J., Kurths, J. Comparison of correlation analysis techniques for irregularly sampled time series. \textit{Nonlinear Processes in Geophysics}, Volume 18, 2011, pp. 389 - 404.
\bibitem{reynolds2017} Reynolds, Matt. Quantum simulator with 51 qubits is largest ever. \textit{New Scientist, Daily News}, July 18, 2017, \url{https://www.newscientist.com/article/2141105-quantum-simulator-with-51-qubits-is-largest-ever/}.
\bibitem{rogers2000} Rogers, L. Chris G., Williams, David. \textit{Diffusions, Markov Processes, and Martingales. Volume 1: Foundations}, 2nd Edition, United Kingdom, Cambridge: Cambridge University Press, 2000.
\bibitem{rosenfeld1968} Rosenfeld, Azriel. \textit{An Introduction to Algebraic Structures}. San Francisco: Holden-Day, 1968.
\bibitem{sabinin1999} Sabinin, Lev, V. \textit{Smooth Quasigroups and Loops}. Netherlands: Kluwer Academic Publishers, Springer-Science+Business Media Dordrecht, 1999.
\bibitem{sabinin2000} Sabinin, Lev, V. Smooth quasigroups and loops: forty-five years of incredible growth. \textit{Commentationes Mathematicae Universitatis Carolinae}, Volume 41, No. 2, 2000, pp. 377 - 400.
\bibitem{salov2007} Salov, Valerii, V. \textit{Modeling Maximum Trading Profits with C++: New Trading and Money Management Concepts}. Hoboken, NJ: John Wiley \& Sons, Inc., 2007.
\bibitem{salov2008} Salov, Valerii, V. Idealized models for real profits. \textit{Futures Magazine}, Volume XXXVII, No. 5, May 2008, pp. 36 - 39.
\bibitem{salov2011} Salov, Valerii, V. Market Profile and the distribution of price. \textit{Futures Magazine}, Vol. XL, No. 6, June 2011, pp. 34 - 36.
\bibitem{salov2011b} Salov, Valerii, V. Trading system analysis: Learning from perfection. \textit{Futures Magazine}, Vol. XL, No. 11, November, 2011, pp. 34 - 39, 43.
\bibitem{salov2012} Salov, Valerii, V. High-frequency trading in live cattle futures. \textit{Futures Magazine}, Vol. XL, No. 6, May 2012, pp. 26 - 27, 31.
\bibitem{salov2012b} Salov, Valerii, V. Notation for Iteration of Functions, Iteral. \textit{Arxiv, Mathematics, Dynamical Systems}, June 30, 2012, pp. 1 - 23, available at \url{http://arxiv.org/abs/1207.0152}
\bibitem{salov2012c} Salov, Valerii, V. Inevitable Dottie Number. Iterals of cosine and sine. \textit{Arxiv, Quantitative Finance, General Finance}, December 1, 2012, pp. 1 - 17, available at \url{http://arxiv.org/abs/1212.1027}
\bibitem{salov2013} Salov, Valerii, V. Optimal Trading Strategies as Measures of Market Disequilibrium. \textit{Arxiv, Quantitative Finance, General Finance}, December 6, 2013, pp. 1 - 222, \url{http://arxiv.org/abs/1312.2004}.
\bibitem{salov2014} Salov, Valerii. "The Gibbon of Math History". Who Invented the St. Petersburg Paradox? Khinchin's resolution. \textit{Arxiv, Mathematics, History and Overview}, March 11, 2014, pp. 1 - 17, \url{http://arxiv-web3.library.cornell.edu/abs/1403.3001}.
\bibitem{salov2015} Salov, Valerii. The Role of Time in Making Risky Decisions and the Function of Choice. \textit{Arxiv, Quantitative Finance, General Finance}, December 27, 2015, pp. 1 - 52, \url{http://arxiv.org/abs/1512.08792}.
\bibitem{salov2017} Salov, Valerii. The Wandering of Corn. \textit{Arxiv, Quantitative Finance, General Finance}, April 3, 2017, pp. 1 - 65, \url{https://arxiv.org/abs/1704.01179?context=q-fin}.
\bibitem{schafer1961} Schafer, R., D. \textit{An Introduction to Nonassociative Algebras}. Stillwater, Oklahoma: Department of Mathematics at Oklahoma State University, 1961.
\bibitem{shepp1993} Shepp, Larry, Shiryaev, Albert, N. The Russian Option: Reduced Regret. \textit{The Annals of Applied Probability}, Volume 3, Number 3, pp. 631 - 640.
\bibitem{sheynin1989} Sheynin, Oscar, B. A. A. Markov's Work on Probability. \textit{Archive for History of Exact Sciences}, Volume 39, No. 4, 1989, pp. 337 - 377.
\bibitem{sheynin2007} Sheynin, Oscar, B. Markov: Integrity Is Just As Important As Scientific Merits. \textit{NTM International Journal of History \& Ethics of Natural Sciences, Technology \& Medicine}, Volume 15, No. 4, October 2007, pp. 289 - 294.
\bibitem{sheyninUnpublished} Sheynin, Oscar, B. \textit{Theory of Probability. An Elementary Treatise against a Historical Background}. Unpublished manuscript the English and Russian versions of which were provided by O.B.S. in a private email.
\bibitem{shiryaev1963} Shiryaev, Albert, N. On Optimum Methods in Quickest Detection Problems. \textit{Teoriya Veroyatnostei i ee Primeneniya}, Volume 8, No. 1, 1963, pp. 26 - 51 (in Russian).
\bibitem{shiryaev1984} Shiryaev, Albert, N. \textit{Probability.} Translated by Boas R. P. from the original Russian edition: Veroyatnost', Moscow: Nauka, 1979. New York: Springer Science + Business Media, LLC, 1984.
\bibitem{shiryaev1996} Shiryaev, Albert, N. \textit{Probability.} Second Edition. Translated by Boas R. P. New York: Springer Science + Business Media, LLC, 1996.
\bibitem{shor1994} Shor, Peter, W. Algorithms for Quantum Computation: Discrete Logarithms and Factoring. \textit{In: Proceedings, 35th Annual Symposium on Foundations of Computer Science, Santa Fe, NM, November 20-22, 1994, IEEE Computer Society Press}, pp. 124 - 134.
\bibitem{shor2000} Shor, Peter, W. Quantum Information Theory: Results and Open Problems. \textit{In Geom. Funct. Anal. (GAFA), Special Volume - GAFA2000}, 2000, pp. 816 - 838.
\bibitem{stanley1990} Stanley, Richard, P. \textit{Enumerative Combinatorics}, Translation from English, Moscow: Mir, 1990 (in Russian).
\bibitem{stanton2017} Stanton, Daniel. Asian perpetual bonds break record, with six months to spare. \textit{Reuters}, June 11, 2017 / 9:38 PM, \url{http://www.reuters.com/article/asia-debt-bonds-idUSL3N1J91I8}.
\bibitem{strahler2003} Strahler, Steven, R. Trading Places. CRAIN'S Chicago Business, June 2, 2003.
\bibitem{stroustrup2013} Stroustrup, Bjarne, \textit{The C++ Programming Language}. Fourth Edition, New York: Addison-Wesley, 2013.
\bibitem{taylor1968} Taylor, Howard, M. Optimal Stopping in a Markov Process. \textit{The Annals of Mathematical Statistics}, Volume 39, No. 4, August 1968, pp. 1333 - 1344.
\bibitem{taylor1965} Taylor, Howard, M. Markovian Sequential Replacement Processes. \textit{The Annals of Mathematical Statistics}, Volume 36, No. 6, December 1965, pp. 1677 - 1694.
\bibitem{uhlenbeck1930} Uhlenbeck, George, E., Ornstein, Leonard, S. On the Theory of the Brownian Motion. \textit{Physical Review}, Volume 36, No. 5, September 1930, pp. 823 - 841.
\bibitem{vince1992} Vince, Ralph. \textit{The Mathematics of Money Management. Risk Analysis Techniques for Traders}.  New York: John Wiley \& Sons, Inc., 1992.
\bibitem{vince1995} Vince, Ralph. \textit{The New Money Management. A Framework for Asset Allocation}. New York: John Wiley \& Sons, Inc., 1995.
\bibitem{voevodin1980} Voevodin, Valentin, V. \textit{Linear Algebra}. Moscow: Nauka, Main Publisher of Physiko-Matematicheskoi Literatury, 1980 (in Russian).
\end{thebibliography}
\end{document}